\documentclass[11pt,letterpaper]{article}

\usepackage[utf8]{inputenc}
\usepackage[T1]{fontenc}
\usepackage{microtype}

\usepackage[margin=1in]{geometry}

\usepackage[small]{caption}
\usepackage{subcaption}

\usepackage{amsmath}
\usepackage{amsfonts}
\usepackage{amssymb}
\usepackage{amsthm}
\usepackage{amstext}
\usepackage{thmtools}

\usepackage{csquotes}
\usepackage{xcolor}
\usepackage{graphicx}

\usepackage{mathtools}
\usepackage{xspace}
\usepackage{enumerate}
\usepackage[ruled]{algorithm2e}
\usepackage{paralist}
\usepackage{bm}
\usepackage{dsfont}
\usepackage{thm-restate}
\usepackage{enumerate}
\usepackage{paralist}
\usepackage{comment}
\usepackage{soul}

\definecolor{job1}{HTML}{6D9DC5}
\definecolor{job2}{HTML}{E0CC74}
\definecolor{job3}{HTML}{649884}
\definecolor{job4}{HTML}{C9809C}
\definecolor{job5}{HTML}{F39B6D}
\definecolor{job6}{HTML}{547781}

\usepackage{tikz}
\usetikzlibrary{arrows}
\usetikzlibrary{fit}
\usetikzlibrary{calc,shapes.geometric, backgrounds}
\usetikzlibrary{decorations.pathreplacing,calligraphy}
\usetikzlibrary{patterns,snakes}
\tikzset{vertex/.style={circle, draw, fill=black, black, inner sep=0pt, minimum width=7pt},
}

\usepackage{url}
\usepackage[colorlinks]{hyperref}
\hypersetup{citecolor=green!60!black,linkcolor=blue!80!black,urlcolor=blue!80!black}
\usepackage[capitalize,noabbrev]{cleveref}

\usepackage[
 backend=biber, 
 style=alphabetic, 
 citetracker, 
 hyperref=auto,
 maxcitenames=5, %
 sortcites,      %
 sorting=nyt,
 maxbibnames=12, %
 date=year,      %
 isbn=false,     %
 url=false,      %
 doi=false,      %
 eprint=false,   %
]{biblatex}
\newbibmacro{string+doiurlisbn}[1]{%
  \iffieldundef{doi}{%
    \iffieldundef{url}{%
      #1
    }{%
      \href{\thefield{url}}{#1}%
    }%
  }{%
    \href{http://dx.doi.org/\thefield{doi}}{#1}%
  }%
}

\DeclareFieldFormat%
[article,inbook,incollection,inproceedings,patent,thesis,unpublished]
  {title}{\usebibmacro{string+doiurlisbn}{#1}}

\theoremstyle{plain}
\newtheorem{theorem}{Theorem}[section]
\newtheorem{proposition}[theorem]{Proposition}
\newtheorem{observation}[theorem]{Observation}
\newtheorem{lemma}[theorem]{Lemma}
\newtheorem{corollary}[theorem]{Corollary}
\theoremstyle{definition}
\newtheorem{definition}[theorem]{Definition}

\theoremstyle{remark}

\DeclarePairedDelimiter\floor{\lfloor}{\rfloor}

\newcommand{\cO}{\mathcal{O}}

\newcommand{\cS}{\ensuremath{\mathcal{S}}\xspace}

\newcommand{\EX}{\mathbb{E}}
\newcommand{\VAR}{\mathbb{V}}
\newcommand{\ind}{\ensuremath{\mathds{1}}}

\newcommand{\borrow}{\leftarrowtail}

\newcommand{\bP}{\bm\bar{P}}

\newcommand{\hy}{\bm\hat{y}}
\newcommand{\bL}{\bm\bar{L}}
\newcommand{\bO}{\bm\bar{O}}
\newcommand{\hO}{\bm\hat{O}}
\newcommand{\bS}{\bm\bar{S}}
\newcommand{\by}{\bm\bar{y}}

\newcommand{\pr}{\mathrm{Pr}}

\newcommand\job{} %
\def\job[#1](#2)(#3){%
  \fill[#1,draw=black,line width=0.6] (#2) rectangle ++(#3);
}

\newcommand\jobT{} %
\def\jobT[#1](#2,#3)(#4,#5)(#6){%
  \fill[#1,draw=black,line width=0.6] (#2,#3) rectangle ++(#4,#5);
  \node[text=black] at (#2+#4/2.0,#3+#5/2.0) {#6};
}

\usepackage[colorinlistoftodos,prependcaption,textsize=scriptsize]{todonotes}

\newcommand{\red}[1]{{\color{red}#1}}

\title{Delayed-Clairvoyant Flow Time Scheduling \\ via a Borrow Graph Analysis\thanks{\red{This is an independent and precedent work to our joint paper with Anupam Gupta, Haim Kaplan and Sorrachai Yingchareonthawornchai~\cite{GuptaKLSY25}. In this paper, we study the same problem and the same algorithm as in \cite{GuptaKLSY25}, but with a completely different analysis, which gives a worse bound on the competitive ratio compared to \cite{GuptaKLSY25}. We uploaded this version to arXiv for the sake of completeness. This draft has been written in early 2025, and only covers related work up to that point in time.}}}

\author{Alexander Lindermayr\thanks{Institut für Mathematik, Technische Universität Berlin, Germany. \texttt{alexander.lindermayr@tu-berlin.de}. This paper was written while the author was affiliated with the University of Bremen, Germany.} \and Jens Schlöter\thanks{Centrum Wiskunde \& Informatica, The Netherlands. \texttt{jens.schloter@cwi.nl}}}

\date{January 2025}

\begin{document}

\maketitle

\begin{abstract}
  We study the problem of preemptively scheduling jobs online over time on a single machine to minimize the total flow time.
  In the traditional \emph{clairvoyant} scheduling model, the scheduler learns about the processing time of a job at its arrival, and scheduling at any time the job with the shortest remaining processing time (SRPT) is optimal. In contrast, the practically relevant \emph{non-clairvoyant} model assumes that the processing time of a job is unknown at its arrival, and is only revealed when it completes. Non-clairvoyant flow time minimization does not admit algorithms with a constant competitive ratio. Consequently, the problem has been studied under speed augmentation (JACM'00) or with predicted processing times (STOC'21, SODA'22) to attain constant guarantees.

  In this paper, we consider \emph{$\alpha$-clairvoyant} scheduling, where the scheduler learns the processing time of a job once it completes an $\alpha$-fraction of its processing time. This naturally interpolates between clairvoyant scheduling ($\alpha=0$) and  non-clairvoyant scheduling ($\alpha=1$). By elegantly fusing two traditional algorithms, we propose a scheduling rule with a competitive ratio of $\mathcal{O}(\frac{1}{1-\alpha})$ whenever $0 \leq \alpha < 1$. As $\alpha$ increases, our competitive guarantee transitions nicely (up to constants) between the previously established bounds for clairvoyant and non-clairvoyant flow time minimization. We complement this positive result with a tight randomized lower bound.
\end{abstract}

\thispagestyle{empty}

\newpage

\tableofcontents

\newpage

\section{Introduction}

We study the following fundamental online scheduling problem: Jobs arrive online over time without prior knowledge, and an algorithm can at any time process a job on a single machine, but cannot change its past decisions. Preemption is allowed, that is, an algorithm can stop processing a job and resume it at a later point in time.
The goal is to minimize the sum of flow times, where the \emph{flow time} of a job is the time between its release time and completion time (also known as response time or pending time).
In the \emph{clairvoyant} setting, where the processing time of a job becomes known at its arrival, a classic result of \textcite{Schrage68} says that scheduling at any time the job with the least amount of remaining work is optimal (called Shortest-Remaining-Processing-Time rule, in short SRPT). In particular, SRPT is a 1-competitive online algorithm for this problem; we say that an algorithm is $\alpha$-\emph{competitive} if its schedule has a total flow time of at most $\alpha$ times the total flow time of an offline optimum solution. The \emph{competitive ratio} of an algorithm is the smallest~$\alpha$ such that it is $\alpha$-competitive.

However, practically relevant in many applications, such as scheduling tasks in computing systems, is the \emph{non-clairvoyant} scenario, where an algorithm has no knowledge about a job's processing time until it completes.
In contrast to the well manageable clairvoyant variant, \textcite{MotwaniPT94} showed a strong lower bound of $\Omega(n^{1/3})$ on the competitive ratio of every deterministic non-clairvoyant algorithm, where $n$ is the total number of jobs.

Yet, due to its practical relevance, several approaches have been considered to overcome this strong negative result:
\begin{description}
  \itemsep0pt
  \item[Randomization:] If the algorithm is allowed to use randomization, a $\cO(\log n)$-competitive algorithm is known~\cite{BecchettiL04,KalyanasundaramP03}, which is best-possible~\cite{MotwaniPT94}.
  \item[Speed Augmentation:] If the algorithm's machine runs at speed $(1+\varepsilon)$ times the speed of the optimum's machine, then a $\cO(\frac{1}{\varepsilon})$-competitive algorithm is known, for any $\varepsilon > 0$~\cite{KalyanasundaramP00}.
  \item[Predictions:] If the algorithm receives at a job's release a prediction on its processing time such that the ratio between the predicted and actual processing time (in any direction) is bounded by some constant $\mu \geq 1$, then a $\cO(\mu^2 \log \mu)$-competitive algorithm is known~\cite{AzarLT21,AzarLT22}.
  Specifically, there are $\cO(1)$-competitive algorithms for the semi-clairvoyant model, where the power-of-2 class of a job's processing time is known~\cite{BecchettiLMP04}.
\end{description}

We follow the spirit of these approaches and consider a minimal relaxation of the non-clairvoyant model for which we give a $\cO(1)$-competitive algorithm.

The model is based on common techniques used in practice, e.g., operating systems or cloud computing: there, job lengths are commonly estimated by \emph{profiling} techniques~\cite{hilman2018task,weerasiri2017taxonomy,dinda2002online,amiri2017survey}. In particular, it is expected that this information gets more accurate the longer the job has been processed~\cite{chtepen2012online,nadeem2009predicting,hilman2018task,xie2021two}.
We embed these practical approaches in the following theoretical model: For any $\alpha \in [0,1]$, every job \emph{emits a signal} after an $\alpha$-fraction of it has been processed. Thus, an algorithm essentially learns about a job's remaining processing time when it emits.
We call this model $\alpha$-clairvoyant. Observe that it naturally interpolates between the clairvoyant setting ($\alpha=0$) and the non-clairvoyant setting ($\alpha=1$).

The $\alpha$-clairvoyant model has previously been mentioned in an extended abstract by \textcite{YingchareonthawornchaiT17}, where they also presented an $\cO(1)$-competitive algorithm if $0 \leq \alpha < \frac{1}{4}$. However, they also showed that their approach is \emph{not} $\cO(1)$-competitive for $\alpha \ge 0.5$, and it remained open whether a $\cO(1)$-competitive algorithm for every constant $\alpha \in [0,1)$ is possible.

While in general $\alpha$-clairvoyance is incomparable
to speed augmentation or predictions, 
there are arguments indicating that it is not stronger than these previously studied models: Unlike in the setting with predictions, an $\alpha$-clairvoyant algorithm has no additional knowledge at the time when a job arrives, and thus, must take potentially bad decisions in order to gain additional knowledge. Given $(1+\varepsilon)$-speed augmentation, an algorithm essentially can learn about a job's processing time when it completes the job and an optimum has only processed a $\Theta(1-\varepsilon)$-fraction of that job. This can roughly be translated to the $\alpha$-clairvoyant setting, with the difference 
here the algorithm must also schedule the remaining $\Theta(\varepsilon)$-fraction to complete it.
Another reason showing that speed augmentation is stronger comes from the more general objective of total \emph{weighted} flow time: Here, even a clairvoyant algorithm cannot have a constant competitive ratio~\cite{BansalC09}, but $(1+\varepsilon)$-speed augmentation suffices to achieve that~\cite{BansalD07,KimC03a}.

Despite these intuitive arguments that the $\alpha$-clairvoyant model is not as strong as the previously considered models, we are still able to give a $\cO(1)$-competitive algorithm for $\alpha$-clairvoyant total flow time scheduling whenever $\alpha \in [0,1)$.

\subsection{Our Results}

We prove the following main theorem:

\begin{theorem}\label{thm:main}
For every $\alpha \in [0,1)$, there exists a $\cO(\frac{1}{1-\alpha})$-competitive $\alpha$-clairvoyant online algorithm for minimizing the total flow time on a single machine.
\end{theorem}

A consequence of this result is the following insight: for any non-constant lower bound for the non-clairvoyant setting, an adversary needs to be able to only fix a job's remaining processing time once its value is \emph{non-constant}, i.e., in $o(1)$, even if the total processing time of the job is in $\Omega(1)$. This is exactly the property on which the $\Omega(n^{1/3})$ lower bound relies on, where the remaining processing times of all available jobs are fixed to $\Theta(n^{-1/3})$~\cite{MotwaniPT94}. This is not possible in the $\alpha$-clairvoyant setting for constant values of $\alpha$. %
Our result shows that such a construction is indeed necessary for proving such strong lower bounds.
In other words, it implies that the classic non-clairvoyant model is exactly at the border for which no constant competitive algorithms are possible.

We complement this observation with a randomized lower bound, which essentially bridges the gap towards the previous lower bounds: the later we receive information about a job's remaining processing time, the larger is the competitive ratio. Specifically, the dependency on $\alpha$ of our analysis is asymptotically best-possible.

\begin{theorem}\label{thm:main-lb}
  For every constant $\alpha \in [0,1)$, every randomized $\alpha$-clairvoyant algorithm has a competitive ratio of at least $\Omega(\frac{1}{1-\alpha})$ for minimizing the total flow time on a single machine.
\end{theorem}

While this lower bounds is mainly interesting when $\alpha$ approaches $1$, we also investigate its behavior when $\alpha$ approaches $0$. Recall that if $\alpha=0$, SRPT produces an optimal schedule.
Naturally, one would therefore expect a ``PTAS-like'' behavior where the competitive ratio approaches $1$ when~$\alpha$ goes to $0$. Surprisingly, we show that this is not the case: for every $\alpha  > 0$, the competitive ratio of every randomized or deterministic $\alpha$-clairvoyant algorithm is at least $\frac{3}{2}$ or $2$, respectively.

\subsection{Further Related Work}

\textcite{MotwaniPT94} showed that every deterministic algorithm has a competitive ratio of at least $P$, where $P$ is the largest to smallest processing time ratio, and that the algorithm which always runs any job to completion achieves this bound.

For minimizing the total \emph{weighted} flow time, \textcite{BansalC09} showed that even a clairvoyant algorithm has a competitive ratio of at least $\Omega((\log W/\log \log W)^\frac{1}{2})$, where $W$ is the largest to smallest weight ratio, thus ruling out $\cO(1)$-competitive algorithms. Currently, the best-known upper bound for clairvoyant algorithms is $\cO(\min\{\log W, \log P, \log D\})$~\cite{AzarT18}, where $D$ denotes the largest to smallest ratio of densities $\frac{w_j}{p_j}$, and builds up on earlier works~\cite{BansalD07,ChekuriKZ01}.
In the offline setting, unlike the unweighted problem, it is NP-hard~\cite{Lenstra_1977}, and a PTAS has been established very recently~\cite{ArmbrusterRW23}.

Under $(1+\varepsilon)$-speed augmentation, simple greedy clairvoyant strategies such as scheduling at any time the job of highest density are $\cO(1)$-competitive for total weighted flow time, even for parallel identical machines~\cite{BecchettiLMP06}. For unrelated machines, where every job has a potentially different processing time on every machine, there also are $\cO(1)$-competitive algorithms for both the clairvoyant~\cite{AnandGK12} and non-clairvoyant setting~\cite{ImKMP14}. These results hold partially true for even more general scheduling environments that can be described by polyhedral constraints~\cite{ImKM18,CIP25}.

For scheduling on $m$ parallel identical machines, SRPT is $\cO(\log(\min\{\frac{n}{m} , P\}))$-competitive, which is best-possible for online algorithms~\cite{LeonardiR97,LeonardiR07}. The best-known upper bound on the non-clairvoyant competitive ratio is $\cO(\log(n) \cdot \log(\min\{\frac{n}{m}, P\}))$~\cite{BecchettiL04}. 
There are also non-clairvoyant results with processing time predictions~\cite{AzarPT22,ZhaoLZ22} and with speed augmentation~\cite{ChekuriGKK04,ImKMP14}.

\subsection{Organization}

In \Cref{sec:algo}, we introduce notation and describe our algorithm. Then, in \Cref{sec:analysis}, we give a guided tour through our analysis. Several proofs and subresults for that are deferred to the appendix. Our lower bounds and the proof of \Cref{thm:main-lb} are deferred to \Cref{app:lowerbounds}.

\section{The Algorithm}\label{sec:algo}

We use standard notation for flow time scheduling. For a job $j$, we denote by $p_j$ its processing time and by $r_j$ its release date. 
For a time $t$, we denote by $y_j(t)$ the total amount of processing that $j$ has received until time $t$, and we use $p_j(t) \coloneq p_j - y_j(t)$ to denote its remaining processing time. Moreover, for a time interval $I=[t_1,t_2]$, we write $q_j(I) \coloneq y_j(t_2) - y_j(t_1)$ to denote the amount of processing that $j$ receives during that interval. We denote by $C_j$ the completion time of $j$.
For a time $t$, we use $A(t) = \{j \mid r_j \leq t < C_j \}$ to refer to the set of available jobs at time $t$.

Our algorithm is a combination of the following two classic algorithms for flow time scheduling from the literature:
\begin{itemize}
  \itemsep0pt
  \item \textbf{Shortest-Remaining-Processing-Time-first (SRPT):} At any time $t$, schedule the job in $A(t)$ that has the smallest remaining processing time at $t$.
  \item \textbf{Shortest-Elapsed-Time-first (SETF):} At any time $t$, process all jobs $j \in A(t)$ with smallest $y_j(t)$ at the same rate, that is, evenly distribute the processing power of $1$ among these jobs.
\end{itemize}
 
SRPT is 1-competitive for clairvoyant online scheduling~\cite{Schrage68}, and SETF is $\cO(\frac{1}{\varepsilon})$-competitive for non-clairvoyant online scheduling with $(1+\varepsilon)$-speed augmentation, for every $\varepsilon > 0$~\cite{KalyanasundaramP00}.

For the $\alpha$-clairvoyant setting, we introduce additional notation.
For a job $j$, we denote by $s_j$ the time at which job $j$ \emph{emitted its signal}, that is, the earliest point in time $t$ at which $y_j(t) = \alpha p_j$.
Moreover, at any time $t$, we partition $A(t)$ into the set $N(t) \coloneq \{j \in A(t) \mid y_j(t) \leq \alpha p_j \}$ of \emph{non-clairvoyant} jobs and the set $C(t)  \coloneq \{j \in A(t) \mid y_j(t) > \alpha p_j \}$ of \emph{clairvoyant} jobs. Finally, we use $D(t)$ for the set of jobs that have been completed by time $t$.

\paragraph{Intuition}

On a high-level, at any time $t$, 
our algorithm either executes SETF on the set of non-clairvoyant jobs $N(t)$,
or it executes SRPT on the set of clairvoyant job $C(t)$.
Note that the latter is possible, because the algorithm can derive the remaining processing time of jobs in $C(t)$.
The crucial decision is which of both rules the algorithm should execute at time $t$. Intuitively, we design this decision such that whenever the algorithm executes SRPT, its behavior is equivalent to executing SRPT on \emph{all} available jobs $A(t)$, and similarly, whenever the algorithm executes SETF, its behavior is equivalent to executing SETF on \emph{all} available jobs $A(t)$. 
In other words, whenever the algorithm knows that a certain job $j$ has the least remaining processing time among all available jobs, it will work on $j$.

\paragraph{Algorithm Description}
To achieve this, we define the algorithm as follows for every time $t$:

\begin{enumerate}
  \item If $\min_{j \in C(t)} p_j(t) \leq \frac{1-\alpha}{\alpha} \cdot \min_{j \in N(t)} y_j(t)$, schedule the job in $C(t)$ with the smallest remaining processing time at time $t$.  If there are multiple such jobs, schedule the job in $C(t)$ that emitted last, i.e., a job $j \in C(t)$ of maximum $s_j$. Break all remaining ties arbitrarily.
  \item Otherwise, that is, $\min_{j \in C(t)} p_j(t) > \frac{1-\alpha}{\alpha} \cdot \min_{j \in N(t)} y_j(t)$, process all jobs in $N(t)$ that have received the smallest amount of processing so far at the same rate.
\end{enumerate}

\paragraph{Properties}

We now verify that this algorithm ensures the properties mentioned above. First, note that if $\min_{j \in C(t)} p_j(t) \leq \frac{1-\alpha}{\alpha} \cdot \min_{j \in N(t)} y_j(t)$, we schedule a job $k$ that satisfies $p_{k}(t) = \min_{j' \in C(t)} p_j(t) \leq  \frac{1-\alpha}{\alpha} \cdot \min_{j' \in N(t)} y_{j'}(t) \leq \frac{1-\alpha}{\alpha} \cdot y_j(t) \leq \frac{1-\alpha}{\alpha} \cdot \alpha p_j \leq (1-\alpha)p_j \leq p_j(t)$ for every job $j \in N(t)$, and $p_{k}(t) \leq p_j(t)$ for every job $j \in C(t)$. Here, $y_j(t) \le \alpha p_j$ and $p_j(t) \ge (1-\alpha)p_j$ hold by definition of $N(t)$ and since $j \in N(t)$. Thus, we arrive at the following observation.

\begin{observation}\label{obs:srpt-like}
    At any time $t$, if $\min_{j \in C(t)} p_j(t) \leq \frac{1-\alpha}{\alpha} \cdot \min_{j \in N(t)} y_j(t)$, then the algorithm processes a job $k \in C(t)$ with shortest remaining processing time, i.e., $p_{k}(t) \leq p_j(t)$ for all $j \in A(t)$.
\end{observation}

Second, if $\min_{j \in C(t)} p_j(t) > \frac{1-\alpha}{\alpha} \cdot \min_{j \in N(t)} y_j(t)$, we schedule jobs $k$ that satisfy $y_{k}(t) \leq y_j(t)$ for every $j \in N(t)$, and $y_{k}(t) = \min_{j' \in N(t)}y_{j'}(t) < \frac{\alpha}{1-\alpha} \cdot \min_{j' \in C(t)} p_{j'}(t) \leq \frac{\alpha}{1-\alpha} p_{j}(t) \leq \alpha p_{j} \leq y_{j}(t)$ for every $j \in C(t)$. This gives us the following observation.

\begin{observation}\label{obs:sept-like}
    At any time $t$, if $\min_{j \in C(t)} p_j(t) > \frac{1-\alpha}{\alpha} \cdot \min_{j \in N(t)} y_j(t)$, then the algorithm processes jobs $S \subseteq N(t)$ with minimal progress, i.e., $y_{j'}(t) \leq y_j(t)$ for all $j' \in S$ and $j \in A(t)$.
\end{observation}

Finally, the definition of our algorithm implies the following property, which will be important in the analysis by describing the structure of the algorithm's schedule.

  \begin{observation}
      \label{obs:clairvoyant-jobs-block-earlier-jobs}
      Consider a job $k$ that is executed at time $t'$ with $k \in C(t')$. Then, for every job $j$ with $r_j < t' < C_j$ it holds that $j$ is not executed during $[t', C_k]$ and $C_k \leq C_j$. In particular, this applies to a job $k$ that emits at time $t'$, that is, $t' = s_k$.
  \end{observation}

  \begin{proof}
      Consider a job $j$ with $r_j < t' < C_j$.
      Since $k$ is executed at time $t'$, we have by the definition of the algorithm that $p_k(t') \leq \frac{1-\alpha}{\alpha} y_j(t')$ if $j \in N(t')$ and $p_k(t') \leq p_j(t')$ if $j \in C(t')$. In any case, since $p_k(t'^+) < p_k(t')$, the algorithm would never process $j$ after time $t'$ if $k$ is also available. %
  \end{proof}

\section{Overview of the Analysis}\label{sec:analysis}

In this section, we give an extensive overview of our analysis. We will prove our main theorem (\Cref{thm:main}) for the algorithm presented in \Cref{sec:algo} via \emph{local competitiveness}~\cite{Schrage68}. To this end, we fix an instance and an optimal solution, and we denote by $O(t)$ the set of available jobs at time $t$ in that solution. 
We assume w.l.o.g.\ that $\frac{1}{1-\alpha}$ is integer, which does not change the bound in \Cref{thm:main}.

We prove the following theorem.

\begin{theorem}\label{thm:main:local}
  At any time $t$, it holds that $|A(t) \setminus O(t)| \leq \big( 3+ \frac{2}{1-\alpha} \big) \cdot |O(t)|$.
\end{theorem}

Then, $|A(t)| \leq \big( 4+ \frac{2}{1-\alpha} \big) \cdot |O(t)|$, and \Cref{thm:main} follows via integration over time.
Thus, it remains to prove \Cref{thm:main:local}. To this end, we fix any time $t$ for the remainder of this section. Note that if $|O(t)|=0$, then also $|A(t)| = 0$ because our algorithm is non-idling. Thus, we can assume that $|O(t)| \geq 1$. Moreover, we assume that $|A(t) \setminus O(t)| \geq 1$, as otherwise the statement is trivial. 
For a job $j$, we denote with $t_j$ the latest point in time before or at $t$ at which $j$ is being processed.

\subsection{Borrowing}

We start by defining \emph{borrowing}~\cite{KalyanasundaramP00}. The goal of this concept is to bring the algorithm's schedule in relation to the fixed optimal schedule by describing local transformations that allow to modify the algorithm's schedule into the optimal schedule. 
We prove in the next section that such a transformation via these local operations must be possible, which will ultimately allow us to compare both schedules.
To make schedules easier to describe, we first introduce \emph{lifetimes}.

\begin{definition}[Lifetime]
  Consider a fixed schedule.
  For a job $j$, we refer to $I_j \coloneq [r_j,\min\{C_j,t\}]$ as the \emph{lifetime} of $j$. The \emph{lifetime} $I(Q)$ of a set $Q \subseteq J$ is defined as $I(Q) \coloneq \bigcup_{j \in Q} I_j$.
\end{definition}

Now, observe that if a job $i$ receives $\gamma$ units of processing during the lifetime of another job $j \in A(t)$, we can find a similar schedule where instead of $i$ we process $j$. As a result, $j$ receives up to $\gamma$ units of processing more until time $t$. We say that $j$ \emph{borrowed} from $i$. We further distinguish whether job $i$ receives processing while being non-clairvoyant or clairvoyant. Clearly, there can be both situations for some fixed $j$ and $i$. In the following, we give precise definitions of this extension of the concept used in~\cite{KalyanasundaramP00}.

Borrowing is not tailored to the schedule of our algorithm. Thus, for the remaining section, we fix any non-idling schedule $S$. For any schedule-specific quantity, we add as superscript the schedule in which we consider this quantity. E.g., we denote with $A^S(t')$ the set of available jobs at time $t'$ in $S$. In later sections, when the schedule is clear from the context, we will omit any such indications.

\begin{definition}[Borrowing]\label{def:borrowing}
  We say that a job $j$ can \emph{directly borrow non-clairvoyantly} from $i$ if $i$ is executed at any time $t'$ with $i \in  N^S(t')$ during the lifetime of $j$, that is, $t' \in I^S_j$, and write $j \borrow^S_N i$.
  We say that a job $j$ can \emph{directly borrow clairvoyantly} from $i$ if $i$ is executed at any time $t'$  with $i \in C^S(t')$ during the lifetime of $j$, that is, $t' \in I^S_j$, and write $j \borrow^S_C i$.
\end{definition}

\begin{definition}[Borrow Graph]\label{def:borrow-graph}
Let $E^S_N \coloneq \{(j,i) \mid j \borrow^S_N i \}$ and $E^S_C \coloneq \{(j,i) \mid j \borrow^S_C i \}$. The borrow graph is the multi-graph $G^S_B = (J, E^S)$ with $E^S = E^S_N \cup E^S_C$.
\end{definition}

For a job $j$, we denote by $R^S_j$ the set of jobs that are reachable from $j$ in $G^S_B$ (including $j$).
Moreover, for every $i \in R^S_j$, we say that $j$ can \emph{borrow} from $i$.
Given a set $Q$ of edges or vertices of $G^S_B$, we call a path a  $Q$-path if it only uses edges or vertices in $Q$, respectively. For example, we call a path $P$ in $G_B$ that only uses edges from $E^S_N$ an \emph{$E^S_N$-path}.

\begin{lemma}\label{obs:lifetime:1}
  Let $P$ be a path in $G^S_B$. Then, $I^S(P) = [\min_{j \in P} r_j, \max_{j \in P} \min\{C^S_j, t\}]$.
\end{lemma}

\begin{proof}
  By definition of $I^S(P)$, we clearly have that $\min_{j \in P} r_j$ is the smallest value in $I^S(P)$, and that $\max_{j \in P} \min\{C^S_j, t\}$ is the largest value in $I^S(P)$. Thus, it remains to argue that $I^S(P)$ has no gaps. 

  To see this, assume otherwise, and let $t'$ be a value with $\min_{j \in P} r_j < t' < \max_{j \in P} \min\{C^S_j, t\}$ and $t' \not\in I^S(P)$. 
  Define $P_1 \coloneq \{j_1 \in P \mid  I^S_{j_1} \subseteq [\min_{j\in P} r_j , t')\}$ and 
  $P_2 \coloneq \{j_2 \in P \mid I^S_{j_2} \subseteq (t', \max_{j \in P} \min\{C^S_j, t\}]\}$. 
  Observe that $P_1 \cup P_2 = P$, and, by definition of $I^S(P)$, neither $P_1$ nor $P_2$ can be empty. 
  Thus, there must exist at least one edge $e$ on $P$ between a job $j_1 \in P_1$ and $j_2 \in P_2$ (in some direction). However, $I^S_{j_1} \cap I^S_{j_2} = \emptyset$ implies that $j_1$ and $j_2$ cannot directly borrow from each other; a contradiction to the existence of $e$ in $G^S_B$.
\end{proof}

\Cref{obs:lifetime:1} immediately gives the following corollaries.

\begin{corollary}\label{obs:flow:cut-lifetime}
  For any path $P$ in $G^S_B$, $I^S(P) = [\min_{j \in P} r_j, \max_{j \in P} \min\{C^S_j, t\}]$. As a consequence, for every job $j$, we have
  $I^S(R^S_j) = [\min_{j \in R^S_j} r_j, \max_{j \in R^S_j} \min\{C^S_j, t\}]$.
\end{corollary}

\begin{corollary}
  \label{obs:lifetime:2}
  If a path $P$ in $G^S_B$ contains a job $j \in A^S(t)$, then $I^S(P) = [\min_{j \in P} r_j, t]$. As a consequence, for every job $j \in A^S(t)$, we have
  $I(R^S_j) = [\min_{j \in R^S_j} r_j, t]$.
\end{corollary}

Finally, we have the following simple proposition.

\begin{proposition}\label{obs:flow:cut-closure}
  Consider a job $j$.
  Every job $j'$ that is executed (in particular, released) during $I(R^S_j)$ satisfies $j' \in R^S_j$.
\end{proposition}

\begin{proof}
  If $j'$ is executed at some time $t' \in I(R^S_j)$, then there is an edge $(d,j')$ in $G^S_B$ for some $d \in R^S_j$, and thus, %
  $j' \in R^S_j$.
\end{proof}

\subsection{The Flow Problem}
\label{sec:flow-problem}

We continue with a refinement of the borrow graph. The idea is that the borrow graph only tells us which job can in principle borrow from which job, but not \emph{how much} amount of processing. For our analysis, however, it will be important to argue about the maximum amount of work that multiple jobs can borrow via some job $i$ in a certain time interval. 

To this end, we introduce a refinement of the borrow graph and define a flow problem for it. A feasible flow then exactly describes how much amount of work jobs can borrow from each other. In particular, this is a non-trivial extensions to the arguments used in \cite{KalyanasundaramP00}. 

This concept, similar to borrowing, is not tailored to the schedule of our algorithm. 
Thus, throughout this section, we consider any non-idling schedule $S$.

\begin{definition}[Flow Network]\label{def:flow-network}
  Let $T^S$ be any time discretization such that $T^S \supseteq \{ r_j \mid r_j \leq t, j \in J \} \cup \{ C^S_j \mid r_j \leq t, j \in J \} \cup \{0, t\}$ and $\tau \leq t$ for every $\tau \in T^S$, and denote $T^S = \{\tau_1,\ldots,\tau_{\ell+1} \}$ such that $\tau_{\ell'} \leq \tau_{\ell''}$ if $\ell' < \ell''$.
  The \emph{flow network} $G^S_F$ w.r.t.\ $T^S$ is defined as follows:
  \begin{itemize}
  \item For every job $j \in J$, there is a vertex $j$.
  \item For every job $i \in J$ %
  and for every $\ell' \in \{1,\ldots,\ell\}$, 
  we introduce one vertex $v^i_{\ell'}$ and an edge $(v^i_{\ell'}, i)$ with capacity $q^S_i([\tau_{\ell'},\tau_{\ell'+1}])$, and for every %
  $j \in J \setminus \{i\}$, we introduce the edge $(j, v_{\ell'}^i)$ with infinite capacity 
  if $[\tau_{\ell'},\tau_{\ell'+1}] \subseteq I_j$.
  We also refer to the vertices $v_{\ell'}^i$ as \emph{dummy vertices}.
  \item Every vertex $j \in A^S(t) \setminus O(t)$ has a supply equal to $p^S_j(t)$ and every vertex $i \in O(t)$ has a demand equal to $y^S_i(t)$. All other supplies and demands are equal to zero.
  \end{itemize}
\end{definition}

Let $f \colon V(G^S_F) \times V(G^S_F) \to \mathbb{R}$ be a feasible flow in $G^S_F$. We use the standard convention that if $(v,w) \in V(G^S_F) \times V(G^S_F)$ is not in $E(G^S_F)$, then $f(v,w)=0$. We extend the notation for $f$ as follows. For all $j,i \in J$, we denote by $f(j,i)$ the total amount of flow that is directly sent from $j$ to $i$ in $G^S_F$, that is, $f(j,i) \coloneq \sum_{\ell'=1}^\ell f(j,v_{\ell'}^i)$.

The main theorem of this section guarantees the existence of a feasible flow in $G^S_F$ that uses the total supply. 

\begin{theorem}\label{thm:flow:main}
  There exists a feasible flow in $G^S_F$ that uses the total supply of jobs in $A^S(t) \setminus O(t)$. %
\end{theorem}

Intuitively, this implies that there exists a set of borrow operations such that every job $j \in A^S(t) \setminus O(t)$ can be completed until time $t$, and only the work of jobs in $O(t)$ until time $t$ decreases by these operations, as these jobs are the only vertices in $G_F^S$ with positive demand.
The remainder of this section is devoted to proving this theorem. %

We start with the observation that we arbitrarily increase the precision of the time discretization without changing $f(j,i)$ for all $j,i \in J$. Note that increasing the precision essentially does not change the structure of the flow network and may only lead to %
more dummy vertices. %

\begin{restatable}{observation}{obsflowincreaseprecision}\label{obs:flow-increase-precision}
  Let $f$ be a feasible flow for $G_F^S$ w.r.t.\ $T^S$, and $0 \leq t' \leq t$. Then, 
  there exists a feasible flow $f'$ for $G_F^S$ w.r.t.\ $T^S \cup \{t'\}$ such that $f(j,i) = f'(j,i)$ for all $j,i \in J$.
\end{restatable}

The following proposition holds as we can compute a maximum flow for $G_F^S$ by using a standard augmenting-path algorithm that never uses vertices in $O(t)$ as inner nodes of an augmenting path.%

\begin{proposition}
  \label{prop:flow:no-outgoing-sink-flow}
  There always exists a maximum feasible flow $f$ in $G_F^S$ such that $\sum_{j \in J \setminus \{i\}} f(i,j) = 0$ for all $i \in O(t)$. That is, vertices with positive demand do not have outgoing flow in $f$. 
\end{proposition}

The following lemma implies that %
we can find a new schedule in which the supply and demands changed according to the flow.

\begin{restatable}{lemma}{lemmaflowborrowapplyedge}\label{lemma:flow:borrow-apply-edge}
  Let $f$ be a feasible flow in $G^S_F$ and let $j,i \in J$ be such that $f(j,i) > 0$ %
  and $\sum_{i' \in J \setminus \{j\}} f(i',j) = 0$.
  Then, there exists a $\gamma > 0$, a feasible non-idling schedule $S'$, and a feasible flow $f'$ for $G_F^{S'}$ such that
  \begin{enumerate}[(a)]
    \item $p_j^{S'}(t) = p^S_j(t) - \gamma$ and $p_i^{S'}(t) = p^S_i(t) + \gamma$,
    \item $y_k^{S'}(t) = y^S_k(t)$ for all jobs $k \in J \setminus \{i,j\}$,
    \item $f'(j,i) = f(j,i) - \gamma$, and
    \item $f'(k,k') = f(k,k')$ for all job pairs $k,k' \in J$ with $\{k,k'\} \not= \{i,j\}$.
  \end{enumerate}
\end{restatable}

\begin{restatable}{lemma}{lemmaflowpositivepaths}\label{lemma:flow:positive-paths}
  Let $j \in A^S(t) \setminus O(t)$ and $i \in O(t)$. Then, there is path through jobs $j,k_1,\ldots,k_m,i$ with positive capacity in $G^S_F$  if and only if there is a path through jobs $j,k_1,\ldots,k_m,i$ in $G_B^S$.
\end{restatable}

We finally prove the main theorem of this section.

\begin{proof}[Proof of \Cref{thm:flow:main}]
  Fix any non-idling schedule $S$.
  For the sake of contradiction, we assume that there exists no feasible flow in $G^S_F$ of value $\sum_{j \in A^S(t) \setminus O(t)} p^S_j(t)$. Let $f$ be a feasible flow of maximum value $v(f)$, which is the total amount of supply used by $f$. By assumption, we have that $v(f) < \sum_{j \in A^S(t) \setminus O(t)} p^S_j(t)$.

  Now, observe that if there exists a job with positive supply, then there also exists a job in $G^S_F$ without any incoming flow, as otherwise we send flow on a cycle. Thus, we can iteratively and exhaustively apply \Cref{lemma:flow:borrow-apply-edge}. Let $S'$ be the resulting schedule, and $f$ the resulting flow in $G_F^{S'}$.

  Note that since $v(f) < \sum_{j \in A^S(t) \setminus O(t)} p_j(t)$ and by the properties of \Cref{lemma:flow:borrow-apply-edge}, $A^{S'}(t)\setminus O(t)$ cannot be empty, that is, there exists a job $j \in J \setminus O(t)$ such that $p^{S'}_j(t) > 0$. Moreover, for every job $j \in A^{S'}(t) \setminus O(t)$, there cannot exist a job $i \in R^{S'}_j \cap O(t)$ with $y^{S'}_i(t) > 0$.
  Otherwise, \Cref{lemma:flow:positive-paths} guarantees that there exists a path $P$ of positive capacity from $j$ to $i$, and since $p_j(t) > 0$ and $y_i(t) > 0$, we can find a flow with positive value on $P$ and apply \Cref{lemma:flow:borrow-apply-edge} to all consecutive jobs on $P$ to find a flow of a larger value; a contradiction. %
  Thus, $R^{S'}_j \cap O(t) = \emptyset$ for every job $j \in A^{S'}(t) \setminus O(t)$, because a job $i$ with $y^{S'}_i(t) = 0$ clearly cannot be reached by any other job in $G_B^{S'}$.

  Let $j_1$ be the job with minimum release date in $A^{S'}(t) \setminus O(t)$. Thus, for all $j \in A^{S'}(t) \setminus O(t)$, \cref{obs:flow:cut-lifetime} gives $j \in R^{S'}_{j_1}$, and   thus, $R^{S'}_j \subseteq R^{S'}_{j_1}$.
  Let $r^*$ denote the minimum release date over the jobs in $R^{S'}_{j_1}$. Note that we might have $r^* \not= r_{j_1}$ if the minimum release date belongs to a job that completes by time $t$ in both $S'$ and the optimal schedule.
  By definition of lifetimes, we have  $I(R^{S'}_{j_1}) = [r^*,t]$.
  Thus, \cref{obs:flow:cut-closure} gives that every job $j \in J$ that is processed after $r^*$ and completed by time $t$ is also contained in $R^{S'}_{j_1}$. In particular, this means that $S'$ only processes jobs in $R^{S'}_{j_1}$ during $[r^*,t]$. However, as $j_1 \in A^{S'}(t) \setminus O(t)$, the schedule $S'$ is not able to complete all of these jobs by point in time $t$. As $S'$ is non-idling and only works on jobs in $R^{S'}_{j_1}$ during $[r^*,t]$, this implies $\sum_{j \in R^{S'}_{j_1}} p_j > t - r^*$.

  Since $R^{S'}_{j_1} \cap O(t) = \emptyset$, we know that the optimal solution is able to complete every job of $R^{S'}_{j_1}$ by only working on them during $[r^*,t]$. 
  However, this is a contradiction to $\sum_{j \in R^{S'}_{j_1}} p_j > t - r^*$.
\end{proof}

\subsection{Borrowing via Flows}
\label{sec:borrowing-via-flows}

For this section (and the following), we again consider the schedule of the algorithm.

\Cref{thm:flow:main} guarantees the existence of a flow $f$ in $G_F$ that uses the total supply of vertices in $A(t) \setminus O(t)$, and thus, has a total flow value of $\sum_{j \in A(t) \setminus O(t)} p_j(t)$. In the following, we decompose such a flow and give properties regarding the total amount of flow going from every source $j \in A(t) \setminus O(t)$ to every sink $i \in O(t)$ in $G_F$. The following definition makes this more precise.

\begin{definition}
  \label{def:borrow:values}
  Let $f$ be a flow in $G_F$ of value $\sum_{j \in A(t) \setminus O(t)} p_j(t)$, which exists by \Cref{thm:flow:main}. Furthermore, let $\mathcal{P}$ denote a decomposition of $f$ into a finite number of path flows. For $j \in A(t) \setminus O(t)$ and $i \in O(t)$, let $\mathcal{P}_{ji}$ denote the subset of paths in $\mathcal{P}$ that start in $j$ and end in $i$. For a path $P \in \mathcal{P}_{ji}$ let $\gamma(P)$ denote the amount of flow that $P$ sends from $j$ to $i$. We define:
  \begin{itemize}
    \item $\beta(j,i) = \sum_{P \in  \mathcal{P}_{ji}} \gamma(P)$ for all pairs $(j,i)$ with $j \in A(t) \setminus O(t)$ and $i \in O(t)$.
    \item $\beta(j,i) = 0$ for all remaining pairs of jobs.
  \end{itemize}
\end{definition}

  In the appendix, we show the following theorem, which establishes properties of the values $\beta(j,i)$ that will be crucial for proving our main result.

\begin{restatable}{theorem}{thmmainborrowing}\label{thm:main:borrowing}
The values $\beta(j,i)$, for every $j \in A(t) \setminus O(t)$ and $i \in O(t)$, satisfy the following properties:
  \begin{enumerate}[(i)]
    \item $\beta(j,i) = 0$ for every $j \in A(t) \setminus O(t)$ and $i \in O(t) \setminus R_j$,
    \item $p_j(t) = \sum_{i \in O(t)} \beta(j,i)$ for every $j \in A(t) \setminus O(t)$, and
    \item $y_i(t) \geq \sum_{j \in A(t) \setminus O(t)} \beta(j,i)$ for every $i \in O(t)$.
  \end{enumerate}
\end{restatable}

\subsection{Matching Arguments}
\label{sec:matching:arguments}

The goal of this section is to give intuition on how we use \Cref{thm:flow:main} to bound $|A(t) \setminus O(t)|$ by some function of $|O(t)|$, and eventually prove \Cref{thm:main:local}. To do so, we heavily use the structure of our algorithm's schedule.

We first discuss a scenario where $D(t) = C(t) = \emptyset$, that is, every job that has been released until time $t$ is unfinished and non-clairvoyant at time $t$.
This means that no job has emitted so far, and thus, the algorithm has only executed SETF. Consider two jobs $j \in N(t) \setminus O(t)$ and $i \in O(t)$.

In this scenario, we can observe that job $j$ can only borrow from $i$ if $j$ can \emph{directly borrow} from~$i$. A more general observation was made by \textcite{KalyanasundaramP00} in their analysis of SETF.

\begin{observation}[Lemma~8 in \cite{KalyanasundaramP00}]
If $D(t) = C(t) = \emptyset$, then $i \in R_j$ implies $(j,i) \in E_N$.
\end{observation}

By the definition of the algorithm, this 
means that if $i \in R_j$ then there must be a point in time $t'$ during $I_i \cap I_j$ at which we have $y_j(t') \geq y_i(t')$. Since both $i$ and $j$ are non-clairvoyant at time $t$, we can conclude $y_j(t) \geq y_i(t)$.

For an arbitrary scenario, i.e., without the assumption that $D(t) = C(t) = \emptyset$,  these two observations are stated in the following two lemmas.

\begin{lemma}[Catch-Up]\label{lemma:ketchup}
  If the algorithm processes job $i$ at time $t'$ and $j,i \in N(t')$, then $y_j(t'') \geq y_i(t'')$ for all times $t'' \geq t'$ such that $j,i \in N(t'')$.
\end{lemma}

\begin{proof}
  Since the algorithm processes job $i$ at time $t'$ and $j,i \in N(t)$, \Cref{obs:sept-like} implies $y_j(t') \geq y_i(t')$. The lemma then follows because as long as both $j$ and $i$ are non-clairvoyant, the algorithm will at any time process $j$ at most as much as $i$.
\end{proof}

\begin{lemma}\label{lemma:nonclairvoyant-direct-borrow}
  If $j,i \in N(t)$ and $(j,i) \in E_N$, then $y_{j}(t) \geq y_i(t)$.
\end{lemma}

\begin{proof}
  Let $t'$ be first point in time during $I_j$ at which $i$ is being processed. Since $j$ and $i$ are both non-clairvoyant at time $t$, we conclude $j,i\in N(t')$. Then, $y_{j}(t) \geq y_i(t)$ follows from \Cref{lemma:ketchup}.
\end{proof}

Continuing the special case of above, since $j \in N(t)$, at least a $(1-\alpha)$-fraction of the job is unfinished at time $t$. More specifically, $p_j(t) \geq (1-\alpha) y_j(t)$, and thus, $\frac{1}{1-\alpha} p_j(t) \geq y_i(t)$. This means that the optimal solution can finish at most $\frac{1}{1-\alpha}$ jobs in $A(t) \setminus O(t)$ by borrowing work from~$i \in O(t)$. Formally, using \Cref{thm:main:borrowing}, we have
\begin{align*}
  |A(t) \setminus O(t)| 
  = \sum_{j \in A(t) \setminus O(t)} \frac{p_j(t)}{p_j(t)} 
  = \sum_{j \in A(t) \setminus O(t)} \sum_{i \in O(t)} \frac{\beta(j,i)}{p_j(t)} 
  &\leq \frac{1}{1-\alpha} \sum_{i \in O(t)} \sum_{j \in A(t) \setminus O(t)}  \frac{\beta(j,i)}{y_i(t)} \\
  &\leq \frac{1}{1-\alpha} \sum_{i \in O(t)} \frac{y_i(t)}{y_i(t)} = \frac{1}{1-\alpha} |O(t)| \ ,
\end{align*}
proving \Cref{thm:main:local} for this scenario. In the first equality, we used that $\frac{1}{1-\alpha} p_j(t) \geq y_i(t)$ for every $i \in O(t) \cap R_j$, which is possible because $\beta(j,i) = 0$ if $i \in O(t) \setminus R_j$ by \Cref{thm:main:borrowing}. 
The second inequality uses Property $(iii)$ of~\Cref{thm:main:borrowing}.
This is roughly the strategy that has been used for the analysis of SETF under speed augmentation~\cite{KalyanasundaramP00,BansalD07}. For our algorithm on the other hand, this proof will \emph{not} work in general.

The reason for this is that we relied on the two inequalities (i) $p_j(t) \geq (1-\alpha) y_j(t)$ and (ii) $y_i(t) \le y_j(t)$. Together, these inequalities imply $\frac{1}{1-\alpha} p_j(t) \geq y_i(t)$ for all $j \in A(t)\setminus O(t)$ and $i \in R_j \cap O(t)$, meaning that the remaining processing time of $j$ (the supply of $j$ in the flow problem) is not too small compared to the received processing time of $i$ (the demand of $i$ in the flow problem).

Inequality (i) only holds because of $j \in N(t)$. However, as $A(t) = N(t) \cup C(t)$, we might have $j \in C(t)$ in the general case. To handle this complication, 
 
we will prove \Cref{thm:main:local} for both $N(t)$ and $C(t)$ separately, which is encapsulated in the following two theorems.

\begin{theorem}\label{thm:local:nonclairvoyant}
  At any time $t$, it holds that $|N(t) \setminus O(t)| \leq \big(2 + \frac{1}{1-\alpha} \big) \cdot |O(t)|$.
\end{theorem}

\begin{theorem}\label{thm:local:clairvoyant}
  At any time $t$, it holds that $|C(t) \setminus O(t)| \leq \big(1 + \frac{1}{1-\alpha} \big) \cdot |O(t)|$.
\end{theorem}

Proving the latter theorem will be way easier than proving the former, and only requires small additions to the scenario sketched above. Thus, we defer the proof of \Cref{thm:local:clairvoyant} to later (cf. \Cref{sec:clairvoyant}), and for now, we concentrate on proving \Cref{thm:local:nonclairvoyant}.

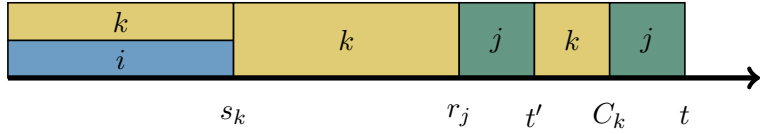
\begin{figure}[tb]
  \begin{center}
    \begin{tikzpicture}
      \jobT[job1](0,0)(3,0.5)($i$);
      \jobT[job2](0,0.5)(3,0.5)($k$);
      \jobT[job2](3,0)(3,1)($k$);
      \jobT[job3](6,0)(1,1)($j$);
      \jobT[job2](7,0)(1,1)($k$);
      \jobT[job3](8,0)(1,1)($j$);

      \node at (3,-0.5) {$s_k$};
      \node at (6,-0.5) {$r_j$};
      \node at (7,-0.5) {$t'$};
      \node at (8,-0.5) {$C_k$};
      \node at (9,-0.5) {$t$};

      \draw[line width=2pt,->] (0,0) -- (10,0);
    \end{tikzpicture}
  \end{center}
  \caption{Situation of our algorithm where the simple matching argument does not work.}\label{fig:example-delta-edge}
\end{figure}

To prove~\Cref{thm:local:nonclairvoyant}, we only need to consider jobs $j \in N(t)\setminus O(t)$ that satisfy Inequality (i) using the same argument as above. Unfortunately, Inequality (ii) does not hold in general for all $j \in N(t)\setminus O(t)$ and $i \in O(t)$, which again breaks the above approach.
In particular, there can be a situation where $i \in R_j$ but $y_{j}(t) < c \cdot y_i(t)$ for any $c > 0$.
To this end, we consider a scenario that is illustrated in \Cref{fig:example-delta-edge}. Consider two jobs $i$ and $k$ that run non-clairvoyantly at the same rate until $k$ emits, that is, we have $y_i(s_k) = y_k(s_k) = \alpha p_k$. Then, our algorithm processes only $k$ until at time $r_j$ job $j$ is being released, and the algorithm switches to processing it until time $t'$, which satisfies $\frac{1-\alpha}{\alpha} y_j(t') \geq p_k(t')$. Then, the algorithm runs $k$ to completion, and continues working on $j$ until time $t$. 

Note that $j$ can borrow from $k$, as $(j,k) \in E_C$, and $k$ can borrow from $i$ as $(k,i) \in E_N$. Thus, $i \in R_j$.
Note that $p_k(t')$ can be an arbitrarily small positive number. Thus, we cannot argue that $y_j(t) \geq y_i(t) = y_i(s_k)$ as in the special case above. In particular, it can be that $y_j(t) < c \cdot y_i(t)$ for every $c > 0$, ruling out the simple analysis outlined above.

Specifically, there could be multiple jobs that, in the same way as $j$, can borrow from $i$ via~$k$. Since $y_j(t)$ can be arbitrarily small compared to $y_i(t)$, the received processing time of $i$ would be enough to cover the remaining processing time of arbitrarily many such jobs $j$. However, a crucial observation will be that such jobs $j$ can borrow from $i$  \emph{only} via $k$, and $k$ can be seen as a \emph{bottleneck}. Our main proof idea is to compare the remaining processing time of $j$ to the capacity of the bottleneck instead of comparing directly with $y_i(t)$. The next section makes this idea more precise.

\subsection{Matching via Segments and the Proof of \Cref{thm:local:nonclairvoyant}}
\label{sec:segments}

For the sake of clarity, we use in this section the notation $N  \coloneq N(t)$ and $O \coloneq O(t)$.

We start with the definition of \emph{segments}. Intuitively, a segment is a subset of jobs $j \in S \subseteq N \setminus O$ that can borrow from the same set of jobs of $i \in O$ such that $y_{j}(t) \geq y_i(t)$. As outlined in the above section, this inequality essentially allows us to argue that at most $O(\frac{1}{1-\alpha})$ jobs $j$ can be completed by borrowing the total amount of processing of $i$.
Actually, for this argument suffices a slightly weaker inequality based on a truncation of $y_j(t)$ to $\alpha p_j$:

\begin{definition}[Truncated progress]
  For every job $j$, we define $\by_j \coloneq \min\{y_j(t), \alpha \cdot p_j\}$.
\end{definition}

We now work towards the formal definition of a segment. We use the following notation:
\begin{definition}\label{def:steal-types}
  For every $j \in N \setminus O$, we define 
  \begin{itemize}
    \item $O_j \coloneq \{ i \in O \mid \by_j \geq \by_i \}$, and 
    \item $\bO_j \coloneq \{i \in O \mid\by_j < \by_i \text{ and } i \in R_j \}$.
  \end{itemize}
\end{definition}

Note that $\bO_j \subseteq O \setminus O_j$ for every $j \in N \setminus O$. In principle, we could also add the requirement that $i \in R_j$ to the definition of $O_j$. However, some proofs will be easier without this requirement.

\begin{definition}[Segment]
  \label{def:segment}
A maximal set of jobs $S \subseteq N \setminus O$ is a \emph{segment} if $S \neq \emptyset$ and $O_j = O_{j'}$ for all $j,j' \in S$. 
We denote by $\cS$ the set of all segments. For a segment $S \in \cS$, we use $O_S$ to refer to the set $O_j$ for each $j \in S$.
\end{definition}

Note that $\cS$ is a partition of $N \setminus O$.
For a segment $S$, we define $\bS \coloneq \{j \in S \mid R_j \cap (O \setminus O_j) \neq \emptyset \}$. Intuitively, if $j \in \bS$ then there exists a job $i \in O$ from which $j$ might borrow but $\by_j < \by_i$. Therefore, the jobs in $\bS$ are of particular interest for our analysis as they separate it from the special case of~\Cref{sec:matching:arguments}: If $\bS = \emptyset$ for all $S \in \cS$ then we can use the analysis of~\Cref{sec:matching:arguments} by plugging in the truncated progress instead of the actual progress.
If $\bS \neq \emptyset$, then let $j^*(S)$ be the job in $\bS$ of minimum release date. Note that $j^*(S) \in S$. 

\begin{lemma}\label{lemma:segments:order}
  There is an order $S_1,\ldots,S_L$ of all segments $\cS$ such that $O_{S_1} \subsetneq \ldots \subsetneq O_{S_L}$.
\end{lemma}

\begin{proof}
  Let $i_1, \ldots, i_{|O|}$ denote the elements of $O$ ordered by non-decreasing $\by_i$, i.e., $\ell \le \ell'$ implies $\by_{i_\ell} \le \by_{i_{\ell'}}$. For each $j \in N \setminus O$, the set $O_j$ is a possibly empty prefix of $i_1, \ldots, i_{|O|}$. This implies that there must exist an order $S_1,\ldots,S_L$ of the segments $\cS$ such that $O_{S_1} \subseteq \ldots \subseteq O_{S_L}$. We conclude the proof by observing that the definition of a segment implies that if $O_{S_\ell} = O_{S_{\ell'}}$ then $S_\ell = S_{\ell'}$.
\end{proof}

As an immediate corollary of \Cref{lemma:segments:order}, we can bound the number of segments by $|O|$.

\begin{corollary}\label{coro:number-of-segments}
  $|\cS| \leq |O| + 1$.
\end{corollary}

For the remainder of this section, let $S_1, \ldots, S_L$ be the segments in $\cS$ indexed as in \Cref{lemma:segments:order}. For a segment $S$, we use the shorthand notation
\[
  S' \coloneq \begin{cases}
    S \setminus \{j^*(S)\} & \text{ if } \bS \neq \emptyset \\
    S & \text{ otherwise.}
  \end{cases}
\]
Clearly, $|S| \geq |S'| \geq |S| - 1$ for any segment $S$. 

The following two lemmas are the main ingredients to the proof of \Cref{thm:local:nonclairvoyant}. They bound the maximum amount of work that jobs in certain segments can borrow from certain jobs in $O$. While the first lemma is essentially just a corollary of \Cref{thm:main:borrowing}, proving the second lemma will require a stronger structural understanding of the borrow graph $G_B$ and flow network $G_F$ regarding bottlenecks, and thus, is deferred to the subsequent sections.
Recall that $L = |\cS|$.

\begin{lemma}[Good borrowing]
  \label{lemma:segment-good-stealing}
  For every $\ell \in \{1,\ldots,L\}$, we have 
  \( \sum_{k=1}^{\ell} \sum_{j \in S_k} \sum_{i \in O_{S_\ell}} \beta(j,i) \leq \sum_{i \in O_{S_{\ell}}} y_i(t).
  \)
\end{lemma}

\begin{proof}
By \Cref{thm:main:borrowing}, it holds that $\sum_{k=1}^\ell \sum_{j \in S_k} \beta(j,i) \leq y_i(t)$ for every $i \in O_{S_\ell}$. Summing over all $i \in O_{S_\ell}$ yields the stated bound.
\end{proof}

\begin{lemma}[Bottleneck borrowing]
  \label{lem:segment:bottleneck:capacity}
  Let $S$ be a segment such that $\bS \neq \emptyset$. Then,
  it holds that
  \(
    \sum_{j \in S} \sum_{i \in O \setminus O_S} \beta(j,i) \leq p_{j^*(S)}(t).
  \)
\end{lemma}

\Cref{lem:segment:bottleneck:capacity} ties into the intuition given in~\Cref{sec:matching:arguments} that jobs $j \in N\setminus O$ can only borrow from jobs $i \in O$ with $\by_i > \by_j$ via certain bottlenecks. We will prove the lemma in the subsequent sections by first showing that there is a \emph{single} bottleneck that needs to be part of \emph{every} path in the borrow graph from \emph{any} $j \in S$ to \emph{any} $i \in O\setminus O_S$, and second proving that the capacity of this bottleneck is at most $p_{j^*(S)}(t)$.

We can compose \Cref{lemma:segment-good-stealing} and \Cref{lem:segment:bottleneck:capacity}
to the following lemma.

\begin{lemma}\label{lemma:segment-stealing}
  For every $\ell \in \{1,\ldots,L\}$, we have 
  \[ \sum_{k=1}^{\ell} \sum_{j \in S_k} \sum_{i \in O} \beta(j,i) \leq \sum_{i \in O_{S_{\ell}}} y_i(t) + \sum_{k=1}^{\ell} p_{j^*(S_k)}(t) \cdot \ind[\bS_k \neq \emptyset] \ .
  \]
\end{lemma}

Using this, we can prove the following technical lemma, which bounds the number of jobs in every prefix of the set of ordered segments. The proof of \Cref{thm:local:nonclairvoyant} is then essentially a corollary of it.

\begin{lemma}
  \label{lemma:segment-inductive-matching}
  For each $\ell \in \{1,\ldots, L \}$, it holds that $\sum_{k = 1}^\ell |S_{k}'| \le \frac{1}{1-\alpha} |O_{S_{\ell}}|$.
\end{lemma}

\begin{proof}
Recall that we assume that $\frac{1}{1-\alpha}$ is integer. %
We will first split every job $i \in O$ into $\frac{1}{1-\alpha}$ \emph{slots}, and each slot receives a $1-\alpha$ fraction of $y_i(t)$ as progress at time $t$. Then, we will prove for each $\ell \in \{1,\ldots, L \}$ that there exists a matching that matches every job in $\sum_{k = 1}^\ell |S_{k}'|$ to some slot, which implies the lemma.

To this end, we first introduce $\frac{1}{1-\alpha}$ slots $\hO(i)$ for every $i \in O$. Let $\hO \coloneq \bigcup_{i \in O} \hO(i)$ be the total set of slots. For every $j \in N \setminus O$, we define $\hO_j \coloneq \bigcup_{i \in O_j} \hO(i)$ and for every segment $S$, we define $\hO_S \coloneq \bigcup_{i \in O_S} \hO(i)$. Note that $\frac{1}{1-\alpha}|O_j| = |\hO_j|$, $\frac{1}{1-\alpha}|O_S| = |\hO_S|$, and $\frac{1}{1-\alpha}|O| = |\hO|$. Further, we define $\hy_{i'}(t) \coloneq (1-\alpha)y_i(t)$ for every $i' \in O(i)$ and $i \in O$. 
Thus, for every segment $S$, every $j \in S$ and $i \in \hO_j$, we have $\by_j \geq \by_i$, and further,
\begin{equation}
  p_j(t) \geq \frac{1-\alpha}{\alpha} y_j(t)
    = \frac{1-\alpha}{\alpha} \by_j
    \geq \frac{1-\alpha}{\alpha} \by_i
    \geq \frac{1-\alpha}{\alpha} y_i(t) \geq \hy_i(t)  \ . \label{eq:slot-y-ineq}
\end{equation}
We now prove by induction that, for every $\ell \in \{1,\ldots, L \}$, it holds that $\sum_{k = 1}^\ell |S_{k}'| \le |\hO_{S_{\ell}}|$. This will conclude the lemma.

We start with $\ell = 1$.
\Cref{thm:main:borrowing} and \Cref{lemma:segment-stealing} imply 
\begin{align*} 
  \sum_{j \in S_1} p_j(t) 
  = \sum_{j \in S_1 }\sum_{i \in O} \beta(j,i) 
  &\leq p_{j^*(S_1)}(t) \cdot \ind[\bS_1 \neq \emptyset] + \sum_{i \in O_{S_1}} y_i(t) \\
  &= p_{j^*(S_1)}(t) \cdot \ind[\bS_1 \neq \emptyset] + \sum_{i \in \hO_{S_1}} \hy_i(t) \ ,
\end{align*}
which implies $\sum_{j \in S'_1} p_j(t) \leq \sum_{i \in \hO_{S_1}} \hy_i(t)$. Now, we can use~\eqref{eq:slot-y-ineq} to conclude that every term in the left sum of the inequality is at least as large as every term in the right sum. Thus, $|S'_1| \leq |\hO_{S_1}|$.

For the induction step, fix $\ell \geq 2$, and assume that for every $\ell' \in \{1,\ldots,\ell-1\}$, 
we have $\sum_{k=1}^{\ell'} |S'_{k}| \leq |\hO_{S_{\ell'}}|$. Also recall that $\hO_j = \hO_{j'}$ for all $j,j' \in S$ for every segment $S$. 
Thus, for every $Z \subseteq \bigcup_{k=1}^{\ell-1} S'_k$, we have $|Z| \leq |\bigcup_{j \in Z} \hO_j|$. 
By Hall's Theorem, this means that for each $j \in \bigcup_{k=1}^{\ell-1}S'_k$, there exists a slot $m(j) \in \hO_j  \subseteq \hO_{S_\ell}$ such that for every slot $i \in \hO_{S_\ell}$ it holds that $|\{j \mid m(j)=i \}| \leq 1$, that is, every slot is matched at most once.
Let $M \coloneq \{m(j) \mid j \in \cup_{k=1}^{\ell-1}S'_k\}$ be the set of matched slots in $\hO_{S_\ell}$.
\Cref{thm:main:borrowing} and \Cref{lemma:segment-stealing} imply 
\[
  \sum_{k=1}^\ell \sum_{j \in S_k} p_j(t) 
  = \sum_{k=1}^\ell \sum_{j \in S_k} \sum_{i \in O} \beta(j,i) 
  \leq \sum_{i \in O_{S_\ell}} y_i(t) + \sum_{k=1}^\ell p_{j^*(S_k)}(t) \cdot \ind[\bS_k \neq \emptyset] \ ,
\]
giving
\[
  \sum_{j \in S'_\ell} p_j(t) + \sum_{k=1}^{\ell-1} \sum_{j \in S'_k} p_j(t)  
  \leq \sum_{i \in O_{S_\ell}} y_i(t) 
  = \sum_{i \in \hO_{S_\ell}} \hy_i(t) 
  = \sum_{i \in M} \hy_i(t) + \sum_{i \in \hO_{S_\ell} \setminus M} \hy_i(t) \ .
\]
The definition of $M$ and~\eqref{eq:slot-y-ineq} yield $\sum_{i \in M} \hy_i(t) \leq \sum_{k=1}^{\ell-1} \sum_{j \in S'_k} p_j(t)$. Thus, the above implies
\[
  \sum_{j \in S'_\ell} p_j(t) \leq \sum_{i \in \hO_{S_\ell} \setminus M} \hy_i(t) \ .
\]
By~\eqref{eq:slot-y-ineq}, the above implies $|S'_\ell| \leq |\hO_{S_\ell} \setminus M|$. Since $|M| = |\bigcup_{k=1}^{\ell-1} S'_k|$ and segments are pairwise disjoint, we conclude that 
\[
  \bigg| \bigcup_{k=1}^{\ell} S'_k \bigg| = \bigg| \bigcup_{k=1}^{\ell-1} S'_k \bigg| + |S'_\ell| \le |\hO_{S_\ell} \setminus M| + |M| \leq |\hO_{S_\ell}| \ ,
\]
which concludes the proof of the statement.
\end{proof}

Finally, we can prove \Cref{thm:local:nonclairvoyant}.

\begin{proof}[Proof of \Cref{thm:local:nonclairvoyant}]
  \Cref{lemma:segment-inductive-matching} gives \(
    \sum_{k = 1}^L |S_{k}'| \le \frac{1}{1-\alpha} |O_{S_{L}}| \leq \frac{1}{1-\alpha} |O(t)|.
  \)
  Moreover, for every segment $S$, we have $|S| \leq |S'| + 1$. Thus, by combining the above with \Cref{coro:number-of-segments}, we obtain
  \[
    |N(t) \setminus O(t)| 
    = \sum_{k=1}^L |S_k| 
    \leq |\cS| + \sum_{k=1}^L |S'_k|
    \leq 1 + |O(t)| + \frac{1}{1-\alpha} |O(t)| \leq \bigg(2 + \frac{1}{1-\alpha} \bigg) \cdot |O(t)| \ .
  \]
  This concludes the proof of the theorem.
\end{proof}

\subsection{Outline: Proof of~\Cref{lem:segment:bottleneck:capacity}}
\label{sec:lemma:outline}

To complete the proof of~\Cref{thm:local:nonclairvoyant}, it remains to prove~\Cref{lem:segment:bottleneck:capacity}. The proof requires multiple ingredients that are shown in the subsequent sections. Here, we give a brief outline of the proof idea.

The idea of our proof is motivated by the following lemma, which we prove in~\Cref{sec:non-clairvoyant:paths}.
\begin{lemma}
  \label{lem:nc-borrowing}
  If a job $j \in A(t)$ can reach a job $i$ via an $E_N$-path, then $\by_j \ge \by_i$.
\end{lemma}

With respect to the segments (cf.~\Cref{def:segment}), this means that if a job $j \in S$ for a segment $S$ can reach a job $i \in O$ via an $E_N$-path, then $\by_j \ge \by_i$, and thus, $i \in O_S$. Vice versa, a job $i \in R_j \setminus O_S$ \emph{cannot} be reachable by $j$ via an $E_N$-path.

Since~\Cref{lem:segment:bottleneck:capacity} considers exactly such jobs $j \in \bS$ and $i\in O\setminus O_S$ that satisfy $\by_i > \by_j$, we now know that there cannot be an $E_N$-path between any such pair $j$ and $i$. In~\Cref{sec:structure}, we exploit this insight and further analyze the structure of the borrow graph, specifically in absence of $E_N$-paths between jobs $j \in N(t)$ and $i \in O(t)$. To this end, we introduce the concept of \emph{bottlenecks}, which are jobs that cannot be avoided on paths between jobs $j \in N(t)$ and jobs $i \in O(t)$ with $\by_i > \by_j$. 

In~\Cref{sec:fine-grained}, we show that for a fixed $j \in \bS$ there is one unique bottleneck $b$ that is the first bottleneck on \emph{any} path from $j$ to \emph{any} $i \in R_j\setminus O_S$. This means that if $j$ wants to borrow from any job in $i \in R_j\setminus O_S$, then this must be via $b$. It is possible to show that the maximum amount of processing that $j$ can borrow via $b$ is $p_j(t)$, which implies $\sum_{i \in O\setminus O_S} \beta(j,i) \le p_j(t)$. If $\bS$ contains only a single job, then this already gives us~\Cref{lem:segment:bottleneck:capacity}.

Since $\bS$ can contain multiple jobs, we show in~\Cref{sec:segment:borrowing} that the bottleneck $b_{j^*(S)}$ is part of \emph{every} path from \emph{any} $j \in \bS$ to any $i \in O\setminus O_S$. To finish the proof, we then prove that \emph{all} jobs in $\bS$ together can borrow at most $p_{j^*(S)}$ units of processing via $b_{j^*(S)}$, which will imply the lemma.

\subsection{The Structure of the Borrow Graph}
\label{sec:structure}

In order to prove \Cref{lem:segment:bottleneck:capacity}, we analyze the structure of the borrow graph.
We do this with respect to a fixed job $i \in O(t)$. Later, we will see that the relevant structures are essentially independent of the choice of $i$.
Thus, fix some $i \in O(t)$. 
In this section, we always refer to the borrow graph when talking about reachability and paths.
We start with the definition of \emph{layers}. 

\begin{definition}[layer]\label{def:layers}
  The \emph{layers} $L_0^i,\ldots,L_{\ell_i}^i$ of a job $i$ are defined as follows:
  \begin{enumerate}[(i)]
      \item $L_0^i$ is the set of jobs that can reach $i$ with an $E_N$-path.
      \item $L_x^i$ for $x \geq 1$ is the set of jobs in $J \setminus \left(\bigcup_{0 \le y \leq x-1} L_y^i\right)$ that can reach some job in $L_{x-1}^i$ via a path in $G_B$ that ends with an edge in $E_C$ but otherwise only uses edges in $E_N$.
  \end{enumerate}
  We write $L_{\leq x}^i \coloneq \bigcup_{0 \le y \le x} L_y^i$ and $L_{\geq x}^i \coloneq \bigcup_{x \le y \le \ell_i} L_y^i$.
\end{definition}
In general, the layer $L_x^i$ is defined to contain all jobs that can reach $i$ in $G_B$ by using $x$ edges from $E_C$ and cannot reach $i$ in $G_B$ by using less than $x$ edges from $E_C$. As outlined in~\Cref{sec:lemma:outline}, \Cref{lem:segment:bottleneck:capacity} is concerned with jobs $j \in N(t)\setminus O(t)$ such that $i \in R_j$ and $\by_i > \by_j$. By \Cref{lem:nc-borrowing}, such jobs must satisfy $j \in L^i_{\ge 1}$.

Next, we define edges and jobs (vertices) that allow paths from higher layers to lower layers; specifically from a job in $L_{x+1}^i$ to a job in $L_x^i$.

\begin{definition}[$\delta$-edges]
We define $\delta_x^i \coloneq \{ (v,u) \in E_C\mid v \in L_{x+1}^i \land u \in L_x^i \}$.
\end{definition}

\begin{definition}[entrypoint]
     A job $k \in L_x^i$ is an \emph{entrypoint of layer $L_x^i$} if both of the following two conditions hold:
    \begin{enumerate}[(i)]
        \item $(j',k) \in \delta_x^i$ for some $j' \in L_{x+1}^i$, and
        \item there exists a path $P$ in $G_B$ from some $j \in A(t) \cap L_{\geq x+1}^i$ to $k$ that, apart from $k$, only visits jobs in $L_{\geq x+1}^i$.
    \end{enumerate}
    We use $H_x^i \subseteq L_x^i$ to denote the set of all entrypoints of layer $L_x^i$. 
\end{definition}

Note that the second condition of the definition of entrypoints ensures that we only consider entrypoints that actually allow jobs in $ A(t) \cap L_{\geq x+1}^i$ to reach layer $L_x^i$. This is because these are the relevant jobs for~\Cref{lem:segment:bottleneck:capacity}. %

The following definition describes paths that go from a higher layer to a lower layer.

\begin{definition}[$x$-entry path]
  We call a path $P$ in $G_B$ \emph{$x$-entry path} if $P$ starts at some job $j \in L_{\geq x+1}^i$, ends at some entrypoint $k \in H_x^i$, and, apart from $k$, only visits jobs in $L_{\geq x+1}^i$.
\end{definition}

Having established layers and entrypoints with respect to $i \in O(t)$, we continue with the definition of \emph{bottlenecks}, which are special entrypoints. As the name suggests, a bottleneck $k$ is an entrypoint if there is some job $j \in A(t)$ that can \emph{only} borrow from $i$ via transferring work through $k$. Formally, we have the following definition.

\begin{definition}
  An entrypoint $k \in H_x^i$ is a \emph{bottleneck} of layer $L_x^i$ if 
  there exists a job $j \in (A(t) \cap L_{x+1}^i) \cup H_{x+1}^i$ that can \emph{only} reach layer $L_x^i$ via paths in $G_B$ that contain $k$.   
\end{definition}

We emphasize that a layer may have no bottleneck. 
On the other side, we show that every layer has also at most one bottleneck.
We will use $b_x^i$ to refer to the bottleneck of layer $L_x^i$ if it has one.

\begin{restatable}{lemma}{uniqueBottleneck}
  \label{thm:bottleneck:uniqueness}
  Every layer of job $i$ has at most one bottleneck.
\end{restatable}

Next, we consider situations of the following type: Assume $L_x^i$ and $L_{x+1}^i$ both have bottlenecks $k_x$ and $k_{x+1}$. From what we have shown so far, it is not clear that a job $j \in A(t) \cap L^i_{x+2}$ that can reach jobs in $L_{x+1}^i$ only via paths that contain the bottleneck $k_{x+1}$ can only reach jobs in $L_{x}^i$ only via paths that contain the bottleneck $k_x$. We show that this is indeed the case. %

\begin{restatable}{lemma}{layeredBottlenecks}
  \label{coro:multiple:bottlenecks}
  Consider a layer $L_x^i$ and let $j \in A(t) \cap L_{x+1}^i$ be a job that only can reach $L_x^i$ through the bottleneck of $L_x^i$. Then, $j$ can only reach $i$ via paths that contain \emph{all} bottlenecks of the layers $L_y^i$ with $y \leq x$.
\end{restatable}

The next lemma and its corollary show that jobs $j$ that belong to a segment $S$ and satisfy $i \in R_j \setminus O_S$ can \emph{only} reach $i$ via paths that contain at least one bottleneck. We prove the lemma in~\Cref{app:no-bottleneck-borrowing}.

\begin{restatable}{lemma}{bottleneckFreePaths}
  \label{thm:yinequality:no:bottlenecks}
  If there exists at least one path from $j \in A(t)$ to some job $i$ in $G_B$ that does not contain any bottleneck, then $\by_j \ge \by_i$.
\end{restatable}

\begin{corollary}
  \label{coro:yinequality:no:bottlenecks}
  Let $j \in S$ for a segment $S \in \mathcal{S}$ and let $i \in R_j \cap O$.
  If $j$ can reach $i$ without visiting any bottleneck, then $i \in O_S$. Consequently, if $i \not\in O_S$, then any path from $j$ to $i$ must contain at least one bottleneck.
\end{corollary}

\subsection{Fine-Grained Borrowing}
\label{sec:fine-grained}

In the previous section, we showed that a path from a job $j \in S$ for a segment $S$ to a job $i \in O(t) \setminus O_S$ must contain at least one bottleneck of $i$. In this section, we generalize this statement and show that there is \emph{one} bottleneck that is part of \emph{every} path from $j$ to \emph{any} $i \in O(t) \setminus O_S$.

To this end, we first introduce some more notation.
Fix a job $i \in O(t)$ and consider a job $j \in A(t)$ with $i \in R_j$. We use $\ell_{ji}$ to denote the index of the layer $L_{x}^i$ with $j \in L_{x}^i$. If every path from $j$ to $i$ contains at least one bottleneck of $i$, then we use $\sigma_{ji}$ to denote the index of the first layer that $j$ cannot reach without visiting a bottleneck, i.e., every path from $j$ to some $i' \in R_j \cap L^i_{\le \sigma_{ji}}$ visits at least one bottleneck of $i$ and for every layer $L_x^i$ with $\sigma_{ji} < x \le \ell_{ji}$ there exists an entrypoint $k \in H_x^i$ that $j$ can reach using a path without bottleneck of $i$. We use $b_{ji}$ 
to refer to the unique bottleneck (cf.\ \Cref{thm:bottleneck:uniqueness}) of layer $L_{\sigma_{ji}}^i$. Intuitively, $b_{ji}$ is the first bottleneck to appear on any path from $j$ to~$i$.

The next theorem and its corollary show that the first bottlenecks on paths from $j \in S$ to any $i \in R_j \setminus O_S$ are the exact same job.

\begin{theorem}
    \label{thm:unique:earliest:bottleneck}
    Let $j \in A(t)$ and let $i_1,i_2$ be two distinct jobs that $j$ can only reach via paths that contain at least one bottleneck. Then, $b_{ji_1} = b_{ji_2}$ or $\by_j \ge \min\{ \by_{i_1}, \by_{i_2}\}$.
\end{theorem}

\begin{corollary}
  \label{coro:unique:earliest:bottleneck}
  Let $j \in S$ for a segment $S$ and $i_1,i_2 \in R_j \setminus O_S$, then $b_{ji_1} = b_{ji_2}$.
\end{corollary}

To prove the theorem, we will need the following auxiliary lemma. The lemma also will be helpful later to upper bound the amount of processing that $j$ can borrow from jobs in $i$ with $\by_i > \by_j$.

\begin{restatable}{lemma}{stealingViaBottlenecks}
  \label{lem:stealing:via:bottlenecks:2}
  Let $j \in A(t)$ be such that on every path from $j$ to $i$ is at least one bottleneck and let $\Delta= \{d_1,\ldots,d_r\} = \{d \in J \mid (d,b_{ji}) \in \delta^i_{\sigma_{ji}}\}$. For each $d_q \in \Delta$, let $\tau_q$ denote the earliest point in time during $I_{d_q}$ at which $b_{ji}$ is processed. 
  If $j$ can reach $d_q \in \Delta$ via an $L^i_{\geq \sigma_{ji}+1}$-path, then $\frac{1-\alpha}{\alpha} \by_j \ge p_{b_{ji}}(\tau_q)$.
\end{restatable}

In order to show the theorem, we also rely on the following two observations.

\begin{observation}
    \label{obs:unique:earliest:bottleneck:1}
    Let $j \in A(t)$ be such that $j$ can reach some job $i$ only via paths that contain at least one bottleneck. Let $k$ be such a bottleneck. Then, there is no path from $j$ to any other job $j'$ that is alive at $s_k^-$ that does not contain $k$.
\end{observation}

\begin{proof}
    Since $k$ is a bottleneck of some layer $L_x^i$ with $k \in L_x^i$ and $j$ cannot reach $i$ without visiting $k$, it must be the case that any path from $j$ to some job in $L_{\leq x}^i$ contains $k$ (cf.~\Cref{coro:multiple:bottlenecks}). 
    For the sake of contradiction, assume that $j$ can reach a job $j'$ that is alive at $s_k^-$ via a path that does not contain $k$. Thus, we must have $j' \in L_{\geq x+1}^i$. Moreover, as $j'$ is alive at $s_k^-$, we have $(j',k) \in E_N$. However, $(j',k) \in E_N$ and $k \in L_x^i$ imply $j' \in L_{\leq x}^i$; a contradiction.
\end{proof}

Recall that for a job $j$, we use $t_j$ to denote the latest point in time $t'$ with $t'\le t$ at which $j$ is being processed.

\begin{observation}
    \label{obs:unique:earliest:bottleneck:2}
    Let $j \in N(t)$ be such that $j$ can reach some job $i$ only via paths that contain at least one bottleneck. 
    Let $k$ be such a bottleneck.
    Then, there is no path from $j$ to any job $j'$ with $t_{j'} \le s_k^-$ that does not contain $k$.
\end{observation}

\begin{proof}
    Via proof by contradiction. Assume $j$ can reach $j'$ with $t_{j'} \le s_k^-$ via a path $P$ that does not contain $k$. By \Cref{obs:unique:earliest:bottleneck:1}, the job $j'$ cannot be alive at $s_k^-$ and, therefore, must satisfy $C_{j'} \le s_k^-$.

    Observe that $r_j \ge s_k$, as otherwise $j \in N(t)$ implies that $(j,k) \in E_N$, which would contradict that $k$ is a bottleneck for $j$ with respect to $i$.
    Together with $C_{j'}<s_k$, we have $s_k \in I(P)$. However, this implies that at least one $d \in P$ is alive at $s_k$ and can be reached by $j$ via a path that does not contain $k$; a contradiction to \Cref{obs:unique:earliest:bottleneck:1}.
\end{proof}

Having the two observations in place, we are ready to show \Cref{thm:unique:earliest:bottleneck}.

\begin{proof}[Proof of \Cref{thm:unique:earliest:bottleneck}]
  To simplify the notation, let $b_1 \coloneq b_{ji_1}$ and $b_2 \coloneq b_{ji_2}$. Assume without loss of generality that $s_{b_1} \le s_{b_2}$.
  We show the theorem by assuming $b_1 \neq b_2$ and proving that this implies $\by_j \geq \min\{\by_{i_1},\by_{i_2}\}$.

  First, observe that we must have $s_{b_2} \le t_{b_1}$, as otherwise 
  \Cref{obs:unique:earliest:bottleneck:2} gives that every path from $j$ to $b_1$ must contain $b_2$; a contradiction to $b_1$ being the first bottleneck 
  on every path from $j$ to $i_1$.
  In the same way, we can argue that $r_{b_2} \ge s_{b_1}$. 
  Otherwise, \Cref{obs:unique:earliest:bottleneck:1} gives that $b_2$ would be alive at $s_{b_1}^-$ and every path from $j$ to $b_2$ must contain $b_1$.
  Thus, $b_2$ is released after $s_{b_1}$ and emits before $b_1$ is completed. \Cref{obs:clairvoyant-jobs-block-earlier-jobs} implies that $b_2$ completes before $b_1$. Combining these observations, we get $I_{b_2} \subseteq I_{b_1}$.

  We claim that $I_{b_2} \subseteq I_{b_1}$ implies $b_2 = i_2$. To see this, assume otherwise, i.e., $b_2 \not= i_2$. Then, there exists a $b_2$-$i_2$-path $P'$. Let $d$ be the direct successor of $b_2$ on this path. This means that $d$ is being processed during $I_{b_2} \subseteq I_{b_1}$. %
  This implies that there exists a $b_1$-$i_2$-path that does not contain $b_2$. Since there also is a $j$-$b_1$-path that does not contain $b_2$ because $b_1$ is the first bottleneck on every path from $j$ to $i_1$, there must exist a $j$-$i_2$-path that does not contain $b_2$; a contradiction.

    Having established $b_2 = i_2$, we finish the proof of the theorem by showing $\by_j \ge \by_{i_2}$. Recall that we already showed that $i_2 = b_2$ satisfies $r_{i_2} \ge s_{b_1}$ and $C_{i_2} = C_{b_2} \le C_{b_1}$. By definition of our algorithm, this implies that $b_1$ is not being processed during $I_{i_2}$. 
    Let $t'$ be the earliest point in time with $t' \ge C_{i_2}$ at which $b_1$ is executed again. 
    Note that this point in time must exist, as otherwise~$t_{b_1} \le s_{b_2}^-$ and, by \Cref{obs:unique:earliest:bottleneck:1}, every path from $j$ to $b_1$ must contain $b_2$. This is a contradiction to $b_1 \neq b_2$ being the first bottleneck on any path from $j$ to $i_1$.
    Next, observe that
    \[  
      p_{b_1}(t') \ge (1-\alpha) \cdot p_{i_2} \ .
    \]
    To see this, note that we must have $(1-\alpha) p_{i_2}= p_{i_2}(s_{i_2}) \le p_{b_1}(s_{i_2})$ since the algorithm finishes $b_2=i_2$ before $b_1$.
    Furthermore, $b_1$ is not being processed between $s_{i_2}$ and $t'$, so we get
    \[
    (1-\alpha) p_{i_2}= p_{i_2}(s_{i_2}) \le p_{b_1}(s_{i_2}) = p_{b_1}(t') \ .
    \]

    By assumption, there exists a path $P'$ from $j$ to $i_2$ that does not contain $i$. We distinguish between two cases: (1) $r_j \le t'$ and (2) $r_j > t'$.

    \paragraph*{Case (1):} Assume  $r_j \le t'$, this means that job $b_1 \in C(t')$ with $p_{b_1}(t') \ge (1-\alpha) p_{i_2}$ is processed at point in time $t' \in I_j$. In the appendix, we show that this implies $\by_j \ge \by_{i_2}$ (\Cref{obs:ec:cut:aux:1}).

    \paragraph*{Case (2):} Assume  $r_j > t'$ and recall that $j$ can reach $b_2 = i_2$ via a path $P$ that does not contain $b_1$. This implies that $P$ only uses jobs in $L_{\geq \sigma_{ji_1} + 1}^{i_1}$. Since $P$ is a path from $j$ to $i_2$, and we have $r_j > t'$ and $C_{i_2} \le t'$, the path must contain a job $d$ that is alive at point in time $t'$. Furthermore, as $P$ only uses jobs in $L_{\geq \sigma_{ji_1} + 1}^{i_1}$, we have $r_d \ge s_{b_1}$ by \Cref{obs:unique:earliest:bottleneck:1}.

    The facts that $r_d \ge s_{b_1}$ and $d \in J_A(t')$ imply $(d,b_1) \in E_C$, and thus, $(d,b_1) \in \delta_{\sigma_{ji_1}}^i$. Since $j$ can reach $d$ via a path that only contains jobs of $L_{\geq \sigma_{ji_1}+1}^{i_1}$, \Cref{lem:stealing:via:bottlenecks:2} implies 
    \[
    \frac{1-\alpha}{\alpha} \by_j \ge p_{b_1}(\tau) \ge p_{b_1}(t') \ge (1-\alpha) \cdot p_{i_2},
    \]
    where $\tau$ is the earliest point in time during $I_d$ at which $b_1$ is being processed.
    By rearranging this inequality, we get 
    \(
        \by_j \ge \alpha \cdot p_{i_2} \geq \by_{i_2}.
    \)
\end{proof}

\subsection{Segment Borrowing and the Proof of \Cref{lem:segment:bottleneck:capacity}}
\label{sec:segment:borrowing}

In the previous section, we showed for a fixed job $j \in S$ and a segment $S$ that there is a bottleneck that is part of every path from $j$ to some $i \in R_j\setminus O_S$. The following lemma generalizes this statement by showing that there is a bottleneck that is part of every path from \emph{any} $j \in S$ to \emph{any} $i \in R_j\setminus O_S$. 
Recall that $j^*(S)$ is a job with the smallest release date in $S$ .

\begin{lemma}
  \label{obs:segment:bottleneck}
  Let $S \in \mathcal{S}$ be a segment such that $\bS \neq \emptyset$. 
  Then, 
  for every $j \in \bS$ and every $i \in O \setminus O_S$, every path from $j$ to $i$ contains the bottleneck $b(S)$, which is the first bottleneck on every path from $j^*(S)$ to any $i \in O \setminus O_S$.
\end{lemma}
\begin{proof}
  First, note that for $j \in S$ and $i \in O \setminus O_S$, there is only a path from $j$ to $i$ if $i \in \bO_j$. Thus, in the following, it suffices to argue for a fixed $j \in S$ about jobs $i \in \bO_j$.

  Let $j^* \coloneq j^*(S)$.
  Recall that $j^*$ is the job of $\bS$ with minimum release date.
  By definition, every job $j \in \bS$ satisfies for every $i \in \bO_S$ that $\by_j < \by_i$ and $i \in R_j$, and thus, \Cref{coro:yinequality:no:bottlenecks} implies that every path from $j$ to $i$ contains at least one bottleneck.
  Moreover, \Cref{thm:unique:earliest:bottleneck} gives that $b_{j^*i_1} = b_{j^*i_2}$ for all (not necessarily distinct) $i_1,i_2 \in \bO_{j^*}$. Let $b(S)$ denote this first bottleneck of $j^*$ on every path from $j$ to any $i \in \bO_{j^*}$, that is, $b(S) = b_{j^*i}$ for every $i \in \bO_{j^*}$.
  Thus, for every $i \in \bO_{j^*}$, every path from $j^*$ to $i$ contains $b(S)$.

  To prove the statement, assume that there exists a path $P$ from $j \in \bS$ to  $i \in \bO_j$ that does not contain $b(S)$. 
  Since $r_j \ge r_{j^*}$ and $j,j^* \in A(t)$, we have $(j^*,j) \in E$ and can therefore extend $P$ to the path $P'$ from $j^*$ to $i$ that does not contain $b(S)$. 
  By~\Cref{coro:unique:earliest:bottleneck}, we have that the existence of $P'$ implies $i \in O_S$ as otherwise $b(S)$ would need to be the earliest bottleneck on $P'$. However, $i \in \bO_j$ implies $i \not\in O_S$, and we arrive at a contradiction.
\end{proof}

Finally, we can prove \Cref{lem:segment:bottleneck:capacity} by bounding the amount of processing that the jobs in $\bS$ can borrow via the bottleneck $b(S)$.

\begin{proof}[Proof of \Cref{lem:segment:bottleneck:capacity}] 
  Fix a segment $S$ such that $\bS \neq \emptyset$.
  Let $j^* \coloneq j^*(S)$.
  \Cref{obs:segment:bottleneck} guarantees that the bottleneck $b(S) = b_{j^*}$ exists.
  Thus, $b(S)$ is the earliest bottleneck that $j^*$ has to visit to reach any job in $O \setminus O_S$. 
  Fix an arbitrary $i \in O \setminus O_S$ such that $i \in R_{j^*}$, and assume that $b(S) \in H_x^i$ (recall that any bottleneck is an entrypoint).
  Since $b(S)$ is a bottleneck, $j^*$ can only reach jobs in $L_{\leq x}^i$ via paths that contain $b(S)$. 
  Thus, using that $b(S) \in L_{\leq x}^i$, any path $P$ from $j^*$ to some $O \setminus O_S$ must %
  contain a prefix $P'$ that starts at $j^*$, ends at $b(S)$, and otherwise only visits jobs in $L_{\geq x+1}^i$. Thus, $P'$ is an $x$-entry path.
  For a job $j \in S$, let $\Delta_j \coloneq \{(d, b(S)) \in E \mid (d, b(S)) \in \delta_{x_i}^i, i \in R_j \}$, where $x_i$ is the index of the layer of $i$ that contains $b(S)$.
  By definition of the layers, the path $P'$ must end with an edge $(d,b(S)) \in \Delta_{j^*}$. %
  By \Cref{lem:stealing:via:bottlenecks:2}, each such edge $(d,b(S))\in \Delta_{j^*}$ satisfies  $\frac{1-\alpha}{\alpha} \by_{j^*} \ge p_{b(S)}(\tau_d)$ for the earliest point in time $\tau_d$ during $I_d$ at which $b(S)$ is being processed. 
  
  Next, consider a job $j' \in S \setminus \{j^*\}$. By choice of $j^*$, we have $r_{j'} \ge r_{j^*}$. Since $j', j^* \in A(t)$, this implies that each job $d$ that is processed during $I_{j'}$ is also processed during $I_{j^*}$. Thus, each path $j'$-$i$-path $P$ in the borrow graph can be transformed into a $j^*$-$i$-path $\bP$ by just exchanging $j'$ with $j^*$. This implies that $\Delta_j \subseteq \Delta_{j^*}$.

  Let $D = \{ d \in J \mid (d, b(S)) \in \Delta_{j^*} \}$. Let $\tau$ be the earliest time at which $b(S)$ is being processed during $I_{d}$ for any $d \in D$. Note that such a time exists as $(d, b(S)) \in E$. Using our observation from above gives $\frac{1-\alpha}{\alpha} \by_{j^*} \ge p_{b(S)}(\tau)$.
  
  Let $T = \{\tau_1,\ldots,\tau_{\ell+1}\}$ as in \Cref{def:flow-network} of $G_F$.
  Let $\tau_{\ell'} \in T$ be the earliest time in $T$ such that there is an edge from some $d \in D$ to $v_{\ell'}^{b(S)}$ in $G_F$, which means that
  $[\tau_{\ell'},\tau_{\ell'+1}] \subseteq I_{d}$, and thus, $r_d \leq \tau_{\ell'}$.
  Combined with the definition of $T$, we conclude $\tau_{\ell'} = \min_{d \in D} r_d$.
  
  By \Cref{obs:segment:bottleneck}, every path from any $j \in S$ to any $i \in O \setminus O_S$ contains $b(S)$, and thus, an edge of $\Delta_{j^*}$. 
  Due to this and by our choice of $\tau_{\ell'}$, every path from $d \in D$ to $b(S)$ in $G_F$ contains one of the vertices $v_{\ell'}^{b(S)},\ldots,v_{\ell}^{b(S)}$ of $G_F$.
  Thus, \Cref{lemma:flow:positive-paths} implies that the flow in $f$ that is sent from any $j \in S$ to any $i \in O \setminus O_S$ in $G_F$ must use an edge from $(v_{\ell'}^{b(S)},b(S)),\ldots,(v_{\ell}^{b(S)},b(S))$ in $G_F$.
  Hence, using the construction of the $\beta$-values in \Cref{def:borrow:values},
  \[
  \sum_{j \in S} \sum_{i \in O \setminus O_S} \beta(j,i)
  \leq \sum_{k=\ell'}^{\ell} f(v_{k}^{b(S)},b(S)) 
  \leq \sum_{k=\ell'}^{\ell} q_{b(S)}([\tau_{k},\tau_{k+1}])
  \leq q_{b(S)}([\tau_{\ell'},\min\{t,C_{b(S)}\}]) = p_{b(S)}(\tau_{\ell'}) \ ,
  \]
  because $\tau_{\ell' + 1} \leq t$.

  Finally, note that $\tau \geq \tau_{\ell'}$, as $\tau$ is by definition the earliest time at which $b(S)$ is being processed during $I_d$ for any $d \in D$, and we observed that $\tau_{\ell'} = \min_{d \in D} r_d$. Specifically, the definition of $\tau$ gives that $p_{b(S)}(\tau_{\ell'}) = p_{b(S)}(\tau)$. The lemma then follows because $p_{b(S)}(\tau) \leq \frac{1-\alpha}{\alpha} \by_{j^*} \leq p_{j^*}(t)$, where the final inequality uses $j^* \in N(t)$.
\end{proof}

\section{Conclusion}

We have shown that any constant-factor lookahead on the job processing times is sufficient to break the deterministic lower bound for non-clairvoyant flow time minimization and to achieve a constant competitive ratio. Similar to the speed-augmented result by \textcite{KalyanasundaramP00}, our result illustrates that non-clairvoyant flow time scheduling is \enquote{almost} constant competitive, as making the algorithm slightly more powerful is sufficient to achieve constant competitiveness.

For future work, we suggest to study %
the problem in a setting where the signals are not completely precise, for example, using learning-augmented algorithm design. Finally, the $\alpha$-clairvoyant model naturally translates to scheduling problems with different objective functions, e.g., makespan minimization.

\printbibliography

\appendix

\section{Flow Problem: Missing Proofs of~\Cref{sec:flow-problem,sec:borrowing-via-flows}}

\obsflowincreaseprecision*

\begin{proof}
  If $t' \in T^S$, the statement is clearly true. Thus, assume that $t' \notin T^S$.
  Let $T^S = \{\tau_1,\ldots,\tau_{\ell+1} \}$ and let $\ell'$ such that $t' \in (\tau_{\ell'}, \tau_{\ell'+1})$. Let $j,i \in J$. If $f(j,i) = 0$, then we can clearly find $f'$ such that $f'(j,i) = 0$.
  Let $i \in J$. 
  Let $J' = \{j \in J \mid f(j,v_{\ell'}^i) > 0\}$.
  For every $j \in J'$ it holds that $[\tau_{\ell'},\tau_{\ell'+1}] \subseteq I_j$ and $\sum_{j \in J'} f(j,v_{\ell'}^i) = f(v_{\ell'}^i,i) \leq q_i^S([\tau_{\ell'},\tau_{\ell'+1}])$.
  Let $u_1^i$ and $u_2^i$ be the dummy vertex in $G_F^S$ w.r.t.\ $T^S \cup \{t'\}$ that corresponds to $[\tau_{\ell'}, t']$ and $[t',\tau_{\ell'+1}]$, respectively.
  Note that for every $j \in J'$, we have $[\tau_{\ell'}, t'],[t',\tau_{\ell'+1}] \subseteq I_j$. Thus, we can find a feasible flow $f'$ for $G_F^S$ w.r.t.\ $T^S \cup \{t'\}$ such that, for every $j \in J'$, $f'(j, u_1^i) + f'(j, u_2^i) = f(j,v_{\ell'}^i)$, $f'(u_1^i,i) + f'(u_2^i,i) = f(v_{\ell'}^i,i)$, $f'(u_1^i,i) \leq q_i^S([\tau_{\ell'}, t'])$, $f'(u_2^i,i) \leq q_i^S([t', \tau_{\ell'+1}])$, and that equals $f$ for all other edges. This implies that $f(j,i) = f'(j,i)$ for every $j \in J'$. Repeating this argument for every $i \in J$ completes the proof.
\end{proof}

\lemmaflowborrowapplyedge*

\begin{proof}
  Let $T = \{\tau_1,\ldots,\tau_{\ell+1}\}$ be as in \Cref{def:flow-network}. 
  Since $f(j,i) > 0$, there exists a smallest index $\ell'$ such that $f(j,v^i_{\ell'}) > 0$. Let $\gamma \coloneq f(j,v^i_{\ell'})$ and note that $\gamma > 0$.
  Since $\sum_{i' \in J} f(i',j) = 0$, the flow $f(j,i)$ must completely be covered by the supply of $j$, and thus, $p_j^S(t) \geq f(j,i) \geq \gamma$.
  Further, as $f(j,v^i_{\ell'}) > 0$, by the definition of $G_F^S$, we have %
  $[\tau_{\ell'},\tau_{\ell'+1}] \subseteq I_j$.
  Moreover, by flow conservation and since $v_{\ell'}^i$ has only one outgoing edge, it must hold that $f(v^i_{\ell'},i) \geq f(j,v^i_{\ell'}) = \gamma$. By the capacity constraint, we additionally get $f(v^i_{\ell'},i) \leq q^S_i([\tau_{\ell'},\tau_{\ell'+1}])$. 
  We conclude $q^S_i([\tau_{\ell'},\tau_{\ell'+1}]) \geq \gamma$.
  We construct the schedule $S'$ as follows. 
  First, we copy $S$ and then only modify it during time $[\tau_{\ell'},\tau_{\ell'+1}]$. 
  During $[\tau_{\ell'},\tau_{\ell'+1}]$, we replace the first $\gamma$ units of work on $i$, which exists as $q^S_i([\tau_{\ell'},\tau_{\ell'+1}]) \geq \gamma$, with $\gamma$ units of work on $j$, which is possible as $p_j^S(t) \geq \gamma$ and $r_j \leq \tau_{\ell'}$. It is clear from the construction that $S'$ satisfies properties (a) and (b). 
  Note that due to this modification, job $j$ might complete before time $t$ in $S'$, and job $i$ becomes unfinished at time $t$ in $S'$. Thus, we might have $T^{S} \neq T^{S'}$. However, using \Cref{obs:flow-increase-precision}, we can assume that  $f$ is feasible for $G_F^S$ w.r.t.\ $T^{S} \cup T^{S'}$, and that any feasible flow for $G_F^{S'}$ w.r.t.\ $T^{S'}$ is feasible for $G_F^{S'}$ w.r.t.\ $T^{S} \cup T^{S'}$. Since $G_F^S$ w.r.t.\ $T^{S} \cup T^{S'}$ is isomorphic to $G_F^S$ w.r.t.\ $T^{S} \cup T^{S'}$, we can w.l.o.g.\ assume that $V(G_F^{S'}) = V(G_F^{S})$ and $E(G_F^{S'}) = E(G_F^{S})$.

  Under this assumption, we construct the flow $f'$ for $G_F^{S'}$ as follows. We set $f'(u,v) \coloneq f'(u,v)$ for all $(u,v) \in E(G^{S'}_F) \setminus \{(j,v_{\ell'}^i),(v_{\ell'}^i,i)\}$, $f'(j,v_{\ell'}^i) \coloneq f(j,v_{\ell'}^i) - \gamma = 0$, and $f'(v_{\ell'}^i,i) \coloneq f(v_{\ell'}^i,i) - \gamma$. 
  By noting that %
  we only modified the flow between $j$, $v_{\ell'}^i$ and $i$, we conclude that (c) and (d) holds.
  It remains to verify that $f'$ is feasible.
  First observe that since %
  $f'(j,v_{\ell'}^i) =  0$, $f'(v_{\ell'}^i,i) = f(v_{\ell'}^i,i) - \gamma \geq \gamma - \gamma = 0$, and 
  \[
    f'(v_{\ell'}^i,i) = f(v_{\ell'}^i,i) - \gamma \leq q^{S}_i([\tau_{\ell'},\tau_{\ell'+1}]) - \gamma = q^{S'}_i([\tau_{\ell'},\tau_{\ell'+1}]) \ ,
  \]
  the flow $f'$ is non-negative and satisfies the capacity constraints. Finally, it is easy to see that $f'$ satisfies the flow conservation at $v_{\ell'}^i$ and also at all other vertices where the flow has not changed. It only remains to argue that $f'$ also satisfies the flow conservation at $i$ and $j$.

  If $i$ has no outgoing flow in $f'$, then $i$ must also have no outgoing flow in $f$, which implies $i \in O(t)$ as $i$ must have a positive demand. Since both, the incoming flow and the  demand of $i$, decrease by $\gamma$ from $f$ and $G_F^{S}$ to $f'$ and $G_F^{S'}$, $i$ must still respect the flow conservation in $f'$ and  $G_F^{S'}$.
  If $i$  has outgoing flow in $f'$ for $G_F^{S'}$, then it must also have had outgoing flow in $f$ for $G_F^{S}$. By~\Cref{prop:flow:no-outgoing-sink-flow}, we must therefore have $i \not\in O(t)$, and thus, $i \in A^{S'}(t) \setminus O(t)$. Hence, $i$ has a supply of $p_i^{S'} = p_i^{S}+\gamma$ in $G_F^{S'}$. Since $f$ is feasible for $G_F^{S}$ and the incoming flow of $i$ from $f$ to $f'$ is decreased by at most $\gamma$, the increased supply of $i$ suffices to cover the decreased incoming flow. We can conclude that $i$ respects the flow conservation.

  Since $j$ has no incoming flow in $f$, we must have $j \in A(t)^{S}\setminus O(t)$ as only such vertices have a positive supply. If $j \in D^{S'}(t)$, then $j$ also has no outgoing flow in $f'$, and therefore, respects the flow conservation. If $j \in  A(t)^{S'}\setminus O(t)$, then the supply of $j$ from $G_F^{S}$ to $G_F^{S'}$ decreases by $\gamma$. However, since the outgoing flow of $j$ also decreases by $\gamma$, $j$ still satisfies the flow conservation.
\end{proof}

\lemmaflowpositivepaths*

\begin{proof}
  First, assume that there is a path $P$ from $j$ to $i$ in $G_F^S$ with positive capacity. Consider any two consecutive jobs $j'$ and $i'$ on this path.
  Since the path has positive capacity, there must exist a vertex $v^{i'}_{\ell'}$ such that the edge $(j',v_{\ell'}^{i'})$ exists in $G_F^S$, giving $[\tau_{\ell'},\tau_{\ell'+1}] \subseteq I_{j'}$, and the edge $(v_{\ell'}^{i'},i')$ has positive capacity, which means that $q_{i'}^S([\tau_{\ell'},\tau_{\ell'+1}]) > 0$. This 
  implies that $i'$ is being processed during $[\tau_{\ell'},\tau_{\ell'+1}] \subseteq I_{j'}$, and thus, $(j',i')$ in an edge in $G_B^S$. Thus, $P$ is a path in $G_B^S$.

  Second, assume that there exists a path $P$ from $j$ to $i$ in $G_B^S$.
  Consider any edge $(j',i')$ on this path. This means that there exists a point in time $t' \in I_{j'}$ at which $i'$ is being processed.
  Since $0$, $r_{j'}$, $C_{j'}$, and $t$ are part of $T^S$, there exist $\tau_{\ell'} \in T^S$ such that $t' \in [\tau_{\ell'},\tau_{\ell'+1}]$ and $[\tau_{\ell'},\tau_{\ell'+1}] \subseteq I_{j'}$.
  Thus, the edge $(j',v_{\ell'}^{i'})$ in $G_F^S$ of infinite capacity exists, and the edge $(v_{\ell'}^{i'},i')$ has capacity $q_{i'}^S([\tau_{\ell'},\tau_{\ell'+1}]) > 0$, because $i'$ is being processed at time $t'$. We conclude that there exists a path from $j$ to $i$ in $G_F^S$ with positive capacity that visits the same jobs as $P$ in the same order.
\end{proof}

\thmmainborrowing*

\begin{proof}
  By \Cref{def:borrow:values} and \Cref{lemma:flow:positive-paths}, $\beta(j,i)$ can only be positive if $i \in R_j$, that is, if there is a path from $j$ to $i$ in $G_F$.
  Thus, by \Cref{def:borrow:values},  we have $\beta(j,i) = 0$ for every $j \in A(t) \setminus O(t)$ and $i \in O(t) \setminus R_j$, giving Property (i).

  Since the values $\beta(j,i)$ are constructed based on a flow $f$ of value $\sum_{j \in A(t) \setminus O(t)} p_j(t)$, which is the total supply, Property (ii) follows. 
  It remains to prove Property (iii). Fix any job $i \in O(t)$, and let $T = \{\tau_1,\ldots,\tau_{\ell+1}\}$ be as in \Cref{def:flow-network}. The edges $(v_{1}^i, i),\ldots,(v_{\ell}^i,i)$ are the only incoming edges of $i$ in $G_F$. The sum of their capacities is $\sum_{\ell'=1}^{\ell} q_{i}([\tau_{\ell'},\tau_{\ell'+1}]) \leq y_i(t)$ because $\tau_{\ell'+1} \leq t$. Thus, the total amount of flow that can reach $i$ in $G_F$ is at most $y_i(t)$, ensuring Property (iii).
\end{proof}

\section{Structure of the Borrow Graph: Proofs of~\Cref{thm:bottleneck:uniqueness,coro:multiple:bottlenecks}}\label{app:structure-borrow}

In this section, we further analyze the structure of the borrow graph and prove the~\Cref{thm:bottleneck:uniqueness,coro:multiple:bottlenecks} as given in~\Cref{sec:structure}.

As in~\Cref{sec:structure}, we fix a job $i \in O(t)$ and analyze the borrow graph with respect to this job.

We start with the following observation stating that active jobs in a layer $L_{x+1}^i$ must be able to reach an entrypoint of layer $L_x^i$ via an $x$-entry path.

\begin{observation}\label{obs:active-entrypath}
  Let $j \in A(t) \cap L_{\geq x+1}^i$ for some $x \geq 0$. Then, there exists an $x$-entry path from $j$ to some $k \in H_x^i$.
\end{observation}

\begin{proof}
  Since $j \in L_{\geq x+1}^i$, there exists a path $P$ from $j$ to $i$. By the definition of layers, there must be at least one job in $P$ that is contained in $L_{x}^i$. Let $k \in L_{x}^i$ be the first such job on $P$. By the definition of layers, it must be that $P$ uses an edge $(d,k) \in \delta_x^i$ to reach $k$, where $d \in L_{x+1}^i$. Our choice of $k$ implies that there exists an $x$-entry path from $j$ to $k$, and since also $j \in A(t)$, we can conclude that $k \in H_x^i$.
\end{proof}

\subsection{Properties of Entrypoints}

Having the layer structure and entrypoints set up, we want to more precisely analyze schedules of our algorithm that belong to a certain borrow graph structure.

First, we characterize the structure of edges $(j,k) \in \delta_x^i$. These edges allow a path in the borrow graph to enter layer $L_x^i$ from layer $L_{x+1}^i$.

\begin{lemma}
  \label{lem:ec:cuts:prop:1}
  Consider two jobs $j$ and $k$ such that $(j,k)\in \delta_x^i$. Then, it holds that
  \begin{enumerate}
      \item $r_j \ge s_{k}$,
      \item $r_j \le C_k \le C_j$, and
      \item $y_j(C_k) \le \alpha \cdot p_j$.
  \end{enumerate}
\end{lemma}

\begin{proof}
  By the layer structure, we have that $(j,k)\in E_C$ but $(j,k) \not\in E_N$.
  We show the first property via proof by contradiction. To this end, assume $r_j < s_k$.
  Since $(j,k)\in E_C$, there must be a point in time $t'$ with $r_j \le t' \le C_j$ such that $k \in C(t')$ and $k$ is executed at $t'$. This directly implies that $j$ must still be alive at point in time $s_k$, i.e., $r_j < s_k \le C_j$. As $k$ is executed at $s_k^-$ while being non-clairvoyant, this implies $(j,k)\in E_N$; a contradiction to the assumption that $(j,k) \not\in E_N$.

  For the second property, let $t' \ge r_j$ denote the earliest point in time at which $k$ is executed during $I_j$. %
  Then, \Cref{obs:clairvoyant-jobs-block-earlier-jobs} implies $C_k \leq C_j$.

  Finally, we show the third property via proof by contradiction. Assume $y_j(C_k) > \alpha \cdot p_j$. This implies $s_j < C_k$ and that $j$ is executed at some point in time $t'$ with $s_j < t' < C_k$; otherwise we would have $y_j(C_k) = \alpha \cdot p_j$. By the first property of this lemma, $k$ is also clairvoyant at $t'$. Thus, by definition of the algorithm, $p_j(t') \le p_k(t')$. However, this implies that $k$ (after $t'$) is not executed again until $j$ completes. This contradicts $C_k \le C_j$.   
\end{proof}

  We continue by establishing several properties of entrypoints. The first observations states that the lifetimes of differents paths between entrypoints of consecutive layers must intersect.

\begin{observation}
    \label{obs:entrypoint:paths}
    Let $k_{x+1}, k'_{x+1} \in H^i_{x+1}$, let $k_x,k_x' \in L_x^i$, let $P$ be a $k_{x+1}$-$k_x$-path in $G_B$ that (apart from $k_x$) only visits jobs in $L_{x+1}^i$ and let $P'$ be a $k'_{x+1}$-$k'_x$-path in $G_B$ that (apart from $k_x'$) only visits jobs in $L_{x+1}^i$. Then $I(P) \cap I(P') \neq \emptyset$.
\end{observation}

\begin{proof}
  Via proof by contradiction. To this end, assume $I(P) \cap I(P') = \emptyset$. Furthermore, without loss of generality, assume that $\max(I(P))< \min(I(P'))$. Since $k_{x+1}$ is an entrypoint of layer $L_{x+1}^i$, there must be an $(x+1)$-entry path $\bP$ from some $h \in A(t)$ to $k_{x+1}$.  Since $k_{x+1} \in \bP$ and $h \in \bP$, we get $[r_{k_{x+1}},t] \subseteq I(\bP)$ by \Cref{obs:lifetime:2}. 
  This, together with our assumption that $\max(I(P))< \min(I(P'))$ and $k_{x+1} \in P$, implies $I(P') \subseteq I(\bP) \setminus I_{k_{x+1}}$. However, this means that $k_x'$ is executed during the lifetime of a job in $\bP \setminus k_{x+1}$, which must be part of $L_{\geq x+1}^i$; a contradiction to the layer definition.
\end{proof}

The next two observations establish an order of the layers over time in the following way: entrypoints of layers with larger indices are executed at later points in time than entrypoints of layers with smaller indices. Intuitively, this means that paths $P$ from jobs $j \in L_x^i$ to $i$ go backwards in time w.r.t.~to the latest points of execution $t_{j'}$ of jobs $j' \in P$ (possibly in a non-monotonous way for the non-entrypoint jobs in $P$). This intuition will be heavily used in the remaining proofs of the paper.

\begin{observation}
  \label{obs:entrypoint:order}
  Let $k_x \in L_x^i$ and let $k_y$ be an entrypoint of a layer $L_y$ with $y \geq x+1$. Then, $t_y > t_x$ for the latest points in time $t_y, t_x \le t$ at which $k_y$ and $k_x$ are executed, respectively. 
  \end{observation}

  \begin{proof}
  Via proof by contradiction. Assume $t_x \ge t_y$. Let $P$ denote a $y$-entry path from some $d \in A(t) \cap L^i_{\geq y + 1}$ to $k_y$, which exists because $k_y$ is an entrypoint. %
  Let $P'$ denote the subpath of $P$ without $k_y$. Then, we must have $t_y \in I(P')$, as otherwise it would be impossible to extend $P'$ to $P$. Since $d \in P'$, we get $[t_y,t] \subseteq I(P')$ by \Cref{obs:lifetime:2}. This implies $t_x \in I(P')$, and thus, the existence of an edge from some $h \in P'$ to $k_x$; a contradiction to the layer definition. 
  \end{proof}

  \begin{observation}
    \label{obs:entrypoint_blocking}
    Let $k \in H_x$ be an entrypoint of layer $L_x$ and let $d \in L^i_{\leq x-1}$ be a job from an earlier layer. Then, $d$ is not executed during $[C_k,t]$. Furthermore, if $d \in L^i_x$, then $d$ is not executed during $[C_k,t]$ while being non-clairvoyant.
\end{observation}

\begin{proof}
    Since $k$ is an entrypoint of layer $L^i_x$, there must be an $x$-entrypath $P$ from some job $j \in A(t) \cap L^i_{\geq x+1}$ to $k$. By~\Cref{obs:lifetime:2}, we have $I(P) \supseteq [r_k,t]$. Furthermore, as some job in $P \setminus \{k\}$ must be interrupted by $k$ for $P$ to be a path, we have $I(P\setminus\{k\}) \supseteq [C_k,t]$.

    If $d\in L^i_{ \leq x-1}$ is executed during $[C_k,t]$, then it interrupts some job $d' \in P\setminus \{k\}$, which implies $(d',d) \in E$. However, since $d \in L^i_{\leq x-1}$ and $d' \in L^i_{\geq x+1}$, the edge $(d',d)$ contradicts the layer definition.

    If  $d \in L^i_x$ is executed during $[C_k,t]$ while being non-clairvoyant, then this implies $(d',d) \in E_N$ for some $d' \in  P\setminus \{k\} \subseteq L^i_{\geq x+1}$. However, the edge $(d',d) \in E_N$ implies $d'\in L^i_{\leq x}$, a contradiction.
\end{proof}

Finally, we show that all entrypoints of a layer $L_{x+1}^i$ can reach exactly the same entrypoints of layer $L_x^i$.

  \begin{lemma}
  \label{lem:entrypoint:reachability}
  Let $k_{x+1}, k'_{x+1} \in H^i_{x+1}$ be two entrypoints of layer $L_{x+1}^i$. Then, $k_{x+1}, k'_{x+1}$ can reach exactly the same entrypoints of layer $L_{x}^i$ (via paths that only visit jobs in $L_{x+1}^i$ apart from the endpoints in $H_x^i$).
  \end{lemma}
  
  Since the paths in the above lemma are specifically $x$-entry paths, we conclude the following corollary.
  
  \begin{corollary}
  \label{coro:entrypoint:reachability-entrypaths}
  Let $k_{x+1}, k'_{x+1} \in H^i_{x+1}$ be two entrypoints of layer $L_{x+1}^i$. Then, $k_{x+1}, k'_{x+1}$ can reach exactly the same entrypoints of layer $L_{x}^i$ via $x$-entry paths.
  \end{corollary}
  
To conclude this section on properties of entrypoints, we prove \Cref{lem:entrypoint:reachability}.

  \begin{proof}[Proof of \Cref{lem:entrypoint:reachability}]
  We prove the lemma via contradiction. 
  To this end, assume that there exists an entrypoint $k_x$ of layer $L_x^i$ that can be reached by $k_{x+1}$ via a path that only visit jobs in $L_{x+1}^i$ apart from $k_x$, but $k_x$ cannot be reached by $k_{x+1}'$. Let $P$ denote the corresponding $k_{x+1}$-$k_x$-path. Since $k'_{x+1} \in L_{x+1}^i$, there exists a path $P'$ from $k_{x+1}'$ to some $k_x' \in L_x^i$ (note that in fact $k_x' \in H_x^i$) such that $P'$ only visits jobs in $L_x^i$ and the last edge of $P'$ is part of $\delta_x^i$. 
  Note that $k_x \neq k'_x$, and $k_x'$ is actually an entrypoint, that is, $k_x' \in H_x^i$. %

  By \Cref{obs:entrypoint:paths}, we have $I(P) \cap I(P') \not= \emptyset$. We distinguish between the two cases (1) $\max(I(P')) \ge \max(I(P))$ and (2) $\max(I(P)) > \max(I(P'))$. In both cases, we will prove a contradiction, which will imply the statement of the lemma.

  \paragraph*{Case (1): $\max(I(P')) \ge \max(I(P))$.} If $\max(I(P')) \ge \max(I(P))$ and $I(P') \cap I(P) \not= \emptyset$, then we must have $\max(I(P)) \in I(P')$. 
  Note that by definition of $I(P)$, we either have (i) $\max(I(P)) = C_d$ for some job $d \in P$ or  (ii) $\max(I(P)) = t$.
  
  \textbf{Subcase (i):} If $\max(I(P)) = C_d$ for some job $d \in P$, then the job $d$ must be executed during the lifetime of some job $h$ in $P'$, and thus, $(h,d) \in E$. If $h \not= k_x'$, then 
  there exists a path from $k_{x+1}'$ to $h$, and via $(h,d)$ to $k_x$ that apart from $k_x$ only visits jobs in $L_{x+1}^i$;
  a contradiction to our assumption. 
  However, if $h = k_x'$, 
  we could construct the same path, but this path would visit two jobs outside of $L_x^i$ and the contradiction does not necessarily follow.
  Thus, it remains to argue that we always have $h \not= k_x'$.  
  If $h = k_x'$, then we must have $t_{k_x'} < \max(I(P))$ as otherwise we would get $t_{k_x'} \ge \max(I(P)) \ge t_{k_{x+1}}$, which is a contradiction to \Cref{obs:entrypoint:order}. This means that $k_x'$ is alive at $\max(I(P))$ (since $h = k_x'$ is not executed at $\max(I(P))$) and not executed again afterwards.
  
  Thus, $k_x'$ is alive at point in time $t$.
  Let $(j',k_x') \in \delta_x^i$ be the final edge of $P'$. %
  By \Cref{lem:ec:cuts:prop:1}, we have $C_{j'} \ge C_{k_x'}$ and, therefore, $C_{j'} > t$. Furthermore, since $(j',k_x') \in \delta_x^i$ is an edge, $j'$ must be released before $t_{k_x'} < \max(I(P))$. This means that $j'$ is alive at $\max(I(P))$. Thus, we can assume $h = j' \not= k_x'$ and arrive at the same contradiction as described at the beginning of Subcase (i).
  
  \textbf{Subcase (ii):} If  $\max(I(P)) = t$, then we also have  $\max(I(P')) = t$ by our assumption that $\max(I(P')) \ge \max(I(P))$.
  Since we must have $\min\{t_{k_{x+1}},t_{k_{x+1}'}\} \ge \max\{t_{k_{x}},t_{k_{x}'}\}$ by \Cref{obs:entrypoint:order}, the fact that $\max(I(P)) = \max(I(P')) = t$ implies $t_{k_{x+1}},t_{k_{x+1}'} \in I(P) \cap I(P')$. 
  
  In particular, we get $t_{k_{x+1}} \in I(P')$, which means that $k_{x+1}$ is processed during the lifetime of some job $h \in P'$ at point in time $t_{k_{x+1}}$. Specifically, $(h,k_{x+1}) \in E$. 
  If $h \not= k_x'$, then
  the path from $k_{x+1}'$ to $h$ and via $(h,d)$ to $k_x$ visits, apart from $k_x$, only jobs in $L_{x+1}^i$;
  a contradiction to our assumption. 
  
  However, if $h = k_x'$, then the combined path would visit two jobs outside $L_{x+1}^i$ and the contradiction would not occur.
  Thus, it remains to argue that $h \not= k_x'$. 
  If $h = k_x'$, then $k_x'$ is still alive at $t_{k_{x+1}}$, because $k_{x+1}$ is being processed at time $t_{k_{x+1}}$ during $I_h$. Since $k_x'$ is not executed after $t_{k_x'} < t_{k_{x+1}}$, this means that $k_x'$ is still alive at point in time $t$. 
  Consider the edge $(j',k_x') \in \delta_x^i$ that is part of path $P'$. By definition of path $P'$ it must contain such an edge. By \Cref{lem:ec:cuts:prop:1}, we have $C_{j'} \ge C_{k_x'}$ and, therefore, $C_{j'} \ge t$. Furthermore, since $(j',k_x') \in \delta_x^i$ is an edge, $j'$ must be released before $t_{k_x'} < t_{k_{x+1}}$. This means that $j'$ is alive at $t_{k_{x+1}}$. Thus, we can assume $h = j' \not= k_x'$ and arrive at the same contradiction as described at the beginning of Subcase (ii).

  \paragraph*{Case (2): $\max(I(P')) < \max(I(P))$.} We start by making the following observations:
  \begin{enumerate}[(a)]
      \item No job in $h \in P \setminus \{k_x\}$ is alive at $s_{k_x'}^-$. Otherwise, we would get an edge $(h,k_x') \in E_N$ for a job $h \in L^i_{\geq x+1}$; a contradiction to the layer definition. Thus, all jobs $h \in P \setminus \{k_x\}$ must either satisfy $C_h < s_{k_x'}$ or $r_h \ge s_{k_x'}$.
      \item No job $h \in P \setminus \{k_x\}$ can have $r_h > I(P')$.
      Let $\bP$ be a $(x+1)$-entry path from some $d \in A(t) \cap L^i_{\geq x+2}$ to $k_{x+1}'$
      If $r_h > I(P')$, then $r_h \in I(\bP) \setminus I_{k_{x+1}'}$, which implies the existence of an edge $(d,h) \in E_N$ with $h \in L_{x+1}^i$ and $d \in L^i_{\geq x+2}$; a contradiction to the layer definition.
  
      More generally, no job in $h \in P \setminus \{k_x\}$ can be executed at a point in time $t' > \max(I(P'))$ while being non-clairvoyant. Otherwise, we arrive at the same contradiction.
  
      \item The combination of the first two observations implies that each $h \in P \setminus \{k_x\}$ either satisfies $C_h < s_{k_x'}$ or $s_{k_x'} \le r_h \le \max(I(P'))$.
  
      \item If $\max(I(P')) \ge r_h \ge C_{k_x'}$ for some $h \in P \setminus \{k_x\}$, then $h$ is being processed during the lifetime of some job in $P' \setminus \{k_x'\}$, which implies that $P$ and $P'$ can be combined into a $k_{x+1}'$-$k_x$-path that apart from $k_x$ only visits jobs in $L_{x+1}^i$; a contradiction to our assumption that $k_{x+1}'$ cannot reach $k_x$ via such a path.
      
      \item More generally, no job in $h \in P \setminus \{k_x\}$ can be executed at a point in time $C_{k_x'} \le  t' \le \max(I(P'))$. Otherwise, we arrive at the same contradiction.
      
      \item Consider a job $h \in P \setminus \{k_x\}$ with $r_h \geq s_{k_x'}$. If $h$ emits before $k_x'$ completes, then the algorithm will complete $h$ before $k_x'$ (cf.~\Cref{obs:clairvoyant-jobs-block-earlier-jobs}).
      
      Thus, all jobs $h \in P \setminus \{k_x\}$ that are still alive at $C_{k_x'}$ are still non-clairvoyant and, by Points (d) and (e), will therefore not be executed again.  
  \end{enumerate}
  
  Let $h \in P \setminus \{k_x\}$ and let $t_h$ be the latest point in time before $t$ when $h$ is executed.
  We claim that $t_h < C_{k_x'}$. If $C_h < s_{k_x'}$ then the claim trivially holds. Thus, by (c), we can assume that $s_{k_x'} \leq r_h \leq \max(I(P'))$.
  By (d), we know that it remains to verify the claim for the case $s_{k_x'} \leq r_h \leq C_{k_x'}$. 
  If $h \in D(C_{k_x'})$, the claim clearly holds.
  Otherwise, that is $h \in A(C_{k_x'})$, (f) gives that $h \in N(C_{k_x'})$. Now, (e) implies that $h$ is not executed between $C_{k_x'}$ and $\max(I(P'))$. Thus, $h \in N(\max(I(P')))$. Finally, (b) gives that $h$ cannot be executed after time $\max(I(P'))$ while being non-clairvoyant. This proves that $t_h < C_{k_x'}$.

  We will conclude the proof by showing that this leads to a contradiction. 
  By the definition of this case (that is, $\max(I(P')) < \max(I(P))$), we can conclude that $C_{k_x'} < t$, and thus, $t_{k_x'} = C_{k_x'}$.
  Now by the above claim using that $k_{x+1} \in P \setminus \{k_x\}$ we obtain $t_{k_{x+1}} < C_{k_x'} = t_{k_x'}$. This is a contradiction to \Cref{obs:entrypoint:order}.
  \end{proof}

\subsection{Bottlenecks}

We continue by proving~\Cref{thm:bottleneck:uniqueness}, i.e., that bottlenecks are unique (if they exist).

\uniqueBottleneck*

If $L_x^i$ has a bottleneck because $|H_{x}^i| =1$, then $L_x^i$ clearly has only one bottleneck since it only has one entrypoint. Thus, it remains to consider the case where $|H_x^i| > 1$ but $L_x^i$ still has a bottleneck.

To prove the theorem, we first show the following two auxiliary lemmas.

\begin{lemma}
  \label{lem:aux:multiple:bottlenecks:1}
  Let $j \in A(t) \cap L_{x+1}^i$ and $k_{x+1} \in H_{x+1}^i$. Let $k_x \in H_x^i$ be a job that can be reached by $k_{x+1}$ via an $x$-entry path. Then, $j$ can also reach $k_x$ via an $x$-entry path. 
\end{lemma}

\begin{proof}
  Let $P$ denote an arbitrary $x$-entry path from $k_{x+1}$ to $k_x$, %
  and let $P'$ be an $x$-entry path from $j$ to some $k_x' \in H_x^i$ (cf.~\Cref{obs:active-entrypath}).
  Recall that $t_{k_{x+1}}$ is the latest point in time before or at time $t$ at which $k_{x+1}$ is executed.

  We first show that $t_{k_{x+1}} \in I(P')$. %
  To this end, assume that $t_{k_{x+1}} \notin I(P')$.
  Observe that, since $j \in A(t) \cap P'$, we have $t \in I(P')$ by \Cref{obs:lifetime:2}. This means that $t_{k_{x+1}} \not\in I(P')$ is only possible if $t_{k_{x+1}} < \min(I(P'))$.
  Since $k_{x+1}$ is an entrypoint, there exists a $(x+1)$-entry path $\bP$
  from some $d \in A(t) \cap L_{\geq x+2}^i$ to $k_{\geq x+1}$.
  Next, observe that $[t_{k+1},t] \subseteq I(\bP)\setminus I_{k_{x+1}}$. To see this, first note that $t \in  I(\bP)\setminus I_{k_{x+1}}$ holds by \Cref{obs:lifetime:2} as $d \in A(t) \cap (\bP \setminus \{k_{x+1}\})$. Second, note that $t_{k+1} \in I(\bP)\setminus I_{k_{x+1}}$ must hold as some job in $\bP \setminus \{k_{x+1}\}$ must be executed during $I_{k_{x+1}}$ by the assumption that $\bP$ is a path in $G_B$.
  Now, $[t_{k+1},t] \subseteq I(\bP)\setminus I_{k_{x+1}}$ and $t_{k_{x+1}} < \min(I(P'))$ imply $I(P') \subseteq I(\bP)\setminus I_{k_{x+1}}$. 
  This means that $k_x'$ is executed during the lifetime of some job in $\bP \setminus \{k_{x+1}\}$ while being non-clairvoyant; a contradiction to $k_x' \in L_x^i$ and $\bP \setminus \{k_{x+1}\} \subseteq L_{\geq x+2}^i$.

  Having established $t_{k_{x+1}} \in I(P')$, it remains to show that $j$ can reach $k_x$ via an $x$-entry path. %
  Observe that $t_{k_{x+1}} \in I(P')$ implies that $k_{x+1}$ 
  is executed during the lifetime of 
  some job $h \in P'$, and thus, there is the edge $(h,k_{x+1})$ in $G_B$.
  Therefore, we can combine the $j$-$h$-subpath of path $P'$, the edge $(h,k_{x+1})$, and path $P$ into a $j$-$k_x$-path in $G_B$. 

  It remains to argue that there exists such a combined path that apart from $k_x$ only visits jobs in $L_{\geq x+1}^i$. To this end, we show that there exists an $h \in L_{\geq x+1}^i$ such that $k_{x+1}'$ is being processed during $I_h$.
  Observe that $t_{k_x'} < t_{k_{x+1}}$ by \Cref{obs:entrypoint:order}. Next, consider $I(P'\setminus \{k_x'\})$. Since $j \in A(t) \cap (P'\setminus \{k_x'\})$, we have $t \in I(P'\setminus \{k_x'\})$ by \Cref{obs:lifetime:2}. Furthermore, we must have $t_{k_x'} \in I(P'\setminus \{k_x'\})$ as there must be some job $h' \in P'\setminus \{k_x'\}$ such that $k_{x}'$ is being processed at some point during $I_{h'}$; otherwise $k_{x}'$ cannot be on $P'$.
  Thus, $[t_{k_x'},t] \subseteq I(P'\setminus \{k_x'\})$, which implies $t_{k_{x+1}} \in  I(P'\setminus \{k_x'\})$. We can conclude that there exists some $h \in L_{\geq x+ 1}^i$ such that $k_{x+1}$ is being processed at some point during $I_{h}$.
\end{proof}

\begin{lemma}
  \label{lem:unique:bottleneck:1}
  Let $j_1,j_2 \in (A(t) \cap L_{x+1}^i) \cup H_{x+1}^i$. If $j_1$ can only reach a single entrypoint in $H_x^i$, say $k_1$, and $j_2$ can only reach a single entrypoint in $H_x^i$, say $k_2$, then $k_1 = k_2$.
\end{lemma}

\begin{proof}

  Consider two jobs $j_1,j_2 \in  (A(t) \cap L_{x+1}^i) \cup H_{x+1}^i$ that each can only reach a single entrypoint of $H_x^i$ via $x$-entry paths. Let $k_1$ and $k_2$ denote those entrypoints. Our goal is to show that $k_1 = k_2$. 

  \textbf{Case 1:} $j_1,j_2 \in  A(t) \cap L_{x+1}^i$. For the sake of contradiction, assume that $k_1 \neq k_2$. W.l.o.g.\ assume $r_{j_1} \le r_{j_2}$. Since $j_1$ and $j_2$ are both still alive at $t$, $r_{j_1} \le r_{j_2}$ implies that if a job $d$ is executed during $I_{j_2}$, then it is also executed during $I_{j_1}$. Thus, the $x$-entry path from $j_2$ to $k_2$ that exists by assumption can be transformed into an $x$-entry path from $j_1$ to $k_2$ by only exchanging $j_2$ with $j_1$ on the path. This implies that $j_1$ can also reach $k_2$ via an $x$-entry path; a contradiction.

  \textbf{Case 2:} $j_1, j_2 \in H_{x+1}^i$.
  \Cref{coro:entrypoint:reachability-entrypaths} implies that $j_1$ and $j_2$ can reach exactly the same entrypoints of $L_x^i$ via $x$-entry paths. Thus, $k_1 = k_2$.

  \textbf{Case 3:} $j_1 \in H_{x+1}^i$ and $j_2 \in A(t) \cap L_{x+1}^i$. \Cref{lem:aux:multiple:bottlenecks:1} implies that $j_2$ can reach every entrypoint of $H_x^i$ that $j_1$ can reach via $x$-entry paths. Thus, $k_1 = k_2$.  
\end{proof}

Having this lemma in place, proving \Cref{thm:bottleneck:uniqueness} is straightforward.

\begin{proof}[Proof of \Cref{thm:bottleneck:uniqueness}]
  If $|H_x^i| = 1$, the claim follows easily. For the case that $|H_x^i| > 1$,
  we assume $L_x^i$ has two bottlenecks $k_x, k_x' \in H_x^i$ with $k_x \not= k_x'$.  By definition there must exist two jobs $j,j' \in (A(t) \cap L_{x+1}^i) \cup H_{x+1}^i$ that can only reach layer $L_x^i$ via paths that contain $k_x$ or $k_x'$, respectively. Since $j,j' \in L_{x+1}^i$, there must be such paths $P$ and $P'$ that are $x$-entry paths, respectively. This allows us to apply \Cref{lem:unique:bottleneck:1} to conclude that $k_x = k_x'$; a contradiction.
\end{proof}

\subsection{Multiple Bottlenecks}

We have shown that a layer $L_x^i$ has at most bottleneck. Next, we consider situations of the following type: Assume $L_x^i$ and $L_{x+1}^i$ both have bottlenecks $k_x$ and $k_{x+1}$. From what we have shown so far, it is not clear that a job $j \in A(t) \cap L^i_{x+2}$ that can reach jobs in $L_{x+1}^i$ only via paths that contain the bottleneck $k_{x+1}$ can only reach jobs in $L_{x}^i$ only via paths that contain the bottleneck $k_x$. In this section, we show that this is indeed the case. More precisely, we show the following theorem, which  directly implies \Cref{coro:multiple:bottlenecks}.

\begin{theorem}
    \label{thm:multiple:bottlenecks}
    Let $j \in A(t) \cap L_{x+1}^i$. Then, every path from $j$ to $i$ contains \emph{all} bottlenecks of the layers $L_y^i$ with $y \leq x-1$.
\end{theorem}

\begin{proof}
  Consider a layer $L_y^i$ with $y \leq x - 1$ that has a bottleneck. If $|H_y^i| = 1$, then the definition of layers implies that $j \in A(t) \cap L_{x+1}^i$ has to visit the bottleneck of $L_y^i$ in order to reach~$i$. Thus, it remains to consider the case $|H_y^i| > 1$.

  In this case, for $L_y^i$ to have a bottleneck $k_y$, there must exist a job $j' \in (A(t) \cap L_{y+1}^i) \cup H_{y+1}^i$ that can only reach layer $L_y^i$ through paths that contain $k_y$.
  We show that $j$ also must cross $k_y$ to reach $i$ via proof by contradiction. To this end, assume that there is a $y$-entry path $P$ from $j$ to some $k_y' \in L_y^i$ with $k_y' \neq k_y$. 
  If $j$ can reach $i$ without visiting $k_y$, then such a path must exist. 
  Note that since $j \in A(t) \cap L_{\geq y+1}^i$ and $P$ is a $y$-entry path, $k_y'$ is an entrypoint of layer $L_y^i$. 

  Observe that $P$ must contain a $k_{y+1}$-$k_y'$-subpath $P'$, where $k_{y+1} \in H_{y+1}^i$ and $P'$ is a $y$-entry path.

  If $j' \in A(t) \cap L_{y+1}^i$, then \Cref{lem:aux:multiple:bottlenecks:1} implies that $j' \in A(t) \cap L_{y+1}^i$ can also reach $k_y'$ via a $y$-entry path. %
  This is a contradiction to our assumption that $j' \in A(t) \cap L_{y+1}^i$ is only able to reach layer $L_y^i$ through paths that contain $k_y$.

  If $j' \in H_{y+1}^i$, then \Cref{lem:entrypoint:reachability} implies that $j'$ and $k_{y+1}$ can reach exactly the same entrypoints of $L_y^i$ via $y$-entry paths. Thus, we must either have $k_y = k'_y$ or $j'$ and $k_{y+1}$ can reach $k_y$ and $k_y'$; both cases are a contradiction.
\end{proof}

The following observation partially extends the statement of~\Cref{thm:multiple:bottlenecks} to entrypoints.

\begin{observation}
  \label{obs:entrypoint:bottleneck}
  If a layer $L_x^i$ has a bottleneck, then every entrypoint $j \in H_{x+1}^i$ can reach jobs in $L_{x}^i$ \emph{only} via paths that contain the bottleneck $b^i_x$.
\end{observation}

\begin{proof}
  Via proof by contradiction. To this end, assume that there is an entrypoint $k_{x+1} \in L_{x+1}^i$ that can reach $L_{x}^i$ via an entrypoint $k_x \in H_x^i \setminus \{b_x^i\}$, i.e., there exists $x$-entry-path $P$ from $k_x$ to $k_{x+1}$.

  Since $b^i_x$ is a bottleneck but $|H_x^i| > 1$, there must be a job $j \in (A(t) \cap L_{x+1}^i) \cup H_{x+1}^i$ that can only reach $L_x^i$ via $b_x^i$.
  If $j \in A(t) \cap L_{x+1}^i$, then \Cref{lem:aux:multiple:bottlenecks:1} implies that $j$ can also reach $L_{x+1}^i$ via $k_x$; a contradiction.
  If $j \in H_{x+1}^i$, then \Cref{lem:entrypoint:reachability} implies that $j$ can also reach $L_{x+1}^i$ via $k_x$; a contradiction.
\end{proof}

\section{Fine-Grained Borrowing}

In this section, we want to consider a fixed job $i \in O(t)$ and analyze how much processing a job $j \in A(t)$ can borrow from $i$, i.e., we want to upper bound $\beta(j,i)$. As seen in~\Cref{sec:matching:arguments,sec:segments}, one useful way of upper bounding $\beta(j,i)$ is by showing the inequality $\by_j \ge \by_i$. Unfortunately, this inequality does not hold for every pair $i \in O(t)$ and $j \in A(t)$. However, as one of the main results in this section, we will show that the inequality holds if $j$ can reach $i$ in $G_B$ via a path that does not contain any bottleneck of $i$.
With respect to the segments (cf.~\Cref{def:segment}), this means that if a job $j \in S$ for a segments $S$ can reach a job $i \in O$ without visiting any bottleneck, then $i \in O_S$. Vice versa, if $j$ can reach a job $i \in O\setminus O_S$, then this must be via at least one bottleneck.

Afterwards, we consider jobs $j \in A(t)$ that must visit at least one bottleneck of $i$ in order to reach $i$ and analyze how much processing $j$ can borrow via that bottleneck. As argued above, for a job $j$ in a segment $S$, all jobs $i \in O\setminus O_S$ that $j$ can actually reach fall into this category.
This analysis will later on be a key ingredient for proving~\Cref{lem:segment:bottleneck:capacity}.

For the remainder of this section, we fix a job $i \in O(t)$ and analyze jobs $j$ with $i \in R_j$.

\subsection{Borrowing via Non-Clairvoyant Paths}
\label{sec:non-clairvoyant:paths}

In order to show that $\by_j \ge \by_i$ holds for jobs $j \in A(t)$ that can reach $i$ without visiting a bottleneck of $i$, we first consider jobs $j \in L_0^i$, i.e., jobs that can reach $i$ via $E_N$-paths. By definition, such jobs $j \in L_0^i$ can reach $i$ without visiting a bottleneck of $i$.
To analyze these jobs, we introduce the concept of \emph{good} edges.

\begin{definition}[Good and bad edges]
  We call an edge $(j,k) \in E_N$ \emph{good} if $\by_j \geq \by_k$, and \emph{bad} otherwise.
  We call a path in $G_B$ \emph{good} if it is composed of good edges.
\end{definition}

If a job $j \in A(t)$ has a good $E_N$-path to job $i$, then the definition of good edges immediately implies $\by_j \ge \by_i$. 
Unfortunately, we will see below that not every $E_N$-path is also good. 
However, we will show that the existence of an $E_N$-path from $j \in A(t)$ to $i$ implies the existence of a good path from $j$ to $i$. More precisely, we will prove the following lemma, which also implies~\Cref{lem:nc-borrowing}.

\begin{lemma}\label{lem:good-path-from-active}
  If a job $j \in A(t)$ can reach a job $i$ with an $E_N$-path, then $j$ can also reach $i$ via a good path, and thus, $\by_j \geq \by_i$.
\end{lemma}

A similar statement has been shown in \cite{KalyanasundaramP00} for SETF. Note that proving this statement for our algorithm is significantly more complex, because our algorithm may execute SRPT before time $t$, which makes the structure of its schedule more complicated.
In order to generalize the statement to our algorithm, we first characterize good and bad edges.

\paragraph*{Characterization of good and bad edges.} 
To this end, we first show that every edge $(j,k) \in E_N$ with $j \in N(t)$ is good. Observe that if $(j,k) \in E_N$ and $j,k \in N(t)$, then $y_j(t) \geq y_k(t)$ by \Cref{lemma:nonclairvoyant-direct-borrow}, and thus, $(j,k)$ is good. Next, we slightly extend this statement.

\begin{lemma}\label{lem:nc:c:edge}
  Let $j \in N(t)$ and $k \in J$. If there is an edge $(j,k) \in E_N$, then $y_j(t) \ge y_k(s_k)$ and $(j,k)$ is good.
\end{lemma}

\begin{proof}
  As stated above, if $k \in N(t)$, then the lemma is implied by~\Cref{lemma:nonclairvoyant-direct-borrow}. Thus, it remains to consider the case $k \in J\setminus N(t)$.

  Since $(j,k) \in E_N$, there is a point in time $t'$ with $k \in N(t')$ and $t' \in I_j$ at which $k$ is being processed. 
  \Cref{obs:sept-like} gives $y_k(t') \leq y_j(t')$. Since $j \in N(t)$ and $k \in N(t') \setminus N(t)$, job $k$ must emit during $[t',t]$, that is, $s_k \in [t',t]$. 
  Thus, \Cref{lemma:ketchup} implies $y_j(s_k) \geq y_k(s_k)$, and we can conclude that $\by_j = y_j(t) \geq y_j(s_k) \geq y_k(s_k) = \by_k$.
\end{proof}

After having shown that every $(j,k) \in E_N$ with $j \in N(t)$ is good, we continue by showing the same statement for $j \in C(t)$. To this end, we give separate lemmas for the two cases $k \in N(t)$ and $k \not\in N(t)$.

\begin{lemma}
  \label{lem:c:nc:edge}
  Let $j \in C(t)$ and $k \in N(t)$. If there is an edge $(j,k) \in E_N$, then $y_k(t) \le y_j(s_j)$, and thus, $(j,k)$ is good.
\end{lemma}

\begin{proof}
  Since $(j,k) \in E_N$, there is a point in time $t' \in I_j$ with $k \in N(t')$ at which $k$ is being processed. By \Cref{obs:sept-like}, we have $y_j(t') \geq y_k(t')$.

  By our assumption that $j \in C(t)$, we know that $j$ emits during $[r_j,t]$, i.e., $s_j \le t$. We finish the proof by distinguishing between the two cases (1) $t' < s_j$ and (2) $t' \ge s_j$.

  \paragraph*{Case (1):} Assume $t' < s_j$. Then $y_k(t') \le y_j(t')$ by \Cref{obs:sept-like}. Since $k$ is already alive at $t' < s_j$ and $k \in N(t)$, we get $k \in N(s_j)$. Thus,  \Cref{lemma:ketchup} gives $y_k(s_j) \le y_j(s_j)$. 
  Furthermore, \Cref{obs:clairvoyant-jobs-block-earlier-jobs} implies that $k$ is not executed during $[s_j,C_j] \supseteq [s_j,t]$. Thus, $\by_j = y_j(s_j) \ge y_k(s_j) = y_k(t) = \by_k$.

  \paragraph*{Case (2):} Assume $t' \ge s_j$. Then \Cref{obs:clairvoyant-jobs-block-earlier-jobs} implies $r_k \ge s_j$.

  If $p_k \le p_j$, then we must have $\by_k = y_k(t) \le y_j(s_j) = \by_j$. Otherwise, i.e., if $y_k(t) > y_j(s_j)$, the fact that $y_j(s_j) \ge \alpha p_j$ would imply $y_j(t) > \alpha p_j \ge \alpha p_k$, which means that $k$ would emit before time $t$ and, thus, not be part of $N(t)$; a contradiction the assumption of the lemma.

  If $p_k \geq p_j$, we again claim that $\by_k = y_k(t) \le y_j(s_j) = \by_j$ must hold. Assume otherwise, i.e., $y_k(t) > y_j(s_j)$. 
  Then, as $r_k \ge s_j$, there must be a point in time $t'' \geq r_k$ with $y_k(t'') = y_j(s_j) = \alpha \cdot p_j$. 
  
  Thus, $\frac{1-\alpha}{\alpha} y_k(t'') = (1-\alpha) \cdot p_j$. However, since $t'' \geq r_k \geq s_j$, we also have $p_j(t'') \le (1-\alpha) \cdot p_j$. Thus, $\frac{1-\alpha}{\alpha} y_k(t'') \geq p_j(t'')$, which means that the algorithm will not work on $k$ after time $t''$ until $j$ is completed. 
  Since $j$ does not complete before $t$, we can conclude with $\by_k = y_k(t) = y_k(t'') = y_j(s_j) = \by_j$. 
\end{proof}

\begin{lemma}
  \label{lem:c:c:edge}
  Let $j \in C(t)$ and $k \in J\setminus N(t)$. If there is an edge $(j,k) \in E_N$, then $\by_k = y_k(s_k) \le y_j(s_j) = \by_j$, and thus, $(j,k)$ is good.
\end{lemma}

\begin{proof}
  First observe that if $s_j = s_k$, then $\by_j = y_j(s_j) = y_k(s_k) = \by_k$, and we are done. For the rest of the proof, assume $s_j \not= s_k$.

  Since $(j,k) \in E_N$, there is a point in time $t' \in I_j$ with $k \in N(t')$ at which $k$ is being processed.  We distinguish between the two cases (1) $t' < s_j$ and (2) $t' \ge s_j$.

  \paragraph*{Case (1):} First, consider the case $t' < s_j$. Then, both jobs are non-clairvoyant at point in time $t'$, i.e., $j,k \in N(t')$. By \Cref{obs:clairvoyant-jobs-block-earlier-jobs} and our assumption that $s_j \not= s_k$, this means that the job which emits first is completed before the other job can emit. 
  In particular, this means that $s_j < s_k$ is not possible as otherwise we would have $j \in D(t)$.
  Thus, we must have $t' \le s_k < s_j$. Since $k,j \in N(t')$ and $k$ is executed at $t'$, we have $y_j(t') \ge y_k(t')$ by definition of the algorithm. Using $t' \le s_k < s_j$,~\Cref{lemma:ketchup} then implies $y_j(s_k) \geq y_k(s_k)$, giving $\alpha p_j \geq y_j(s_k) \geq y_k(s_k) = \alpha p_k$. Thus, $\by_k = y_k(s_k) = \alpha p_k \le \alpha p_j =  y_j(s_j) = \by_j$.

  \paragraph*{Case (2):} Next, consider the case $t' \ge s_j$. 
  If $p_k \le p_j$, then we immediately get $\by_k = y_k(s_k) = \alpha p_k \le \alpha p_j \le y_j(s_j) = \by_j$.

  If $p_k > p_j$, then we claim that $y_k(t) \le y_j(s_j)$ holds, which implies $\by_k = y_k(s_k) \le y_k(t) \le y_j(s_j) = \by_j$. For the remainder of the proof, we show that the claim indeed holds.

  To see that the claim $y_k(t) \le y_j(s_j)$ holds, assume otherwise, i.e., $y_k(t) > y_j(s_j)$. 
  Since we still have $y_k(t') \le y_j(t') \le y_j(s_j)$ by \Cref{obs:sept-like} (and using the fact that $k \in N(t')$ is executed at $t'$),  $y_k(t) > y_j(s_j)$ implies that there must be a point in time $t''$ with $t' \leq t'' \leq t$ and $y_k(t'') = y_j(s_j) = \alpha \cdot p_j$. By our assumption that $p_k > p_j$, we then have $y_k(t'') < \alpha p_k$ and, thus, $k$ is still non-clairvoyant at point in time $t''$. Furthermore, $y_k(t'') = \alpha \cdot p_j$ implies $\frac{1-\alpha}{\alpha} y_k(t'') = (1-\alpha) \cdot p_j$. However, as $t'' \geq t' \geq s_j$, we also have $p_j(t'') \le (1-\alpha) \cdot p_j$. 
  Thus, $\frac{1-\alpha}{\alpha} y_k(t'') \geq p_j(t'')$, which means that the algorithm will not work on $k$ after time $t''$ until $j$ is completed.
  Since $j \in C(t)$, $k$ is not executed during $[t'',t]$. However, using $k \in N(t'')$, this implies $k \in N(t)$; a contradiction to the assumption of the lemma that $k \in J \setminus N(t)$.
\end{proof}

\Cref{lem:nc:c:edge,lem:c:c:edge,lem:c:nc:edge} imply that each edge $(j,k) \in E_N$ with $j \in C(t) \cup N(t)$ is good. This only leaves edges $(j,k) \in E_N$ with $j \in D(t)$ to potentially be bad and the example in~\Cref{fig:ex:bad-edge} shows that such edges can indeed be bad.
The crux of the example is that $k$ is executed after the job $j$ emits. 
With the next lemma, we show that this is indeed a necessary condition for such an edge to be bad.
Later, this condition will be used to prove the existence of shortcuts that allow to skip bad edges.

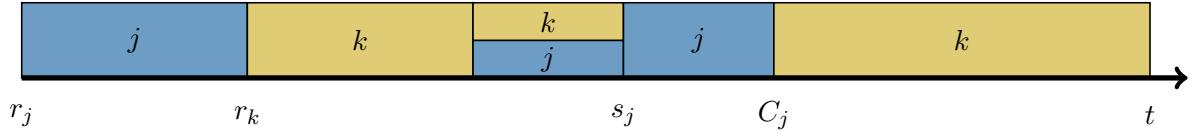
\begin{figure}[t]
  \centering
  \begin{tikzpicture}
    \jobT[job1](0,0)(3,1)($j$);
    \jobT[job2](3,0)(3,1)($k$);
    \node at (0,-0.5) {$r_j$};
    \node at (3,-0.5) {$r_k$};

    \jobT[job1](6,0)(2,0.5)($j$);
    \jobT[job2](6,0.5)(2,0.5)($k$);
    \node at (8,-0.5) {$s_j$};

    \jobT[job1](8,0)(2,1)($j$);
    \node at (10,-0.5) {$C_j$};

    \jobT[job2](10,0)(5,1)($k$);
    \node at (15,-0.5) {$t$};

    \draw[line width=2pt,->] (0,0) -- (15.5,0);
  \end{tikzpicture}
  \caption{Example of a job $j \in D(t)$ and $k \in N(t)$ with a bad edge $(j,k) \in E_N$. At point in time $s_j$, the jobs satisfy $y_j(s_j) = y_k(s_j)$ as they are processed in parallel. However, as $k$ is executed during $[s_j,t]$, we have $\by_j = y_j(s_j) = y_k(s_j) < y_k(t) = \by_k$, and the edge is bad.}
  \label{fig:ex:bad-edge}
\end{figure}

\begin{lemma}\label{lem:final:good:edge}
  Let $j \in D(t)$ and $(j,k) \in E_N$. Then, $(j,k)$ is good if at least one of the following conditions hold: 
  \begin{enumerate}[(i)]
    \item $k \in N(t)$ and $k$ is not executed during $[s_j,t]$.
    \item $k \not\in N(t)$ and $s_k \le s_j$.
  \end{enumerate}
\end{lemma}

\begin{proof}
  Since $(j,k) \in E_N$, there is a point in time $t' \in I_j$ with $k \in N(t')$ at which $k$ is being processed. 
  We distinguish between which of the two stated conditions holds:
  \begin{enumerate}[(i)]
      \item If $k \in N(t)$ and $k$ is not executed during $[s_j,t]$, then we must have $t' < s_j$. Thus, $j,k \in N(t')$, and $y_j(t') \geq y_k(t')$ holds by \Cref{obs:sept-like}. Since $k \in N(t)$, \Cref{lemma:ketchup} implies 
      $y_j(s_j) \geq y_k(s_j)$, and since $j$ is being processed at time $s_j$, $y_j(s_j) = y_k(s_j)$.
      Thus, using that $k$ is not executed during $[s_j,t]$, $\by_j = y_j(s_j) = y_k(s_j) = y_k(t) = \by_k$, which implies that $(j,k)$ is good.

      \item If $k \not\in N(t)$ and $s_k \le s_j$, we distinguish two subcases. If $s_k = s_j$, then we have $\by_j = y_j(s_j) = y_k(s_k) = \by_k$. This is because both jobs emitting at the same time implies that the algorithm worked on both at that time because of $y_j(s_j) = y_k(s_j)$. This implies that $(j,k)$ is good.
      
      If $s_k < s_j$, then $j,k \in N(t')$. %
      Since there exists the point in time $t'$ such that $j,k \in N(t')$, we have $y_j(t') \geq y_k(t')$ by \Cref{obs:sept-like}. Thus,
      \Cref{lemma:ketchup} and $s_k < s_j$ imply $\alpha p_j \geq y_j(s_k) \geq y_k(s_k) = \alpha p_k$.
      We can conclude with $\by_j = \alpha p_j \geq \alpha p_k = \by_k$. Thus, $(j,k)$ is good.
  \end{enumerate}
\end{proof}

The~\Cref{lem:nc:c:edge,lem:c:nc:edge,lem:c:c:edge,lem:final:good:edge} imply the following corollary, which characterizes bad edges.

\begin{corollary}
    \label{coro:bad-edges}
    Let $(j,k) \in E_N$. If $(j,k)$ is bad, then $j \in D(t)$ and at least one of the following conditions holds:
    \begin{enumerate}[(i)]
      \item $k \not\in N(t)$ and $s_k > s_j$.
      \item $k \in N(t)$ and $k$ is executed during $[s_j,t)$.
    \end{enumerate}
    In either case, there exists a point in time $t' \in [s_j,t]$ at which $k$ is being processed and $k \in N(t')$.
\end{corollary}

\paragraph*{$E_N$-paths that start in $A(t)$.} Having characterized good and bad edges in $E_N$, we proceed to prove~\Cref{lem:good-path-from-active}, which states that an $E_N$-path from $j \in A(t)$ to $i$ implies the existence of a good path from $j$ to $i$.

Since $j \in A(t)$, we know by \Cref{coro:bad-edges} that an $E_N$-path $P$ starting with $j$ must start with a good edge. If $P$ is not good and, therefore, contains at least one bad edge, then the path must contain two consecutive edges $e,e' \in E_N$ such that $e$ is good and $e'$ is bad. The next two auxiliary lemmas show that such a situation often implies the existence of a shortcut that skips the bad edge $e'$.    

\begin{lemma}
  \label{lem:shortcut:1}
  Let $(j,k),(k,h) \in E_N$ such that $(j,k)$ is good and $(k,d)$ is bad. If $j \in A(t)$, then $(j,h) \in E_N$.
\end{lemma}

\begin{proof}
  First observe that, since $j \in A(t)$, every job $d$ that is executed after $r_j$ is specifically executed during $I_j$, and thus, $(j,d) \in E$. 

As $(k,h)$ is bad, \Cref{coro:bad-edges} implies that $h$ is executed during $[s_k,t]$. 
Furthermore, $(j,k) \in E_N$ implies that there is a point in time $t'$ with $r_j \le t' \le s_k$ at which $k$ is executed. 
Thus, $h$ is executed after $r_j$ and, as argued above, there exists an edge $(j,h) \in E$. %
If $h \in N(t)$, then clearly $(j,h) \in E_N$. If $h \not\in N(t)$, then \Cref{coro:bad-edges} implies $s_h > s_k$. 
Thus, $h$ is executed after $s_k \ge r_j$ at time $s_h^-$ with $h \in N(s_h^-)$, and we can conclude that $(j,h) \in E_N$.
\end{proof}

\begin{lemma}
  \label{lem:shortcut:2}
  Let $(j,k),(k,h) \in E_N$ such that $(j,k)$ is good and $(k,h)$ is bad. If $j \in D(t)$ and $h$ is being processed at a time $t'$ such that $t' \in I_k$,  $t' \le C_j$, and $h \in N(t')$, then $(j,h) \in E_N$.
\end{lemma}

\begin{proof}
  Since $(k,h)$ is bad, $k \in D(t)$ by \Cref{coro:bad-edges}.
  Observe that a good edge $(j,k) \in E_N$ with $j,k \in D(t)$ must satisfy $p_k \le p_j$. If $p_j < p_k$, then $\by_j = \alpha \cdot p_j < \alpha \cdot p_k = \by_k$ and $(j,k)$ is bad; a contradiction.
  
  We continue by showing the statement via proof by contradiction. To this end, assume $(j,h) \not\in E_N$. This means that $h$ is not executed during $I_j$ while being non-clairvoyant. Thus, the point in time $t'$ of the lemma cannot be contained in $I_j= [r_j, C_j]$. Since $t' \le C_j$ by assumption of the lemma, we must have $t' < r_j$. Then, the other assumption of the lemma that $t' \in I_k$ implies $r_k \le t' < r_j$.

  Next, consider the job $k$. By assumption that $(j,k)  \in E_N$, the job $k$ must be executed at a point in time $t'' \in I_j$ while being non-clairvoyant. Since $r_k < r_j$, we must have $t'' \in [r_j,s_j]$; an execution of $k$ during $(r_j,s_j]$ would contradict~\Cref{obs:clairvoyant-jobs-block-earlier-jobs}. Also, $k$ being processed at $t'' > t' \ge r_k$ while being non-clairvoyant implies $k \in N(t')$.
  Using the fact that $(k,h)$ is bad,~\Cref{coro:bad-edges} either gives us $h \in N(t)$ or $s_h > s_k$.  Both cases imply $h \in N(t'')$.

  To finish the proof, note that we identified the two point in time $t' < t''$ with the following properties: (i) $k,h \in N(t')$ and $h$ is executed at $t'$ and (ii) $k,h \in N(t'')$ and $k$ is executed at $t''$. The existence of these points in time and~\Cref{lemma:ketchup} imply that $y_k(t'') = y_h(t'')$. However, by definition of the algorithms, this implies that $k$ and $h$ are executed in parallel at $t''$. Since $t'' \in I_j$ and $h \in N(t'')$, this implies that edge $(j,h) \in E_N$ exists; a contradiction.  
\end{proof}

Having established the two shortcut lemmas, we need one more minor observation to be ready to prove~\Cref{lem:good-path-from-active}.

\begin{observation}
  \label{obs:shortcut:3}
  Let $P$ be an $E_N$-path that starts at a job $j$ but otherwise only contains jobs $k \in D(t)$ with $C_k \le t'$. Then, $r_j \le t'$. 
\end{observation}

\begin{proof}
  Let $k$ be the direct successor of $j$ on $P$. By assumption $C_k\le t'$, and thus, $r_k \le t'$. The edge $(j,k) \in E_N$ implies that there is a point in time $t''$ with $r_j \le t'' \le C_j$ at which $k$ is executed. Clearly, $r_j \le t'' \le C_k \le t'$
\end{proof}

\begin{proof} [Proof of \Cref{lem:good-path-from-active}]
  Let $P$ be an $E_N$-path from $j \in A(t)$ to $i$, and assume that $P$ contains at least one bad edge. %
  We exhaustively apply the shortcut rules of \Cref{lem:shortcut:1} and \Cref{lem:shortcut:2} to $P$. Note that this never increases the number of bad edges. Let $P'$ be the resulting path. If $P'$ is good, we are done.

  Otherwise, $P'$ contains at least one bad edge. Let $(k_1,k_2)$ be the first bad edge on $P'$.
  By \Cref{coro:bad-edges}, we know $k_1 \in D(t)$. 
  Since $j \not\in D(t)$, $(k_1,k_2)$ cannot be the first edge on $P'$. Let $k_0$ be the predecessor of $k_1$ on $P'$. 
  Note that $(k_0,k_1)$ is good due to our choice of $(k_1,k_2)$.
  Furthermore, we know that the shortcuts of \Cref{lem:shortcut:1,lem:shortcut:2} do not apply to the edges $(k_0,k_1)$ and $(k_1,k_2)$. This implies $k_0 \in D(t)$ and that $k_2$ is executed at a time $t' \in I_{k_1}$ with $t' > C_{k_0}$. 

  Let $P''$ denote a subpath of $P'$ that ends at $k_0$, starts at a job $h$ with $C_h \ge t'$, and between $h$ and $k_0$ only contains jobs $k$ with $C_k < t'$. Since the job $j \in A(t)$ satisfies $C_j \ge t \ge t'$, such a subpath $P''$ must exist.
  By \Cref{obs:shortcut:3}, we have $r_h \le t' \le C_h$. As $k_2$ is executed at point in time $t'$ while being non-clairvoyant, we get $(h,k_2) \in E_N$. Thus, we can replace $P''$ and $(k_0,k_1),(k_1,k_2)$ in $P'$ with the shortcut $(h,k_2)$.

  As long as our path contains bad edges, we can use these shortcuts to find a shorter path. %
  Since each shortcut strictly decreases the number of edges on the path, the procedure terminates if the path is good or contains the single edge $(j,i)$. In the latter case, $j \in N(t)$ and \Cref{coro:bad-edges} imply that $(j,i)$ is good.
\end{proof}

\paragraph*{$E_N$-paths that start with an entrypoint.} 
With~\Cref{lem:good-path-from-active} we have shown that an $E_N$ path from $j$ to $i$ implies the existence of a good path from $j$ to $i$ \emph{if} $j \in A(t)$. Next, we want to show that the implication also holds for $E_N$-paths that start at an entrypoint $k \in H_x^i$. Since we may have $k \not\in A(t)$, this is not implied by~\Cref{lem:good-path-from-active}. Thus, our goal is to show the following lemma. Later, this lemma will be useful for proving that all jobs $j$, which can reach $i$ without visiting a bottleneck, satisfy $\by_j \ge \by_i$.

\begin{lemma}
  \label{lem:good-path-from-entrypoints}
  Let $k \in H_x^i$ be an entrypoint of layer $L_x^i$ and let $i' \in L_x^i$ be a job that $k$ can reach via an $E_N$-path. Then, there exists a good path from $k$ to $i'$, and thus, $\by_k \ge \by_{i'}$.
\end{lemma}

Our strategy to prove the lemma is similar to the proof of~\Cref{lem:good-path-from-active}: We want to show that, whenever the path contains a bad edge, we can find a shortcut that skips the bad edge. However, for the proof of~\Cref{lem:good-path-from-active} it is essential that the first edge on the $E_N$-path is good. If the first job on the path is in $A(t)$, then we immediately get this property by~\Cref{coro:bad-edges}. We continue by proving a series of auxiliary lemmas, which will show that this property also holds if the path starts with an entrypoint.

\begin{lemma}
  \label{lem:layers:entrypoints:1}
  Let $k \in H_{x}^i$ be an entrypoint of layer $L_x^i$, let $(j,k) \in \delta_x^i$ and let $(k,d) \in E_N$. Then, at least one of the following two properties holds:
  \begin{enumerate}
      \item $(k,d)$ is good,
      \item $d$ is executed during $[C_j, t]$ while being non-clairvoyant.
  \end{enumerate}
\end{lemma}

\begin{proof}
  Note that $d \in L_{x}^i$ and $j \not\in L_{\leq x}^i$. Thus, $d$ cannot be executed during $[r_j,C_j)$ while being non-clairvoyant, as the edge $(j,d) \in E_N$ would lead to a contradiction to the layer definition.
  Furthermore, by \Cref{lem:ec:cuts:prop:1}, we have $r_k \le s_k \le r_j \le C_k \le C_j$.

\begin{figure}[t]
  \centering
  \begin{tikzpicture}
      
      \draw[->,thick] (-1,0) -- (12,0);
      \node at (12.75,0) {time};

      \draw[thick,dashed,line width = 1.5pt] (1,0.75) -- (1,-0.75);
      \node at (1,-1) {$r_{k}$};

      \draw[thick,dashed,line width = 1.5pt] (3,0.75) -- (3,-0.75);
      \node at (3,-1) {$s_{k}$};

      \draw[thick,dashed,line width = 1.5pt] (4.5,0.75) -- (4.5,-0.75);
      \node at (4.5   ,-1) {$r_{j}$};

      \draw[thick,dashed,line width = 1.5pt] (8,0.75) -- (8,-0.75);
      \node at (8,-1) {$C_{k}$};

      \draw[thick,dashed,line width = 1.5pt] (9.5,0.75) -- (9.5,-0.75);
      \node at (9.5,-1) {$C_{j}$};

      \node at (2,1) {Case (1)};
      \node at (3.75,1) {Case (2)};

      \fill[pattern=north west lines, pattern color=red] (4.5,-0.3) rectangle (11.95,0.3);
      \fill[pattern=north west lines, pattern color=red] (-1,-0.3) rectangle (1,0.3);
      \fill[pattern=north west lines, pattern color=green] (1,-0.3) rectangle (4.5,0.3);
    \end{tikzpicture}
  \caption{Illustrates the possible values of $t'$ in the proof of~\Cref{lem:layers:entrypoints:1}: The red area indicates values that are not possible while the green area indicates values that are indeed possible.}
  \label{fig:layers:entrypoints:1}
\end{figure}
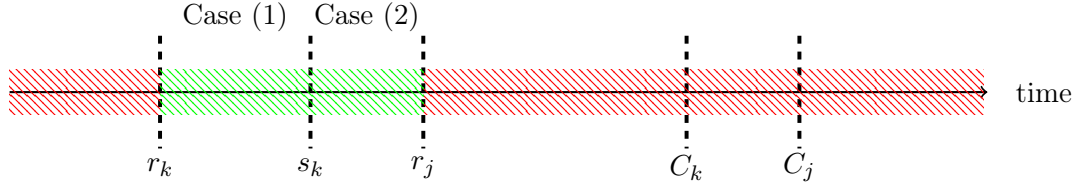

  The edge $(k,d) \in E_N$ implies that $d$ is being processed at some $t' \in I_k$ while being non-clairvoyant. By the observations above, we can distinguish between (1) $t' \in [r_k,s_k)$ or (2) $t' \in [s_k,r_j)$. \Cref{fig:layers:entrypoints:1} illustrates the two cases.

  \paragraph*{Case (1)} We have that $d$ is executed during $[r_k,s_k)$ while being non-clairvoyant.

  If $d$ emits before or at time $s_k$, i.e., $s_d \leq s_k$, then we must have $r_k \le s_d \le s_k$ by the assumption that $d$ is executed during $[r_k,s_k)$ while being non-clairvoyant. Thus, $\by_d = \alpha \cdot p_d= y_d(s_d) \le y_k(s_d) \le \alpha \cdot p_k = \by_k$ holds, because $k$ is non-clairvoyant at time $s_d$. Hence, $(k,d)$ is good and the first property of the lemma is satisfied.

  For the remainder of the proof of Case (1), assume that $d$ does not emit before or at $s_k$,~i.e., $s_d > s_k$. Since $d$ is executed at $t' \in [r_k,s_k)$ with $d,k \in N(t')$, we have $y_d(t') \leq y_k(t')$ by~\Cref{obs:sept-like}. On the other hand, $k$ is executed at point in time $s_k^-$ with $k,d \in N(s_k^-)$, so~\Cref{obs:sept-like} and~\Cref{lemma:ketchup} imply $y_d(s_k^-) = y_k(s_k^-)$ and $y_d(s_k) = y_k(s_k)$.
  Thus,  $y_d(s_k) = y_k(s_k) = \alpha p_k = \by_k$. 
  
  If $d$ is not executed again before $t$, i.e.~during $[s_k,t]$, then this implies $\by_d = y_d(t) = y_d(s_k) \le \by_k$ and $(k,d)$ is good.  
  If $d$ is executed again before $t$, i.e., during $[s_k,t]$, then \Cref{obs:clairvoyant-jobs-block-earlier-jobs} implies that this must happen after $C_k$. Furthermore, as $d$ has not emitted yet, i.e., $d \in N(s_k)$, the first execution of $d$ during $[C_k,t]$ must be non-clairvoyant. 
  Moreover, using the observation from above that $d$ cannot be executed during $[r_j, C_j)$ and $C_k \in[r_j,C_j)$, we can conclude that this execution must actually take place after $C_j$.
  Thus, $d$ is executed during $[C_j, t]$ and the second property of the lemma is satisfied.

  \paragraph*{Case (2)} By assumption of this case, we have that $d$ is executed during $[s_k,r_j)$ while being non-clairvoyant. Since $r_j \leq C_k$, we also have $s_k \leq r_d < r_j$ by \Cref{obs:clairvoyant-jobs-block-earlier-jobs}. We distinguish between the two subcases (a) $s_d < r_j$ and (b) $s_d \ge r_j$:
  \begin{enumerate}
    \item[\textbf{Subcase (a)}:]   If $d$ emits before $r_j$, that is, $s_d < r_j \le C_k$, then we must have $s_k \le s_d < r_j \le C_k$ as we already argued that $s_k \leq r_d$.

    If $s_d > s_k$, then the definition of our algorithm (tiebreaking rule for clairvoyant jobs) implies that $d$ completes before $k$. In particular, there must be a point in time during $[s_d,C_d] \subset[s_k,C_k]$ at which $k$ and $d$ are both clairvoyant and the algorithm executes $d$ but not $k$. By \Cref{obs:srpt-like}, this implies $p_d \le p_k$, and thus, $\by_d = \alpha \cdot p_d \le \alpha \cdot p_k = \by_k$. Hence, the edge $(k,d)$ is good.
    
    If $s_d = s_k$, then the algorithm must executed $d$ and $k$ in parallel at $s_d^- = s_k^-$. Thus, by~\Cref{obs:sept-like}, $y_k(s_k^-) = y_d(s_d^-)$ and, therefore, $\by_k = \alpha \cdot p_k = \alpha \cdot p_d = \by_d$. Hence, the edge $(k,d)$ is good.

    \item[\textbf{Subcase (b)}:] If $d$ does not emit before $r_j$, then there is a latest point in time $t'' \le r_j \le s_d$ at which $d$ is executed. As argued at the beginning of Case (2), we also have $s_k \le r_d \le t''$. This implies $k \in C(t'')$ and, thus, $k$ is not executed at point in time $t''$ as clairvoyant and non-clairvoyant jobs are never executed in parallel.
    
    Since $d$ is non-clairvoyant at time $t''$ and $k$ is not being processed at time $t''$, we have $\frac{1-\alpha}{\alpha} \cdot y_d(t'') \le p_k(t'') \le (1-\alpha) \cdot p_k$. This implies $y_d(t'') \le \alpha \cdot p_k = \by_k$. As argued at the very beginning of the proof, $d$ is not executed during $[r_j,C_j]$ while being non-clairvoyant, so $y_d(t'') \le \alpha \cdot p_k = \by_k$ and $d \in N(t'')$ imply $y_d(C_j) \le \alpha \cdot p_k = \by_k$
    
    If $d$ is not executed during $[C_j,t]$, then $\by_d = y_d(C_j) \le \by_k$ and the edge $(k,d)$ is good. Hence, the first property of the lemma is satisfied.
    If $d$ is executed during $[C_j,t]$, then it must also be executed during $[C_j,t]$ while being non-clairvoyant. Thus, the second property of the lemma holds.
  \end{enumerate}  
\end{proof}

\Cref{lem:layers:entrypoints:1} does not immediately show that the first edge on an $E_N$-path that starts with an entrypoint $k$ is good, as the second property of the lemma might hold. However, with the following lemma we show that the second property leads to a contradiction and, thus, cannot actually happen.

\begin{lemma}
  \label{lem:layers:entrypoints:2}
  Let $k \in H_x^i$ be an entrypoint of layer $L_x^i$ and $P$ be an $E_N$-path from $k$ to some job $i' \in L_x^i$. Then, the first edge on path $P$ is good.
\end{lemma}

\begin{proof}
    Let $(k,d)$ denote the first edge on path $P$. Since $k$ is an entrypoint, there must be an $x$-entry path $P'$ from some job $j \in A(t) \cap L_{\ge x+1}$ to $k$. Such a path $P'$ must end with an edge $(j',k) \in \delta_x^i$.      
    This allows us to apply~\Cref{lem:layers:entrypoints:1} to the edges $(k,d)$ and $(j',k)$ to conclude that either $(k,d)$ is good, or $d$ is executed during $[C_{j'},t]$ while being non-clairvoyant. If the former holds, then this immediately implies this lemma. Therefore, assume that $d$ is executed during $[C_{j'},t]$ while being non-clairvoyant.

    Let $P''$ denote the $j$-$j'$-subpath of $P'$. Since $j \in A(t)$,~\Cref{obs:lifetime:2} implies $I(P'') \supseteq [r_{j'},t]$. Thus, $d$ is executed during $I(P'') \supseteq [C_{j'}, t]$ while being non-clairvoyant.  This implies that $d \in L_x^i$ is executed during $I_h$ for some job $h \in P''$ while being non-clairvoyant.
    However, we have $h \in P'' \subseteq L_{\ge x+1}^i$ as $P'$ is an $x$-entry path and $d \in P \subseteq L_{x}^i$ by assumption of the lemma. Thus, the non-clairvoyant execution of $d$ during $I_h$ is a contradiction to the layer definition.
\end{proof}

Having established that the first edge on an $E_N$-path that starts with an entrypoint is good, we are ready to prove~\Cref{lem:good-path-from-entrypoints} using a similar proof to~\Cref{lem:good-path-from-active}.

\begin{proof}[Proof of~\Cref{lem:good-path-from-entrypoints}]
  If $k = i'$, then the lemma holds trivially, so assume $k \not= i'$.
  By assumption, there exists an $E_N$-path $P'$ from $k$ to $i'$. %
  Let $d$ be the direct successor of $k$ on $P'$. 
  By \Cref{lem:layers:entrypoints:2}, we have that $(k,d)$ is good.
  
    We exhaustively apply the shortcuts of \Cref{lem:shortcut:1,lem:shortcut:2} to $P'$. Note that this never increases the number of bad edges on the path.
     Furthermore, note that the first edge on the path remains good by \Cref{lem:layers:entrypoints:2}, as it always starts with job $k$, even if it is changed by applying the shortcuts.
   If the resulting path $P''$ is good, we are done. 
  
  Thus, assume $P''$ contains at least one bad edge, and let $(k_1,k_2)$ be the first bad edge on $P''$. Recall that this edge is in $E_N$.
  By \Cref{coro:bad-edges}, we have $k_1 \in D(t)$. 
  Since we already established that $(k,d)$ is good, $(k_1,k_2)$ cannot be the first edge on $P'$. 
  Let $k_0$ be the predecessor of $k_1$ on $P''$. By assumption that $(k_1,k_2)$ is the first bad edge on $P''$, we have that $(k_0,k_1)$ is good. Furthermore, we know that the shortcuts of \Cref{lem:shortcut:1,lem:shortcut:2} do not apply to the edges $(k_0,k_1)$ and $(k_1,k_2)$. This implies $k_0 \in D(t)$ and that $k_2$ is executed during  $t' \in I_{k_1}$ such that $k_2 \in N(t')$ and $t' > C_{k_0}$. 

  We claim that, if a situation like this occurs, then there must exist at least one job $h$ that is a predecessor of $k_1$ on $P''$ and satisfies $C_h \ge t'$. 
  We show this claim via proof by contradiction. To this end, assume the claim is not true, i.e., that all jobs $h$ that are predecessors of $k_1$ on $P''$ have $C_h < t'$. In particular, this gives us $C_k < t'$. 
  Since $k$ is an entrypoint of layer $L_x^i$, there must be a job $j \in A(t) \setminus L_{\leq x}^i$ that can reach $k$ via an $x$-entry path $P$. %
   Since $j \in J_A(t)$, $P$ certainly contains at least one job $h'$ with $C_{h'} \ge t'$. Let $h'$ denote the latest such job on path $P$. That is, all successors of $h'$ on $P$ complete before $t'$. Since we already established $C_k < t'$, we have $h' \not= k$. Since $P$ is an $x$-entry path, this implies $h' \not\in L_{\leq x}^i$.  By \Cref{obs:shortcut:3}, we get $r_{h'} \le t' \le C_{h'}$. 
   Thus, $t' \in I_{h'}$, and, since $k_2$ is being processed at time $t'$ and $k_2 \in N(t')$, we conclude $(h',k_2) \in E_N$, which is a contradiction to $h' \not\in L_{\leq x}^i$ and $k_2 \in L_{\leq x}^i$.
   Note that $k_2 \in L_{\leq x}^i$ holds by the definition of layers as $k_2$ can reach the job $i' \in L_x^i$ via an $E_N$-path.

  With the claim in place, we know that there must be a predecessor $h$ of $k_2$ with $C_h \ge t'$. Picking the latest such predecessor on $P''$ allows us to use \Cref{obs:shortcut:3} again to obtain $r_h \le t'$. As before, this implies that the edge $(h,k_2)$ is in $E_N$. We can apply this shortcut to reduce the number of edges on the path $P''$ by at least one. 
  As long as our path contains bad edges, we can use these shortcuts to find a shorter path. %
  Since each shortcut strictly decreases the number of edges on the path, the procedure terminates if the path is good or contains the single edge $(k,i')$. In the latter case, \Cref{lem:layers:entrypoints:2} implies that $(k,i')$ is good.
\end{proof}

\subsection{Borrowing via Multiple Layers without Bottlenecks}
\label{app:no-bottleneck-borrowing}

The goal of this section is to prove that if a job $j \in N(t)$ can reach a job $i$ in the borrow graph $G_B$ without visiting any bottlenecks, then $\by_j \ge \by_i$. This statement is particularly relevant in the context of the segments (cf.~\Cref{def:segment}) as it implies that each $i \in O(t)$ that can be reached by a job $j$ that is part of a segment $S$ must satisfy $i \in O_S$. Vice versa, if a job $j$ of a segment $S$ can reach a job $i \in O(t) \setminus O_S$, then $j$ must visit at least one bottleneck in order to reach $i$.

The following theorem is the main result of this section.

\bottleneckFreePaths*

To prove it, we heavily rely on the following key lemma, which we state here but prove later in~\Cref{sec:key:lemma:proof}.

\begin{restatable}{lemma}{lemkey}
  \label{lem:multiple:entrypoints}
  If a job $j \in (A(t) \cap L^i_{x+1}) \cup H^i_{x+1}$ can reach at least two distinct entrypoints of $L^i_x$ via $x$-entry paths, then $\by_j \ge \by_k$ for at least one $k \in H^i_x$.
\end{restatable}

To prove the theorem, we first introduce some more notation.
Fix a job $i \in O(t)$ and consider a job $j \in A(t)$. We use $\ell_{ji}$ to denote the index of the layer $L_{x}^i$ with $j \in L_{x}^i$. If every path from $j$ to $i$ contains at least one bottleneck of $i$, then we use $\sigma_{ji}$ to denote the index of the first layer that $j$ cannot reach without visiting a bottleneck, i.e., every path from $j$ to some $i' \in R_j \cap L^i_{\le \sigma_{ji}}$ visits at least one bottleneck of $i$ and for every layer $L_x^i$ with $\sigma_{ji} < x \le \ell_{ji}$ there exists an entrypoint $k \in H_x^i$ that $j$ can reach using a path without bottleneck of $i$. We use $b_{ji}$ %
to refer to the unique bottleneck (cf.\ \Cref{thm:bottleneck:uniqueness}) of layer $L_{\sigma_{ji}}^i$. %

Using this notation, we can show the following lemma and its corollary, which make proving \Cref{thm:yinequality:no:bottlenecks} straightforward.

\begin{lemma}
  \label{lem:aux:inequality}
  Let $j \in A(t)\setminus L^i_0$ and let $L^i_y$ with $y < \ell_{j}$ be a layer such that $j$ can reach at least one entrypoint $k \in H^i_y$ without visiting a bottleneck.
  Then, for each layer $L^i_x$ with $y \le x \leq \ell_j - 1$, there exists an entrypoint $k \in H^i_x$ with $\by_j \ge \by_k$.
\end{lemma}

\begin{proof}
  Via induction over $x \in \{\ell_{ji}-1, \ldots, y\}$.

  \textbf{Base case.} For the base case, consider layer $L^i_x$ with $x = \ell_{ji}-1$. This is the first layer on the path from $j$ to $i$ that does not contain $j$ itself. 
  
  If $L^i_x$ has a bottleneck, then $j$ must be able to reach at least two distinct
  entrypoints of $L^i_x$ via $x$-entry paths. Otherwise, \Cref{lem:unique:bottleneck:1} and the definition of bottlenecks would imply that $j$ can only reach $L^i_x$ via the bottleneck. The fact that $j$ can reach two distinct entrypoints of $L^i_x$ via $x$-entry paths allows us to apply \Cref{lem:multiple:entrypoints} to conclude that $j$ can reach an entrypoint $k$ with $\by_j \ge \by_k$. 

  If $L^i_x$ has no bottleneck, then we must have $|H^i_x| > 1$ and each job in $A(t) \cap L^i_{x+1}$ must be able to reach layer $L^i_x$ via at least two entrypoints of $L^i_x$. This allows us to apply \Cref{lem:multiple:entrypoints} again to conclude that $j$ can reach an entrypoint $k \in H^i_{x}$ with $\by_j \ge \by_k$. 

  \textbf{Induction step.} Consider a layer $L^i_x$ with $x \leq \ell_{ji} - 2$. By induction hypothesis, there exists an entrypoint $k_{x+1} \in H^i_{x+1}$ with $\by_j \ge \by_{k_{x+1}}$.  
  
  Furthermore, \Cref{thm:multiple:bottlenecks} implies that $j$ must visit all bottlenecks of layers $L_r^i$ with $r \leq \ell_{ji}-2$ to reach $i$. 
  By our assumption that $j$ can reach an entrypoint of layer $L^i_y$ without passing a bottleneck, this means that $L^i_x$ cannot have a bottleneck; otherwise $j$ could reach $i$ without visiting the bottleneck of $L^i_x$.
  
  By definition of bottlenecks, this means that $k_{x+1}$ can reach multiple entrypoints of $H^i_{x}$ via $x$-entry paths. \Cref{lem:multiple:entrypoints} implies that there exists an entrypoint $k_x$ of $L^i_x$ with $\by_{k_{x+1}} \ge \by_{k_x}$. As we already know that $\by_j \ge \by_{k_{x+1}}$ by induction hypothesis, we get $\by_j \ge \by_{k_x}$.
\end{proof}

If $j$ is a job such that every path from $j$ to $i$ contains at least one bottleneck of $i$, then we can plug $y = \sigma_{ji}+1$ into \Cref{lem:aux:inequality} to obtain the following corollary.

\begin{corollary}
  \label{lem:aux:inequality:1}
  Let $j \in A(t)$ be a job such that each path from $j$ to $i$ in $G_B$ contains a bottleneck of $i$. Then, for each layer $L^i_x$ with $\sigma_{ji} +1 \leq x \leq \ell_{ji} - 1$, there exists an entrypoint $k \in H^i_x$ with $\by_j \ge \by_k$.
\end{corollary}

Having established~\Cref{lem:aux:inequality}, the proof of~\Cref{thm:yinequality:no:bottlenecks} becomes quite straightforward.

\begin{proof}[Proof of \Cref{thm:yinequality:no:bottlenecks}]
  We distinguish between the two cases (1) $\ell_{ji} = 0$ and (2) $\ell_{ji} > 0$.

  \textbf{Case (1):} If $j \in A(t) \cap L^i_0$, then $j$ can reach $i$ by only using edges in $E_N$. \Cref{lem:good-path-from-active} implies that 
  $\by_j \ge \by_i$.

  \textbf{Case (2):} If $\ell_{ji} \geq 1$, then \Cref{lem:aux:inequality} implies that there is an entrypoint $k$ of $L^i_0$ with $\by_j \ge \by_k$. 
  Furthermore, \Cref{lem:good-path-from-entrypoints} implies that 
  $\by_k \ge \by_i$.
  Combining both inequalities yields $\by_j \ge \by_i$. 
\end{proof}

\subsubsection{Proof of~\Cref{lem:multiple:entrypoints}}
\label{sec:key:lemma:proof}

The main goal of this section is to prove \Cref{lem:multiple:entrypoints}, which states that if a job $j \in (A(t) \cap L^i_{x+1}) \cup H^i_{x+1}$ can reach layer $L^i_x$ without using a bottleneck, then it satisfies $\by_j \ge \by_k$ for at least one entrypoint $k \in H^i_x$.

\lemkey*

To prove the lemma, we rely on several observations and auxiliary lemmas. The first observation implies the lemma for the special case where there is an interruption of $j$ by an entrypoint $k_2$ with some special properties and will be used in most of the subsequent proofs.

\begin{observation}
    \label{obs:ec:cut:aux:1}
    Let $k_1, k_2 $ be two entrypoints of any layer (possibly different ones) and let $t'$ denote a point in time such that $k_2$ is executed at $t'$, $k_2 \in C(t')$, and $p_{k_2}(t') \ge (1-\alpha) \cdot \min\{p_{k_1},p_{k_2}\}$. Then, every job $j \in A(t')$ satisfies $\by_j \ge \min\{\by_{k_1},\by_{k_2}\}$.
\end{observation}

\begin{proof}
    We distinguish between the two cases (i) $j \in N(t')$ and (ii) $j \in C(t')$.

    \paragraph*{Case (i):} If $j \in N(t')$, then the definition of our algorithm and our assumption that $k_2$ is executed at $t'$ imply:
    \begin{align*}
        \frac{1-\alpha}{\alpha} y_{j}(t') \ge p_{k_2}(t') \ge (1-\alpha) \cdot \min\{p_{k_1},p_{k_2}\} \ ,
    \end{align*}
    which yields $y_j(t') \ge \alpha \cdot \min\{p_{k_1},p_{k_2}\}$.
    Since $k_1$ and $k_2$ are entrypoints, we must have $k_1,k_2 \not\in N(t)$, and thus, $\by_{k_1} = \alpha \cdot p_{k_1}$ and $\by_{k_2} = \alpha \cdot p_{k_2}$. Thus, $\alpha \cdot \min\{p_{k_1},p_{k_2}\} = \min\{\by_{k_1},\by_{k_2}\}$. Furthermore, as $j \in N(t')$, we have $y_j(t') \le \alpha p_j$ and $y_j(t') \le y_j(t)$. Thus, $\by_j \ge y_j(t')$. We conclude with
    $$
    \by_j \ge y_j(t') \ge \alpha \cdot \min\{p_{k_1},p_{k_2}\} = \min\{\by_{k_1},\by_{k_2}\} \ .
    $$

    \paragraph*{Case (ii):} If $j \in C(t')$, then the definition of our algorithm and the assumption that $k_2$ is executed at $t'$ imply
    $$
    p_j(t') \ge p_{k_2}(t') \ge (1-\alpha) \cdot \min\{p_{k_1},p_{k_2}\} \ .
    $$
    By assumption that $j \in C(t')$, we have $(1- \alpha) \cdot p_j \ge p_j(t')$.
    By plugging this into our initial inequality, we get
    \begin{align*}
        (1-\alpha) \cdot p_j \ge p_j(t') \ge p_{k_2}(t') \ge (1-\alpha) \cdot \min\{p_{k_1},p_{k_2}\} \ , %
    \end{align*}
    which yields $\alpha p_j \ge \alpha \cdot \min\{p_{k_1},p_{k_2}\}$, and thus, $\by_j \ge \min\{\by_{k_1}, \by_{k_2}\}$. 
\end{proof}

\begin{figure}[t]
    \centering
    \begin{tikzpicture}
        
        \draw[thick,dashed,red,line width = 1.5pt] (0,0) -- (0,-2.5);
        \node at (0,0.5) {\large\textcolor{red}{$\delta^i_x$}};

        \node at (-1.5,0.5) {\large $L^i_{x+1}$};
        \node at (1.5,0.5) {\large $L^i_{x}$};
        
        \node[vertex,label=above:{$j$}] (v0) at (-1.5,-1.25) {};

        \node[vertex,label=above:{$k_1$}] (v1) at (1.5,-0.75) {};
        \node[vertex,label=below:{$k_2$}] (v2) at (1.5,-1.75) {};

        \path 
            (v0) edge[->,thick] (v1)
            (v0) edge[->,thick] (v2);

    \end{tikzpicture}
    \caption{The situation considered in \Cref{lem:multiple:entries:1}.}
    \label{fig:lem:multiple:entries:0}
\end{figure}
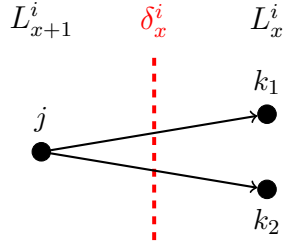

Having this observation in place, we are ready to show that if a job $j \in (A(t) \cap L^i_{x+1}) \cup H^i_{x+1}$ has edges to two distinct entrypoints $k_1,k_2 \in H_x^i$ of layer $L_x^i$ (cf.~\Cref{fig:lem:multiple:entries:0}), then it satisfies $\by_j \geq \by_{k_1}$ or $\by_j \geq \by_{k_2}$. On the one hand, this proves a special case of \Cref{lem:multiple:entrypoints} and, on the other hand, this will be useful as a kind of base case for the general proof.

\begin{lemma}
    \label{lem:multiple:entries:1}
    Let $j \in (A(t) \cap L^i_{x+1}) \cup H^i_{x+1}$ and $k_1,k_2 \in H^i_x$ such that $(j,k_1),(j,k_2) \in \delta^i_x$. Then, $\by_j \ge \min\{\by_{k_1}, \by_{k_2}\}$. 
\end{lemma}

\begin{proof}
    Let $t_d$ with $r_j \le t_d \le C_j$ be the earliest point in time such that $k_d$ is executed at $t_d$ and $k_d \in C(t_d)$, for $d \in \{1,2\}$. Since $(j,k_1),(j,k_2) \in \delta^i_x$, such points in time must exist due to \Cref{lem:ec:cuts:prop:1}.  
 
    As two clairvoyant jobs are never executed in parallel by definition of our algorithm, we have $t_1 \not= t_2$. Assume w.l.o.g.~that $t_1 < t_2$.
    Since $(j,k_1),(j,k_2) \in \delta^i_x$, \Cref{lem:ec:cuts:prop:1} implies $r_{j} \ge s_{k_1}$ and $C_j \ge C_{k_1}$ (see~\Cref{fig:multiple:entries:1} for an illustration).

    Next, consider the job $k_2$. We want to narrow down the possible values of $s_{k_2}$ in order to afterwards complete the proof via case distinction over these values. 
    To this end, we first show which values of $s_{k_2}$ are \emph{not} possible:

    \begin{enumerate}[(a)]
        \item  Since $(j,k_2)\in \delta_x$, we have $(j,k_2)\not\in E_N$. This directly implies that $s_{k_2} \not\in (r_j,C_j]$; otherwise $r_j \le s_{k_2}^- \le C_j$ which implies $(j,k_2) \in E_N$ as $k_2$ is non-clairvoyant and executed at $s_{k_2}^-$.
        \item Similarly, we also get $s_{k_2} \not\in (s_{k_1},r_j)$. If $k_2$ would emit during this interval, then $s_{k_1} < s_{k_2} < r_j < C_{k_1}$ and the definition of our algorithm (tiebreaking rule for clairvoyant jobs) would imply that $k_2$ finishes before $k_1$. This, however, means that $k_1$ is not executed during $[s_{k_2},C_{k_2}] \supseteq [r_j,C_{k_2}]$, which in turn implies $t_2 < t_1$; a contradiction to our assumption. 
        \item  Next, we claim $s_{k_2} \not\in (r_{k_1},s_{k_1})$. If $k_2$ would emit during this interval, then \Cref{obs:clairvoyant-jobs-block-earlier-jobs} would imply that $k_2$ needs to complete before $k_1$ can emit. However, then $k_2$ would finish before $r_j$, which contradicts the existence of the edge $(j,k_2) \in \delta^i_x$.
        \item Finally, observe that $s_{k_2} < C_j$, because $s_{k_2} \geq C_j$ would contradict the existence of $(j,k_2) \in E_C$.
    \end{enumerate}

    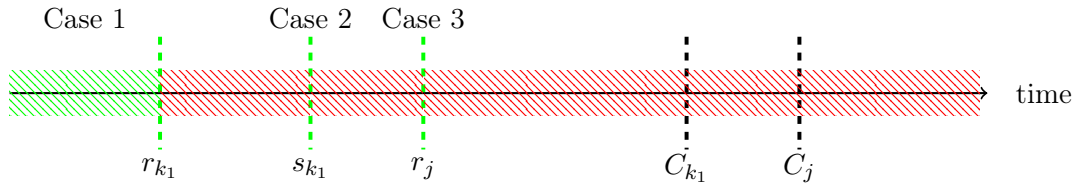
\begin{figure}[b]
        \centering
        \begin{tikzpicture}

            \draw[->,thick] (-1,0) -- (12,0);
            \node at (12.75,0) {time};

            \draw[thick,dashed,green,line width = 1.5pt] (1,0.75) -- (1,-0.75);
            \node at (1,-1) {$r_{k_1}$};

            \draw[thick,dashed,green,line width = 1.5pt] (3,0.75) -- (3,-0.75);
            \node at (3,-1) {$s_{k_1}$};

            \draw[thick,dashed,green,line width = 1.5pt] (4.5,0.75) -- (4.5,-0.75);
            \node at (4.5   ,-1) {$r_{j}$};

            \draw[thick,dashed,line width = 1.5pt] (8,0.75) -- (8,-0.75);
            \node at (8,-1) {$C_{k_1}$};

            \draw[thick,dashed,line width = 1.5pt] (9.5,0.75) -- (9.5,-0.75);
            \node at (9.5,-1) {$C_{j}$};

            \fill[pattern=north west lines, pattern color=green] (-1,-0.3) rectangle (1,0.3);
            \fill[pattern=north west lines, pattern color=red] (4.51,-0.3) rectangle (11.9,0.3);
            \fill[pattern=north west lines, pattern color=red] (3.01,-0.3) rectangle (4.49,0.3);
            \fill[pattern=north west lines, pattern color=red] (1.01,-0.3) rectangle (2.99,0.3);

            \node at (0,1) {Case 1};
            \node at (3,1) {Case 2};
            \node at (4.5,1) {Case 3};

        \end{tikzpicture}
        \caption{Illustrates the possible (green) and impossible (red) values of $s_{k_2}$ in the proof of \Cref{lem:multiple:entries:1}}
        \label{fig:multiple:entries:1}
    \end{figure}

    The red part of~\Cref{fig:multiple:entries:1} illustrates the values of $s_{k_2}$ that are \emph{not} possible by the arguments above. This leaves the following three cases (green part in~\Cref{fig:multiple:entries:1}):

    \begin{enumerate}
        \item $s_{k_2} = r_j$,
        \item $s_{k_2} = s_{k_1}$, or
        \item $s_{k_2} \le r_{k_1}$.
    \end{enumerate}
    We finish the proof by distinguishing between these three cases.

    \paragraph*{Case 1:} Assume $s_{k_2} = r_j$. 
    In this case, we have $y_{k_2}(r_j)= \alpha \cdot p_{k_2}$. Recall that $t_2 \ge r_j$ is the earliest point during $I_j$ at which $k_2$ is being processed. At this point in time, we still have $y_{k_2}(t_2)= \alpha \cdot p_{k_2}$, and thus, $p_{k_2}(t_2) = (1-\alpha) \cdot p_{k_2}$. This allows us to apply \Cref{obs:ec:cut:aux:1} to conclude with $\by_j\ge \min\{\by_{k_1},\by_{k_2}\}$.

    \paragraph*{Case 2:} Assume $s_{k_2} = s_{k_1}$. In this case, $k_1$ and $k_2$ must be executed in parallel immediately before point in time $s_{k_2} = s_{k_1}$. This implies $y_{k_1}(s_{k_1}) = y_{k_2}(s_{k_2})$, and thus, $\alpha \cdot p_{k_1} = \alpha \cdot p_{k_2}$. By our assumption that $k_1$ 
    is being processed before $k_2$ during $I_j$, 
    the algorithm must finish $k_1$ before $k_2$. This means that $k_2$ is not executed during $[s_{k_2},t_2]$ (recall that $t_2$ is the earliest point in time at which $k_2$ interrupts $j$). Thus, $p_{k_2}(t_2) = p_{k_2}(s_{k_2}) = (1-\alpha) \cdot p_{k_2}$.

    The fact that  $p_{k_2}(t_2) = (1-\alpha) \cdot p_{k_2}$ allows us to apply \Cref{obs:ec:cut:aux:1} to obtain $\by_j\ge \min\{\by_{k_1},\by_{k_2}\}$. 

    \paragraph*{Case 3:} Assume $s_{k_2} \le r_{k_1}$. By \Cref{lem:ec:cuts:prop:1}, we have $r_{j} \ge s_{k_1}$. This means that $k_2$ must still be alive at $s_{k_1}$, as otherwise cannot be executed during $I_j$, which would contradicts $(j,k_2) \in E_C$. By our assumption that $k_1$ is executed during $I_j$ before $k_2$, we must have $p_{k_2}(s_{k_1}) \ge p_{k_1}(s_{k_1}) = (1-\alpha) \cdot p_{k_1}$; otherwise $k_2$ would either complete before $r_j$ or have priority over $k_1$ during $I_j$, both options contradict our assumptions. For the same reason, $k_2$ cannot be executed between $s_{k_1}$ and $r_j$, which implies $p_{k_2}(s_{k_1}) = p_{k_2}(r_j) = p_{k_2}(t_2)$.

    Consider now the earliest point in time $t_2 \in I_j$ at which $k_2$ is being processed. At this point, we still have  $p_{k_2}(t_2) = p_{k_2}(s_{k_1}) \ge (1-\alpha) \cdot p_{k_1}$.
    This allows us to apply \Cref{obs:ec:cut:aux:1} to conclude with $\by_j\ge \min\{\by_{k_1},\by_{k_2}\}$.
\end{proof}

The next auxiliary lemma generalizes \Cref{lem:multiple:entries:1} by dropping the assumption that $j$ has edges to $k_1$ and $k_2$ and instead only demanding an $x$-entry path to some job $j'$ (in the lemma below either $j_1$ or $j_2$) that has edges to $k_1$ and $k_2$. For technical reasons that will be useful in the final proof of~\Cref{lem:multiple:entrypoints}, we phrase the lemma in a slightly different way. See~\Cref{fig:lem:multiple:entries:2} for an illustration of the situations considered in~\Cref{lem:multiple:entries:2}

\begin{figure}[b]

    \begin{subfigure}{0.325\textwidth}
        \centering
        \begin{tikzpicture}[scale=0.9,transform shape]
        
            \draw[thick,dashed,red,line width = 1.5pt] (0,0) -- (0,-2.5);
            \node at (0,0.5) {\large\textcolor{red}{$\delta^i_x$}};
    
            \node at (-2,0.5) {\large $L^i_{x+1}$};
            \node at (1.5,0.5) {\large $L^i_{x}$};
            
            \node[vertex,label=left:{$j$}] (v0) at (-3,-1.25) {};
    
            \node[vertex,label=above:{$j_1$}] (v01) at (-1,-0.75) {};
            \node[vertex,label=below:{$j_2$}] (v02) at (-1,-1.75) {};

            \node[vertex,label=above:{$k_1$}] (v1) at (1,-0.75) {};
            \node[vertex,label=below:{$k_2$}] (v2) at (1,-1.75) {};
    
            \path 
                (v01) edge[->,thick] (v1)
                (v02) edge[->,thick] (v2)
                
                (v01) edge[->,thick,orange] (v2);
    
            \draw[->,snake=zigzag,segment length=15pt] (v0) -- (v01); 
            \draw[->,snake=zigzag,segment length=15pt] (v0) -- (v02);
    
            \node at (-2.75,-1.95) {\small$P_2' \subseteq L^i_{\ge x+1}$};
    
            \node at (-2.75,-0.65) {\small$P_1' \subseteq E_N$};

        \end{tikzpicture}
        \caption{}
        \label{fig:lem:multiple:entries:2a}
    \end{subfigure}
    \begin{subfigure}{0.325\textwidth}
        \centering
        \begin{tikzpicture}[scale=0.9,transform shape]
        
            \draw[thick,dashed,red,line width = 1.5pt] (0,0) -- (0,-2.5);
            \node at (0,0.5) {\large\textcolor{red}{$\delta^i_x$}};
    
            \node at (-2,0.5) {\large $L^i_{x+1}$};
            \node at (1.5,0.5) {\large $L^i_{x}$};
            
            \node[vertex,label=left:{$j$}] (v0) at (-3,-1.25) {};
    
            \node[vertex,label=above:{$j_1$}] (v01) at (-1,-0.75) {};
            \node[vertex,label=below:{$j_2$}] (v02) at (-1,-1.75) {};

            \node[vertex,label=above:{$k_1$}] (v1) at (1,-0.75) {};
            \node[vertex,label=below:{$k_2$}] (v2) at (1,-1.75) {};
    
            \path 
                (v01) edge[->,thick] (v1)
                (v02) edge[->,thick] (v2)
                (v02) edge[->,thick,orange] (v1);

            \draw[->,snake=zigzag,segment length=15pt] (v0) -- (v01); 
            \draw[->,snake=zigzag,segment length=15pt] (v0) -- (v02);
    
            \node at (-2.75,-1.95) {\small$P_2' \subseteq L^i_{\ge x+1}$};
    
            \node at (-2.75,-0.65) {\small$P_1' \subseteq E_N$};
        \end{tikzpicture}
        \caption{}
        \label{fig:lem:multiple:entries:2b}
    \end{subfigure}
    \begin{subfigure}{0.325\textwidth}
        \centering
        \begin{tikzpicture}[scale=0.9,transform shape]
        
            \draw[thick,dashed,red,line width = 1.5pt] (0,0) -- (0,-2.5);
            \node at (0,0.5) {\large\textcolor{red}{$\delta^i_x$}};
    
            \node at (-2,0.5) {\large $L^i_{x+1}$};
            \node at (1.5,0.5) {\large $L^i_{x}$};
            
            \node[vertex,label=left:{$j$}] (v0) at (-3,-1.25) {};
    
            \node[vertex,label=above:{$j_1$}] (v01) at (-1,-0.75) {};
            \node[vertex,label=below:{$j_2$}] (v02) at (-1,-1.75) {};

            \node[vertex,label=above:{$k_1$}] (v1) at (1,-0.75) {};
            \node[vertex,label=below:{$k_2$}] (v2) at (1,-1.75) {};
    
            \path 
                (v01) edge[->,thick] (v1)
                (v02) edge[->,thick] (v2)
                (v02) edge[->,thick,orange] (v1)
                (v01) edge[->,thick,orange] (v2);

            \draw[->,snake=zigzag,segment length=15pt] (v0) -- (v01); 
            \draw[->,snake=zigzag,segment length=15pt] (v0) -- (v02);
    
            \node at (-2.75,-1.95) {\small$P_2' \subseteq L^i_{\ge x+1}$};
    
            \node at (-2.75,-0.65) {\small$P_1' \subseteq E_N$};

        \end{tikzpicture}
        \caption{}
        \label{fig:lem:multiple:entries:2c}
    \end{subfigure}
    \centering
   
    \caption{The three possible situations considered in~\Cref{lem:multiple:entries:2}, where $P_1' = P_1 \setminus \{(j_1,k_1)\}$ and $P_2' = P_2 \setminus \{(j_2,k_2)\}$.}
    \label{fig:lem:multiple:entries:2}
\end{figure}
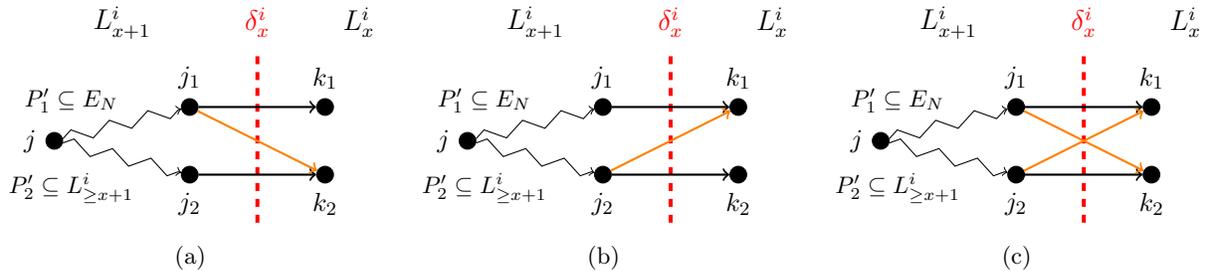

\begin{lemma}
    \label{lem:multiple:entries:2}
    Let $j \in (A(t) \cap L^i_{x+1}) \cup H_{x+1}$ be a job with $x$-entrypaths to two distinct entrypoints $k_1,k_2 \in H^i_x$.
    Let $P_1$ and $P_2$ denote the corresponding paths and let $(j_1,k_1),(j_2,k_2) \in \delta^i_x$ denote the last edge on $P_1$ and $P_2$, respectively. If $j$ can reach $j_1$ via an $E_N$-path and $\{(j_1,k_2),(j_2,k_1)\} \cap E_C \not= \emptyset$, then $\by_j \ge \min\{\by_{k_1},\by_{k_2}\}$. 
\end{lemma}

\begin{proof}
    Assume that $j$ can reach $j_1$ via an $E_N$-path and $\{(j_1,k_2),(j_2,k_1)\} \cap E \not= \emptyset$. Our goal is to show that then $\by_j \ge \min\{\by_{k_1},\by_{k_2}\}$. We distinguish between the two cases (1) $(j_1,k_2) \in E_C$ (see~\Cref{fig:lem:multiple:entries:2a,fig:lem:multiple:entries:2c})  and (2) $(j_2,k_1) \in E_C$ (see~\Cref{fig:lem:multiple:entries:2b,fig:lem:multiple:entries:2c}).

    \paragraph{Case (1)} First, assume $(j_1,k_2) \in E_C$. Since $j_1 \in L_{x+1}$ and $k_2 \in H_x$, we must have $(j_1,k_2) \in \delta^i_x$ by definition of the layers. This allows us to apply~\Cref{lem:multiple:entries:2} to job $j_1 \in L^i_{x+1}$ and the edges $(j_1,k_1),(j_1,k_2) \in \delta^i_x$ in order to obtain $\by_{j_1} \ge \min\{\by_{k_1},\by_{k_2}\}$. We continue by proving $\by_j \ge \by_{j_1}$, which then implies $\by_j \ge \min\{\by_{k_1},\by_{k_2}\}$.

    To see that $\by_j \ge \by_{j_1}$ holds, first recall that $j$ has an $E_N$-path to $j_1$ and either $j \in A(t) \cap L^i_{x+1}$ or $j \in H^i_{x+1}$. By~\Cref{lem:good-path-from-active} (if $j \in A(t) \cap L^i_{x+1}$) and~\Cref{lem:good-path-from-entrypoints} (if $j \in H^i_{x+1}$), this implies that $j$ can reach $j_1$ via a good $E_N$-path. According to the definition of good edges, this gives us $\by_j \ge \by_{j_1}$, and thus, $\by_j \ge \min\{\by_{k_1},\by_{k_2}\}$.
    
    \paragraph{Case (2)} Next, assume $(j_2,k_1)\in E_C$. As in the beginning of Case (1), we can use~\Cref{lem:multiple:entries:2} to obtain $\by_{j_2} \ge \min\{\by_{k_1},\by_{k_2}\}$. However, in contrast to Case (1), our assumptions do not directly give us an $E_N$-path from $j$ to $j_2$. Thus, we cannot finish this case in the same way as Case (1). Instead, we will show that whenever there is no $E_N$-path from $j$ to $j_2$, the schedule and borrow graph have certain structural properties that still allow us to prove the statement.
    
    To do so, we first identify a simple case where an $E_N$-path from $j$ to $j_2$ exists:
    If $j_2$ is executed at a point in time $t'$ with $t \ge t' \ge r_{j_1}$ while being non-clairvoyant, then the existence of an $E_N$-path from $j$ to $j_1$ and the fact that $j \in (A(t) \cap L_{x+1}) \cup H_{x+1}$ imply the existence of an $E_N$-path from $j$ to $j_2$.
    In this case, we can finish the proof in the same way as in Case (1), by just swapping $j_1$ and $j_2$. Thus, for the remainder of this proof, assume that $j_2$ is not executed during $[r_{j_1},t]$ while being non-clairvoyant. In particular, this gives us $r_{j_1} \ge r_{j_2}$.

    To complete the proof, we first establish the following claim, which states a useful implication of our assumption that $j_2$ is not executed during $[r_{j_1},t]$ while being non-clairvoyant.

    \textbf{Claim 1:} We claim that the assumption that $j_2$ is not executed during $[r_{j_1},t]$ while being non-clairvoyant implies that $j_2$ is not executed at all during $[\tau_{1,2},t]$ where $\tau_{1,2}$ is the earliest point in time at which $k_1$ interrupts $j_2$. Note that the point in time $\tau_{1,2}$ must exist since $(j_2,k_1) \in \delta_x$ by definition of Case (2). Furthermore, the interruption by $k_1$ at $\tau_{1,2}$ must be clairvoyant by the layer definition. To see that the claim holds, consider the following arguments:
    \begin{enumerate}
        \item Since $j_2$ is interrupted by the clairvoyant job $k_1$ at $\tau_{1,2}$, it will not be executed again before $k_1$ completes (by~\Cref{obs:clairvoyant-jobs-block-earlier-jobs}). Thus, the earliest point in time at which $j_2$ could be executed again is $C_{k_1}$.
        \item As $(j_2,k_1) \in \delta^i_x$, we get $r_{j_2} \ge s_{k_1}$ by~\Cref{lem:ec:cuts:prop:1}. This implies that $j_2$ must still be non-clairvoyant at $\tau_{1,2}$, i.e., $j_2 \in N(\tau_{1,2})$. Otherwise, the interruption of $j_2$ by $k_1$ at $\tau_{1,2}$ would contradict the tiebreaking rule for clairvoyant jobs of our algorithm. In combination with the previous argument, this means that the next execution of $j_2$ after $\tau_{1,2}$ must be non-clairvoyant and take place after $C_{k_1}$.
        \item As $(j_1,k_1) \in \delta^i_x$, we get $r_{j_1} \le C_{k_1} \le C_{j_1}$ by~\Cref{lem:ec:cuts:prop:1}. In particular, this implies $C_{k_1} \in [r_{j_1},t]$. By the previous argument, we get that the next execution of $j_2$ after $\tau_{1,2}$ (but before $t$) must be non-clairvoyant and take place during $[C_{k_1},t] \subseteq [r_{j_1},t]$. However, this would contradict our assumption that $j_2$ is not executed during $[r_{j_1},t]$ while being non-clairvoyant. Thus, $j_2$ cannot be executed at all during $[\tau_{1,2},t]$.
    \end{enumerate}

    Having the claim in place, we are ready to complete the proof of Case (2) via a case distinction over $\tau_{2,2}$, the earliest point during $I_{j_2}$ at which $k_2$ is being processed. Note that such a point in time must exist as $(j_2,k_2) \in \delta^i_x$. Furthermore, $k_2 \in C(\tau_{2,2})$.

    \paragraph*{Case (2a)} Assume $\tau_{2,2} \ge r_{j_1}$. First observe that if $\tau_{2,2} \in [r_{j_1},C_{j_1})$, then $(j_1,k_2)\in \delta^i_x$ as $k_2$ is executed at $\tau_{2,2}$. In this case, we have $(j_1,k_1),(j_1,k_2) \in \delta^i_x$ and arrive at Case (1) again. Thus, we may assume w.l.o.g.~that $\tau_{2,2} \ge C_{j_1}$. We finish the proof of Case (2a) by observing the following two arguments:
    \begin{enumerate}
        \item We first show that $p_{k_2}(\tau_{2,2}) \ge (1-\alpha) \cdot p_{k_1}$. To see this, note that we have $r_{j_1} \le \tau_{1,1} < C_{j_1} \le \tau_{2,2}$, where $\tau_{1,1}$ is the earliest point in time during $I_{j_1}$ at which $k_1$ is being processed. Furthermore, we have $s_{k_1},s_{k_2} \le r_{j_2}\le r_{j_1}$ by~\Cref{lem:ec:cuts:prop:1} as $(j_2,k_1),(j_2,k_2)\in \delta^i_x$. This means that both, $k_1$ and $k_2$, are alive and clairvoyant at point in time $\tau_{1,1}$. Since $k_1$ is executed over $k_2$ at point in time $\tau_{1,1}$, the definition of our algorithm (tiebreaking rule for clairvoyant jobs) implies $s_{k_2} \le s_{k_1}$ and $p_{k_2}(s_{k_1}) \ge p_{k_1}(s_{k_1}) = (1-\alpha)p_{k_1}$. Furthermore, as $k_1$ has priority over $k_2$, job $k_2$ will only be executed again after $C_{k_1} \ge \tau_{1,1} \ge r_{j_1} \ge r_{j_2}$. This implies that the earliest point in time $t' \ge s_{k_1}$ at which $k_2$ is executed is $\tau_{2,2}$. Thus, we can conclude with $p_{k_2}(\tau_{2,2}) = p_{k_2}(s_{k_1}) \ge (1-\alpha) \cdot p_{k_1}$.
        \item Next, consider the $E_N$-path $P_1'$ from $j$ to $j_1$. Note that such a path exists by assumption of the lemma. We show that there is a job $d \in P_1'$ that is alive at $\tau_{2,2}$.

        If $j \in A(t) \cap L^i_{x+1}$, then we get $I(P_1') \supseteq [r_{j_1},t]$ by~\Cref{lem:ec:cuts:prop:1}. This implies $\tau_{2,2} \in I(P)$, so there must be a job $d \in P_1'$ that is alive at $\tau_{2,2}$.

        If $j \in H^i_{x+1}$, then~\Cref{obs:entrypoint_blocking} implies that $k_2 \in L^i_x$ is not executed after $C_j$, which implies $C_j \ge \tau_{2,2} \ge C_{j_1}$. Since $I(P_1') \supseteq [r_{j_1},C_j]$, we get $\tau_{2,2} \in I(P)$, so there must be a job $d \in P$ that is alive at $\tau_{2,2}$.
    \end{enumerate}

    The second argument shows that there is a job $d$ on the $E_N$-path $P_1'$ that is alive at $\tau_{2,2}$. The first argument shows $p_{k_2}(\tau_{2,2}) \ge (1-\alpha) \cdot p_{k_1}$. Thus,~\Cref{obs:ec:cut:aux:1} implies $\by_d \ge \by_{k_2}$.
    Since there is an $E_N$-path $P' \subseteq P_1'$ from $j$ to $d$, we get $\by_j \ge \by_d \ge \by_{k_2}$ by either~\Cref{lem:good-path-from-active} (if $j \in A(t) \cap L^i_{x+1})$ or by~\Cref{lem:good-path-from-entrypoints} (if $j \in H^i_{x+1}$). 

    \paragraph*{Case (2b)} Assume $\tau_{2,2} < r_{j_1}$. Similar to the first argument in Case (2a), we first show that $p_{k_1}(\tau_{1,2}) \ge (1-\alpha)p_{k_2}$.
    
    To see this, first observe that $\tau_{2,2} < r_{j_1}$ implies $\tau_{2,2} < \tau_{1,1}$ as $r_{j_1} \le \tau_{1,1} \le C_{j_1}$ holds by definition of $\tau_{1,1}$ (the earliest point in time during $I_{j_1}$ at which $k_1$ is being processed).
    Since $s_{k_1}, s_{k_2} \le r_{j_2}$ holds by~\Cref{lem:ec:cuts:prop:1}, we have that $k_1$ and $k_2$ are both alive and clairvoyant at point in time $\tau_{2,2} \ge r_{j_2}$. As $k_2$ is executed over $k_1$ at point in time $\tau_{2,2}$, the definition of our algorithm (tiebreaking rule for clairvoyant jobs) implies $s_{k_1} \le s_{k_2}$ and $p_{k_1}(s_{k_2}) \ge p_{k_2}(s_{k_2}) = (1-\alpha)p_{k_2}$. Furthermore, as $k_2$ has priority over $k_1$, we have that $k_1$ is not executed between $s_{k_2}$ and $\tau_{1,2} \ge r_{j_2}$. Thus, we get $p_{k_1}(\tau_{1,2}) = p_{k_1}(s_{k_2}) \ge p_{k_2}(s_{k_2}) = (1-\alpha)p_{k_2}$.

    Having shown the inequality $p_{k_1}(\tau_{1,2}) \ge (1-\alpha)p_{k_2}$, we complete the proof by distinguishing between different possible values of $\tau_{1,2}$:

    \begin{enumerate}
        \item First, assume $\tau_{1,2} \ge r_{j_1}$. If $C_j < t$ (which by the assumption of the lemma can only be the case if $j \in H^i_{x+1}$), then~\Cref{obs:entrypoint_blocking} implies that $k_1$ is not executed during $[C_j,t]$. Thus, we must have $r_{j_1} \le \tau_{1,2} \le \min\{C_{j},t\}$. 
        
        Consider an $E_N$-path $P_1'$ from $j$ to $j_1$. Such a path must exist by the assumptions of the lemma. By~\Cref{obs:lifetime:1}, we have $I(P_1') \supseteq [r_{j_1}, \min\{C_{j},t\}]$, and thus, $\tau_{1,2} \in I(P_1')$.

        The fact that $\tau_{1,2} \in I(P_1')$ implies that there is a job $d \in P_1'$ that is alive at point in time $\tau_{1,2}$. Since we already argued that $p_{k_1}(\tau_{1,2}) \ge (1-\alpha)p_{k_2}$ and $k_1$ is executed at $\tau_{1,2}$, we can apply~\Cref{obs:ec:cut:aux:1} to get $\by_d \ge \by_{k_2}$. By choice of $d \in P_1'$, there is a $E_N$-path $P_1'' \subseteq P_1'$ from $j$ to $d$. Thus, we get $\by_j \ge \by_d \ge \by_{k_2}$ by either~\Cref{lem:good-path-from-active} (if $j \in A(t) \cap L^i_{x+1}$) or by~\Cref{lem:good-path-from-entrypoints} (if $j \in H^i_{x+1}$). 

        \item Next, assume $\tau_{1,2} < r_{j_1}$. Consider the $x$-entrypath $P_2$ from $j$ to $k_2$ that exists by assumption of the lemma and let $P_2'$ denote the $j$-$j_2$-subpath of $P_2$. Since $P_2$ is an $x$-entrypath, we have $P_2' \cap \bar{L}_x = \emptyset$.
        
        If $r_j \le \tau_{1,2}$, then we have $\tau_{1,2} \in I_j$, either by definition in case that $j \in A(t) \cap L^i_{x+1}$ or by~\Cref{obs:entrypoint_blocking} in case that $j \in H^i_{x+1}$. Having $\tau_{1,2} \in I_j$, we can use the fact that $p_{k_1}(\tau_{1,2}) \ge (1-\alpha)p_{k_2}$ and apply~\Cref{obs:ec:cut:aux:1} to $j$ in order to conclude with $\by_j \ge \by_{k_2}$. Thus, assume $r_j > \tau_{1,2}$ for the remainder of this proof.

        The assumption that $r_j > \tau_{1,2}$ and the fact that $j_2$ is not executed again after $\tau_{1,2}$ (Claim 1) implies that path $P_2'$ must contain a job $d$ that is alive at point in time $\tau_{1,2}$ and executed at some point in time $t' > \tau_{1,2}$. If such a job $d$ would not exist, then $P_2'$ would not be a $j$-$j_2$-path.

        Using again that the inequality $p_{k_1}(\tau_{1,2}) \ge (1-\alpha)p_{k_2}$ holds, we can apply~\Cref{obs:ec:cut:aux:1} to job $d$ to obtain $\by_d \ge \by_{k_2}$. We finish the proof by showing that there is an $E_N$-path $P$ from $j$ to $d$, which then implies $\by_j \ge \by_d \ge \by_{k_2}$ by either~\Cref{lem:good-path-from-active} (if $j \in A(t) \cap L^i_{x+1}$) or by~\Cref{lem:good-path-from-entrypoints} (if $j \in H^i_{x+1}$). 

        In order to prove that the path $P$ exists, we first characterize some properties of job $d$. Since $d \in P_2' \subseteq L^i_{\geq x+1}$, $k_1 \in L^i_x$ and $(d,k_1) \in E_C$ (as $k_1$ is executed $\tau_{1,2} \in I_d$), we get that $d \in L^i_{x+1}$ and $(d,k_1) \in \delta^i_x$ by the definition of layers. Using~\Cref{lem:ec:cuts:prop:1}, the latter gives us $s_{k_1} \le r_d$. This in turn implies $d \in N(\tau_{1,2})$, as otherwise the algorithm would execute $d$ over $k_1$ at point in time $\tau_{1,2}$ (cf.~tiebreaking rule for clairvoyant jobs). The fact that $k_1$ is executed over $d$ at $\tau_{1,2}$ means that $d$ is only executed again after $k_1$ completes (cf.~\Cref{obs:clairvoyant-jobs-block-earlier-jobs}). Thus, the earliest execution of $d$ after $\tau_{1,2}$ is non-clairvoyant and takes place after $C_{k_1} \ge r_{j_1}$.

        If $j \in A(t) \cap L^i_{x+1}$, then the earliest execution of $d$ after $\tau_{1,2}$ takes place during $I(P_1') \supseteq [r_{j_1},t]$ for the $j$-$j_1$-path $P_1'$ that only uses edges in $E_N$ (exists by assumption of the lemma). This directly implies the existence of an edge $(d',d)$ for some $d' \in P_1'$, which in turn implies the existence of a $j$-$d$-path $P$ that only uses edges in $E_N$.

        If $j \in H^i_{x+1}$, then~\Cref{obs:entrypoint_blocking} implies that the earliest execution of $d$ after $\tau_{1,2}$ (and thus after $C_{k_1} \ge r_{j_1}$) must take place during $[r_{j_1}, C_j]$. Thus, the earliest execution of $d$ after $\tau_{1,2}$ takes place during $I(P_1') \supseteq [r_{j_1},C_j]$ for the $j$-$j_1$-path $P_1'$ that only uses edges in $E_N$ (exists by assumption of the lemma). This directly implies the existence of an edge $(d',d)$ for some $d' \in P_1'$, which in turn implies the existence of a $j$-$d$-path $P$ that only uses edges in $E_N$.
    \end{enumerate}
    This finally concludes the proof of the lemma.
\end{proof}

The next auxiliary lemma considers the same situation as before, but this time assumes that $\{(j_1,k_2),(j_2,k_1)\} \cap E_C = \emptyset$ (cf.~\Cref{fig:lem:multiple:entries:3}). In the previous lemma, we used the existence of at least one of the edges $(j_1,k_2)$ and $(j_2,k_1)$ to show that $\by_j \ge \min\{\by_{k_1},\by_{k_2}\}$. In the following lemma, we will observe that the absence of the edges $(j_1,k_2)$ and $(j_2,k_1)$ enforces a certain structure in the borrow graph and schedule that can be exploited to show the same statement.

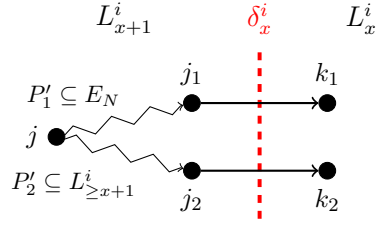
\begin{figure}[t]
        \centering
        \begin{tikzpicture}[scale=0.9,transform shape]
        
            \draw[thick,dashed,red,line width = 1.5pt] (0,0) -- (0,-2.5);
            \node at (0,0.5) {\large\textcolor{red}{$\delta^i_x$}};
    
            \node at (-2,0.5) {\large $L^i_{x+1}$};
            \node at (1.5,0.5) {\large $L^i_{x}$};
            
            \node[vertex,label=left:{$j$}] (v0) at (-3,-1.25) {};
    
            \node[vertex,label=above:{$j_1$}] (v01) at (-1,-0.75) {};
            \node[vertex,label=below:{$j_2$}] (v02) at (-1,-1.75) {};

            \node[vertex,label=above:{$k_1$}] (v1) at (1,-0.75) {};
            \node[vertex,label=below:{$k_2$}] (v2) at (1,-1.75) {};
    
            \path 
                (v01) edge[->,thick] (v1)
                (v02) edge[->,thick] (v2);                
    
            \draw[->,snake=zigzag,segment length=15pt] (v0) -- (v01); 
            \draw[->,snake=zigzag,segment length=15pt] (v0) -- (v02);
    
            \node at (-2.75,-1.95) {\small$P_2' \subseteq L^i_{\ge x+1}$};
    
            \node at (-2.75,-0.65) {\small$P_1' \subseteq E_N$};

        \end{tikzpicture}
    \caption{The three possible situations considered in~\Cref{lem:multiple:entries:3}, where $P_1' = P_1 \setminus \{(j_1,k_1)\}$ and $P_2' = P_2 \setminus \{(j_2,k_2)\}$.}
    \label{fig:lem:multiple:entries:3}
\end{figure}

\begin{lemma}
    \label{lem:multiple:entries:3}
    Let $j \in (A(t) \cap L^i_{x+1}) \cup H^i_{x+1}$ be a job with $x$-entrypaths to two distinct entrypoints $k_1,k_2 \in H^i_x$.
    Let $P_1$ and $P_2$ denote the corresponding paths and let $(j_1,k_1),(j_2,k_2) \in \delta^i_x$ denote the last edge on $P_1$ and $P_2$, respectively. If $j$ can reach $j_1$ via an $E_N$-path and $\{(j_1,k_2),(j_2,k_1)\} \cap E_C = \emptyset$, then $\by_j \ge \min\{\by_{k_1},\by_{k_2}\}$. 
\end{lemma}

\begin{proof}
    In the first part of the proof, we will establish structural properties of the schedule and borrow graph that corresponds to the situation of the lemma. Afterwards, we will use these properties to prove the lemma.

    To start the proof of the structural properties, we establish that we can assume w.l.o.g.~that $r_{j_1} \ge r_{j_2}$. Note that this does not trivially hold by renaming the jobs because $j_1$ and $j_2$ might have different properties: there is an $E_N$-path from $j$ to $j_1$, but there might not be an $E_N$-path from $j$ to $j_2$.

    \paragraph{Proof that the assumption $r_{j_1} \ge r_{j_2}$ is w.l.o.g.} First, observe that if $r_{j_1} < r_{j_2}$, then $j$ can also reach $j_2$ via an $E_N$-path. To see this, consider the $E_N$-path $P_1'$ from $j$ to $j_1$. If $r_{j_2} \in I(P_1')$, then there exists an edge $(d,j_2) \in E_N$ for some $d \in P_1'$, which implies the existence of an $E_N$-path from $j$ to $d$. 
    In case that $j \in A(t) \cap L^i_{x+1}$, we have $I(P_1') \supseteq [r_{j_1},t]$ and, thus, $r_{j_1} < r_{j_2}$ implies $r_{j_2} \in I(P_1')$. In case that $j \in H^i_{x+1}$, we could in principle have $r_{j_2} \not\in I(P_1')$ if $r_{j_2} \ge C_{j}$. However, this would contradict~\Cref{obs:entrypoint_blocking} and, thus, $r_{j_1} < r_{j_2}$ implies $r_{j_2} \in I(P_1')$.

    Having established this observation, we can assume without loss of generality that $r_{j_1} \ge r_{j_2}$: If $r_{j_1} \ge r_{j_2}$ does not hold, then $r_{j_1} < r_{j_2}$ and, therefore, $j$ can reach $j_2$ via an $E_N$-path, which allows us to just swap the roles of $j_1$ and $j_2$ without changing the assumptions of the lemma. Thus, for the remainder of this proof, we assume $r_{j_1} \ge r_{j_2}$.  

    \medskip

    Next, consider the edges $(j_1,k_1),(j_2,k_2) \in \delta^i_x$ and let $t_\ell$ with $\ell \in \{1,2\}$ denote the earliest point in time at which $k_\ell$ is processed during $I_{j_\ell}$. Note that the points in time $t_{\ell}$ must exist and that $k_\ell$ must be clairvoyant at $t_\ell$ by definition of $\delta^i_x$. 
    We continue by showing several properties of the schedule.
 
    \begin{enumerate}
        \item[\textbf{Property 1:} $t_1 > t_2$.] To see why this property holds, assume otherwise, i.e., $t_2 > t_1$; note that $t_1 = t_2$ is not possible since no two clairvoyant jobs are ever executed in parallel by definition of our algorithm. By definition of $t_1$, we must have $t_1 \ge r_{j_1}$ and, therefore, $t_2 > t_1 \ge r_{j_1} \ge r_{j_2}$. Since $(j_2,k_2) \in \delta^i_x$, we also get $C_{j_2} \ge C_{k_2} \ge t_2 > t_1$ by~\Cref{lem:ec:cuts:prop:1}. Combining these inequalities gives us $C_{j_2} > t_1 \ge r_{j_2}$. However, this implies $(j_2,k_1) \in E_C$, which contradicts the assumption of the lemma that $\{(j_1,k_2),(j_2,k_1)\} \cap E_C = \emptyset$.
        \item[\textbf{Property 2:} $r_{j_1} \ge C_{k_2}$.]   Recall that we already assume $r_{j_1} \ge r_{j_2}$, so it remains to prove that $j_1$ cannot be released during $[r_{j_2}, C_{j_2})$. We complete the proof by separately arguing over subintervals of $[r_{j_2}, C_{j_2})$:
        \begin{enumerate}
            \item We first argue that $r_{j_1} \not\in [r_{j_2},t_2)$. To see this, assume otherwise, i.e., $r_{j_1} < t_2$ . Then we must have $C_{j_1} \le t_2$; otherwise $r_{j_1} \le t_2 < C_{j_1}$ would imply $(j_1,k_2) \in E_C$ which contradicts the assumption of the lemma.
            However, $C_{j_1} \le t_2$ implies $t_1 \le C_{j_1} \le t_2$; a contradiction to the claim $t_1 > t_2$.
    
            \item Next, we show that $r_{j_1} \not\in [t_2,C_{k_2})$. For the sake of contradiction assume $t_2 \le r_{j_1} < C_{k_2}$. Then, we must also have $C_{j_1} < C_{k_2}$; otherwise the edge $(j_1,k_2) \in E_C$ would exist and contradict the assumption of the lemma. However, this implies $r_{j_2} \le r_{j_1} \le t_1 \le C_{j_1} \le C_{k_2} \le C_{j_2}$, where the last inequality holds by~\Cref{lem:ec:cuts:prop:1}. The inequality $r_{j_2} \le t_1 \le C_{j_2}$ implies $(j_2,k_1) \in E_C$, which contradicts the assumption of the lemma.
        \end{enumerate}
        The facts that $r_{j_1} \ge r_{j_2}$,  $r_{j_1} \not\in [r_{j_2},t_2)$  and $r_{j_1} \not\in [t_2,C_{k_2})$ imply $r_{j_1} \ge C_{k_2}$. 
        \item[\textbf{Property 3:} $s_{k_1} \le s_{k_2}$.] 
        We start by showing the slightly weaker inequality $s_{k_1} \le r_{j_2}$. To this end, consider the $x$-entry path $P_2$ from $j$ to $k_2$ and, in particular, the $j$-$j_2$-subpath $P_2'$ of $P_2$. We distinguish between $j \in A(t) \cap L^i_{x+1}$ and $j \in H^i_{x+1}$:
        \begin{enumerate}
            \item If $j \in A(t) \cap L^i_{x+1}$, then we have $I(P_2') \supseteq [r_{j_2},t]$ by~\Cref{obs:lifetime:2}. Thus, if $s_{k_1} \in (r_{j_2},t]$, then $s_{k_1}^- \in [r_{j_2},t] \subseteq I(P_2')$. The fact that $s_{k_1}^- \in I(P_2')$ implies that there is a job $d \in P_2'$ with $s_{k_1}^- \in I_d$. Since $P_2$ is an $x$-entry path, we have $P_2' \subseteq L^i_{\ge x+1}$. However, $s_{k_1}^- \in I_d$ and $d \in L^i_{\ge x+1}$ imply that there is an edge $(d,k_2) \in E_N$ with $d \in L^i_{\ge x+1}$ and $k_2 \in L^i_{x}$; a contradiction to the layer definition.
            \item If $j \in H^i_{x+1}$, then  $s_{k_1}^- \in [r_{j_2},t]$ does not necessarily imply $s_{k_1}^- \in I(P_2')$ since we cannot apply~\Cref{obs:lifetime:2} if $j \not\in A(t) \cap L^i_{x+1}$. However, using $j \in H^i_{x+1}$, we can apply~\Cref{obs:entrypoint_blocking} to conclude that $k_2 \in L^i_x$ is not executed during $[C_j,t]$. Thus,  $s_{k_1}^- \in [r_{j_2},t]$ implies $s_{k_1}^-  \in [r_{j_2},C_j] \subseteq I(P_2')$, where $[r_{j_2},C_j] \subseteq I(P_2')$ holds as $j_2,j \in P_2'$. Having established $s_{k_1}^- \in  I(P_2')$, we can finish the argument as in the previous case $j \in A(t) \cap L^i_{x+1}$.
        \end{enumerate}
        It remains to show that $s_{k_1} \le s_{k_2}$. To see this, assume otherwise, i.e., $s_{k_1} > s_{k_2}$. Using $(j_1,k_1),(j_2,k_2) \in \delta^i_x$,~\Cref{lem:ec:cuts:prop:1} and our assumption that $r_{j_1} \ge r_{j_2}$, we obtain $C_{k_1} \ge r_{j_1} \ge r_{j_2}$ and $C_{k_2} \ge r_{j_2}$. Furthermore, we have $s_{k_2} \le r_{j_2}$ by~\Cref{lem:ec:cuts:prop:1} and $s_{k_1} \le r_{j_2}$ as established above. 
        Together, the inequalities $s_{k_2} < s_{k_1} \le  r_{j_2}$ and $C_{k_2},C_{k_1} \ge r_{j_2}$ imply that $k_2$ is alive and clairvoyant at point in time $s_{k_1}$. By the definition of our algorithm (tiebreaking rule for clairvoyant jobs), this implies $C_{k_2} \le C_{k_1}$. 
        Recall that we already established $r_{j_1} \ge C_{k_2}$ with the previous claim. This gives us $r_{j_1} \ge C_{k_1}$. However, $r_{j_1} \ge C_{k_1}$ is a contradiction to $(j_1,k_1) \in \delta^i_x$ and, in particular, to~\Cref{lem:ec:cuts:prop:1}. Thus, we can conclude with the claim $s_{k_1} \le s_{k_2}$.
        \item[\textbf{Property 4:} $r_{j_1} \ge C_{j_2}$.] We show the claim via proof by contradiction. To this end, assume $r_{j_1} < C_{j_2}$. In combination with our assumption that $r_{j_2} \le r_{j_1}$ this gives us $I_{j_1} \subseteq I_{j_2}$. In particular, this implies $t_1 \in I_{j_2}$. However, this contradicts the assumption of the lemma that $(j_2,k_1) \not\in E_C$. 
        \item[\textbf{Property 5:}  $C_{k_1} \ge C_{k_2}$.] The existence of $(j_2,k_2) \in \delta^i_x$ and~\Cref{lem:ec:cuts:prop:1} imply $C_{k_2} \le C_{j_2}$. The previous property $r_{j_1} \ge C_{j_2}$, the edge $(j_1,k_1) \in \delta^i_x$ and~\Cref{lem:ec:cuts:prop:1} together imply $C_{k_1} \ge r_{j_1} \ge C_{j_2}$. Putting both inequalities together yields $C_{k_1} \ge C_{k_2}$. 
    \end{enumerate}

    For the next property, let $t_1'$ denote the earliest point in time with $t_1' \ge C_{k_2}$ at which $k_1$ is executed. Note that such a point in time must exist as $(j_1,k_1) \in E_C$ and $r_{j_1} \ge C_{k_2}$. 

    \begin{enumerate}
        \item[\textbf{Property 6:}  $p_{k_1}(t_1') \ge (1- \alpha) \cdot p_{k_2}$.] To see the claim, first recall that $k_1$ is alive and clairvoyant at point in time $s_{k_2}$ since we showed that $s_{k_1} \le s_{k_2} \le r_{j_1} \le C_{k_1}$.
        As we also established $C_{k_1} \ge C_{k_2}$, this implies $p_{k_1}(s_{k_2}) \ge p_{k_2}(s_{k_2}) = (1-\alpha) \cdot p_{k_2}$ by the definition of the algorithm (tiebreaking rule for clairvoyant jobs). Furthermore, since $k_2$ is clairvoyant and executed at point in time $t_2$ while $k_1 \in C(t_2)$ %
        implies that $k_1$ is not executed at all during $[s_{k_2},C_{k_2})$. By definition of $t_1'$, this means that $k_1$ is not executed during at all during $[s_{k_2},t_1')$. Thus, we can conclude with $p_{k_1}(t_1') = p_{k_1}(s_{k_2}) \ge (1- \alpha) \cdot p_{k_2}$. 
    \end{enumerate}

    These six properties and the assumption that $r_{j_1} \ge r_{j_2}$ finally allow us to prove the lemma. %

    \paragraph*{Proof of the lemma.} %
     The schedule of the algorithm has the following properties:
    \begin{enumerate}[(a)]
        \item $s_{k_1} \le s_{k_2} \le r_{j_2} \le t_2 \le C_{k_2} \le C_{j_2} \le r_{j_1} \le t_1$; see~\Cref{fig:8er:lemma:2} for an illustration. The chain of these inequalities is a consequence of the first five properties and of the assumption that $r_{j_1} \ge r_{j_2}$.
        \item $p_{k_1}(t_1') \ge (1-\alpha) \cdot p_{k_2}$ for the earliest point in time $t_1' \ge C_{k_2}$ at which $k_1$ is executed (Property~6).
    \end{enumerate}

    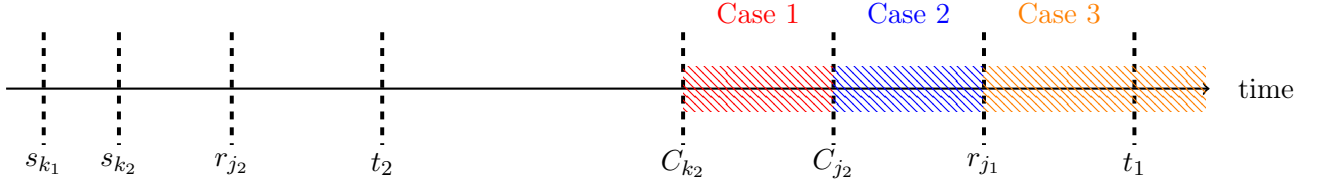
\begin{figure}[t]
        \centering
        \begin{tikzpicture}

            \draw[->,thick] (-4,0) -- (12,0);
            \node at (12.75,0) {time};

            \draw[thick,dashed,line width = 1.5pt] (-3.5,0.75) -- (-3.5,-0.75);
            \node at (-3.5,-1) {$s_{k_1}$};

            \draw[thick,dashed,line width = 1.5pt] (-2.5,0.75) -- (-2.5,-0.75);
            \node at (-2.5,-1) {$s_{k_2}$};

            \draw[thick,dashed,line width = 1.5pt] (-1,0.75) -- (-1,-0.75);
            \node at (-1,-1) {$r_{j_2}$};

            \draw[thick,dashed,line width = 1.5pt] (1,0.75) -- (1,-0.75);
            \node at (1,-1) {$t_2$};

            \draw[thick,dashed,line width = 1.5pt] (5,0.75) -- (5,-0.75);
            \node at (5,-1) {$C_{k_2}$};

            \draw[thick,dashed,line width = 1.5pt] (7,0.75) -- (7,-0.75);
            \node at (7,-1) {$C_{j_2}$};

            \draw[thick,dashed,line width = 1.5pt] (9,0.75) -- (9,-0.75);
            \node at (9,-1) {$r_{j_1}$};

            \draw[thick,dashed,line width = 1.5pt] (11,0.75) -- (11,-0.75);
            \node at (11,-1) {$t_{1}$};

            \fill[pattern=north west lines, pattern color=red] (5,-0.3) rectangle (7,0.3);
            \node at (6,1) {\textcolor{red}{Case 1}};

            \fill[pattern=north west lines, pattern color=blue] (7,-0.3) rectangle (9,0.3);
            \node at (8,1) {\textcolor{blue}{Case 2}};

            \fill[pattern=north west lines, pattern color=orange] (9,-0.3) rectangle (11.95,0.3);
            \node at (10,1) {\textcolor{orange}{Case 3}};

        \end{tikzpicture}
        \caption{Illustrates the order of relevant points in time in the proof of~\Cref{lem:multiple:entries:3} and the three cases for the value of $t_1'$ (colored areas).}
        \label{fig:8er:lemma:2}
    \end{figure}

    To complete the proof, we consider $t_1'$ and distinguish between the three cases (1) $t_1' \in [C_{k_2},C_{j_2})$ (red area in~\Cref{fig:8er:lemma:2}), (2) $t_1' \in  [C_{j_2},r_{j_1})$ (blue area in~\Cref{fig:8er:lemma:2}) and (3) $t_1' \in  [r_{j_1},t]$ (orange area in~\Cref{fig:8er:lemma:2}). Note that $t_1' < C_{k_2}$ is not possible by definition of $t_1'$ and that $t_1' \le t_1$ follows from $t_1 \ge C_{k_2}$.

    \paragraph{Case 1:} Assume that $t_1' \in [C_{k_2},C_{j_2})$. Since we have $r_{j_2} \le C_{k_2}$, $t_1' \in [C_{k_2},C_{j_2})$ would imply $t_1' \in I_{j_1}$ and, thus, the existence of edge $(j_2,k_1) \in E_C$. However, the edge $(j_2,k_1)$ cannot exist due to the assumption of the lemma. Therefore, Case 1 cannot actually occur.

    \paragraph*{Case 2:} Assume that $t_1' \ge C_{j_2}$. Consider the $j$-$j_2$-subpath $P'_2$ of $P_2$. We first argue that we must have $t_1' \le \min\{C_j,t\}$, which then (in combination with $t_1' \ge C_{k_2} \ge C_{j_2} \ge r_{j_2}$) implies $t_1' \in [r_{j_2},C_j] \subseteq I(P'_2)$. If $j \in A(t) \cap L^i_{x+1}$, then $\min\{C_j,t\} = t$ and $t_1' \le \min\{C_j,t\}$ holds because of $t_1' \le t_1 \le t$. If $j \not\in A(t) \cap L^i_{x+1}$, then $j \in H^i_{x+1}$ and $t_1' \le \min\{C_j,t\}$ follows from \Cref{obs:entrypoint_blocking}.

    Having shown that $t_1' \in I(P'_2)$, there must be a job $d \in P'_2$ with $t_1' \in I_d$. If $d = j$, then $p_{k_1}(t_1') \ge (1-\alpha) \cdot p_{k_2}$  and~\Cref{obs:ec:cut:aux:1} imply $\by_{j} \ge \by_{k_2}$. Thus, assume for the remainder of the proof that $d \not= j$ and, in particular, that $j$ is not alive at $t_1'$.

    Since~\Cref{obs:entrypoint_blocking} implies $C_{j}> t_1'$, $j$ not being alive at $t_1'$ means that we must have $r_j > t_1'$. Thus, we have $C_{j_2} < t_1'$ \emph{and} $r_j > t_1'$. As $P_2'$ is a $j$-$j_2$-path, this means that there must not only be a job $d \in P_2'$ with $t_1' \in I_d$, but job $d$ must also be executed again at some point in time $t' \ge t_1'$; if such a job $d$ would not exist, then $P_2'$ would not be a $j$-$j_2$-path. This job $d$ additionally has the following properties:
    \begin{enumerate}
        \item $d \in L^i_{x+1}$ and $(d,k_1) \in \delta^i_x$. Since $P_2 \supseteq P_2'$ is an $x$-entry path, $d \in P_2'$ implies $d \in L^i_{\ge x+1}$. Furthermore, $t_1' \in I_d$ implies that the edge $(d,k_1) \in E_C$ exists. Thus, $k_1 \in L^i_x$ implies $d \in L^i_{\le x+1}$. In conclusion, we get $d \in L^i_{x+1}$ and $(d,k_1) \in \delta^i_x$.
        \item $d \in N(t_1')$. Since $(d,k_1) \in \delta^i_x$, we can apply~\Cref{lem:ec:cuts:prop:1} to get $s_{k_1} \le r_d$. If $d$ was not in $N(t_1')$, then it would need to emit after $s_{k_1}$ and before $t_1' \le C_{k_1}$. However, then the definition of our algorithm (tiebreaking rule for clairvoyant jobs) would imply that $k_1$ is only executed again (after $s_d$) once $d$ is completed. This is a contradiction to $t_1' \in I_d$, and we can conclude with  $d \in N(t_1')$.
        \item The earliest point in time $t_d' \ge t_1'$ at which $d$ is executed satisfies $t_d' \ge C_{k_1}$. This is because the clairvoyant job $k_1$ is executed at time $t_1'$ while $d$ is also alive and clairvoyant. By~\Cref{obs:clairvoyant-jobs-block-earlier-jobs}, this directly implies $t_d' \ge C_{k_1}$.
        \item $\by_d \ge \by_{k_1}$. This is directly implied by $t_1' \in I_d$ and~\Cref{obs:ec:cut:aux:1}.
    \end{enumerate}

    The second and third property of $d$ imply that $d$ is executed after $C_{k_1}$ (at point in time $t_d'$) while being non-clairvoyant. Using~\Cref{obs:entrypoint_blocking} and~\Cref{lem:ec:cuts:prop:1}, we get that $r_{j_1} \le C_{k_1} \le t_d' \le \min\{C_j,t\}$ must hold. In particular, this implies $t_d' \in I(P_1')$, where $P_1'$ is the $E_N$-path from $j$ to $j_1$ that exists by assumption of the lemma. The fact that $t_d' \in I(P_1')$ in turn implies that there is an $E_N$-path from $j$ to $d$, and thus, $\by_j \ge \by_d$. Using the fourth property above, we can conclude with $\by_j \ge \by_d \ge \by_{k_1}$.

    \paragraph{Case 3:} Assume that $t_1' \in [r_{j_1},t)$. Since $t_1$ is the earliest point in time at or after $r_{j_1}$ at which $k_1$ is executed, we get $t_1 = t_1'$. By Property 6, this implies $p_{k_1}(t_1) \ge (1- \alpha) \cdot p_{k_2}$. Since $t_1 \in I_{j_1}$, \Cref{obs:ec:cut:aux:1} implies $\by_{j_1} \ge \by_{k_2}$. We can conclude with $\by_j \ge \by_{k_2}$ by observing that the $E_N$-path from $j$ to $j_1$ implies $\by_j \ge \by_{j_1} \ge \by_{k_2}$ (using~\Cref{lem:good-path-from-entrypoints,lem:good-path-from-active}).
\end{proof}

Exploiting all these auxiliary lemmas, we are finally ready to prove~\Cref{lem:multiple:entrypoints}.

\begin{proof}[Proof of~\Cref{lem:multiple:entrypoints}]
     Let $j \in (A(t) \cap L^i_{x+1}) \cup H^i_{x+1}$ be a job that can reach at least two distinct entrypoints of $L^i_x$ via $x$-entry paths.  

     By definition of the layers, there is at least one $k_1 \in H^i_x$ that $j$ can reach via an $x$-entry path $P_1$ that, apart from the final edge, only uses edges in $E_N$.

     By the assumption that $j$ can reach two distinct entrypoints of $L^i_x$ via $x$-entry paths, there must be a second entrypoint $k_2 \in H^i_x$ with $k_1 \not= k_2$ that $j$ can reach via an $x$-entry path $P_2$.

     Let $(j_1,k_1),(j_2,k_2) \in \delta^i_x$ denote the last edges on $P_1$ and $P_2$, respectively. If the edges $(j_1,k_2)$ and $(j_2,k_1)$ do both not exist, then~\Cref{lem:multiple:entries:2} implies $\by_j \ge \min\{\by_{k_1},\by_{k_2}\}$.
     Otherwise, \Cref{lem:multiple:entries:3} implies $\by_j \ge \min\{\by_{k_1},\by_{k_2}\}$. In both cases, we can conclude the proof.
\end{proof}

\subsection{Borrowing via Bottlenecks}

\Cref{thm:yinequality:no:bottlenecks} allows us to compare the truncated progress of $i$ with the truncated progress of jobs~$j$ that can reach $i$ without visiting a bottleneck of $i$. In this section, we want to consider jobs $j$ that need to visit at least one bottleneck of $i$ in order to reach $j$. In contrast to \Cref{thm:yinequality:no:bottlenecks}, we will not be able to show $\by_j \ge \by_i$ for such jobs $j$ as the inequality simply might not hold, as we argued in \Cref{sec:matching:arguments} for the example given in \Cref{fig:example-delta-edge}.

Instead, we want to exploit that each path from $j$ to $i$ must visit the bottleneck $b_{ji}$.
Let $\Delta= \{d_1,\ldots,d_r\} = \{d \in J \mid (d,b_{ji}) \in \delta^i_{\sigma_{ji}}\}$. By definition of the layers, each path from $j$ to $i$ starts with a $\sigma_{ji}$-entry path that ends with an edge $(d,b_{ji}) \in \delta^i_{\sigma_{ji}}$. Thus, the maximum amount of processing that $j$ can borrow from $i$, that is, the maximum flow from $j$ to $i$ in $G_F$, is limited by the maximum amount of flow that $j$ can send to $b_{ji}$ via edges $(d,b_{ji}) \in \delta^i_{\sigma_{ji}}$ for jobs $d$ that $j$ can actually reach via an $L^i_{\geq \sigma_{ji}+1}$-path (recall that such a path only uses vertices of the set $L^i_{\geq \sigma_{ji}+1}$). We are interested in analyzing this amount of processing.

For each $d_q \in \Delta$, let $\tau_q$ denote the earliest point in time in $I_{d_q}$ at which $b_{ji}$ is executed. Since $(d_q,b_{ji}) \in \delta^i_{\sigma_{ji}}$, such a point in time must exist. By definition of the flow network $G_F$, the maximum amount of flow that can be sent from $d_q$ to $b_{ji}$ is at most $p_{b_{ji}}(\tau_{d_q})$ (we formalize this in the proof of \Cref{lem:segment:bottleneck:capacity}). Instead of lower bounding $\by_j$ in terms of $\by_i$ as in the previous section, our goal is now to lower bound $\by_j$ in terms of $p_{b_{ji}}(\tau_{d_q})$. In particular, we want to do so for every edge  $(d_q,b_{ji}) \in \delta^i_{\sigma_{ji}}$ such that $j$ can reach $d_q$ via an $L^i_{\geq \sigma_{ji}+1}$-path. The main result of this section is the following lemma, which can be viewed as a first step towards proving~\Cref{lem:segment:bottleneck:capacity}.

\stealingViaBottlenecks*

To prove the lemma, we first look at the edges $(d,b_{ji}) \in \delta^i_{\sigma_{ji}}$ and observe the following, which essentially proves the statement of \Cref{lem:stealing:via:bottlenecks:2} for the special case $d_q = j$, and will be useful for proving the full lemma.

\begin{observation}\label{lem:aux:ec:edges:2}
  If $(d_q,b_{ji}) \in \delta^i_{\sigma_{ji}}$, then $\frac{1-\alpha}{\alpha} \by_{d_q} \ge \frac{1-\alpha}{\alpha} y_{d_q}(\tau_q) \ge p_{b_{ji}}(\tau_q)$. 
\end{observation}
 
\begin{proof}
 Recall that $\tau_q$ is the earliest point in time during $I_{d_q}$ at which $b_{ji}$ is processed. Since  $(d_q,b_{ji}) \in \delta^i_{\sigma_{ji}}$, such a point in time exists. Furthermore, by~\Cref{lem:ec:cuts:prop:1}, we have $b_{ji} \in C(\tau_q)$, i.e., job $b_{ji}$ is already clairvoyant at the earliest time during $I_{d_q}$ when it is being processed. Moreover, it must hold $d_q \in N(\tau_q)$, as otherwise, that is, $d_q \in C(\tau_q)$ would contradict the definition of the algorithm and the definition of time $\tau_d$.
 
 Since we have $b_{ji} \in C(\tau_q)$, $d_q \in N(\tau_q)$ and $b_{ji}$ is processed %
 at time $\tau_q$, the definition of our algorithm implies $\frac{1-\alpha}{\alpha} y_{d_q}(\tau_q) \ge p_{b_{ji}}(\tau_q)$.
 By \Cref{lem:ec:cuts:prop:1}, we have $y_{d_q}(\tau_q) \le y_{d_q}(C_{b_{ji}}) \le \alpha \cdot p_{d_q}$. Thus, $y_{d_q}(\tau_q) \le \min\{y_{d_q}(t), \alpha \cdot p_{d_q} \} = \by_{d_q}$, and we can conclude with $\frac{1-\alpha}{\alpha} \by_{d_q} \ge \frac{1-\alpha}{\alpha} y_{d_q}(\tau_q) \ge p_{b_{ji}}(\tau_q)$
\end{proof}

Next, we show that there exists at least one job $d_q \in \Delta$ with $\frac{1-\alpha}{\alpha} \by_j \ge p_{b_{ji}}(\tau_q)$. The existence of this one job will later help us to show that the inequality actually holds for \emph{all} jobs $d_q \in \Delta$ that job~$j$ can reach via an $L_{>\sigma_{ji}}$-path.

\begin{lemma}
  \label{lem:bottleneck:inequality:1}
  Let $j \in A(t)$ be a job such that every path from $j$ to $i$ contains at least one bottleneck. Then, there is at least one job $d_q \in \Delta = \{d_1,\ldots,d_r\} = \{d \in J \mid (d,b_{ji}) \in \delta^i_{\sigma_{ji}}\}$ such that $\frac{1-\alpha}{\alpha} \by_j \ge \frac{1-\alpha}{\alpha} y_{d_q}(\tau_q) \ge p_{b_{ji}}(\tau_q)$. Furthermore, there exists a job $j' \in (A(t) \cap L^i_{\sigma_{ji}+1}) \cup H^i_{\sigma_{ji}+1}$ with $\by_{j} \ge \by_{j'}$ that can reach $d_q$ via an $E_N$-path.
\end{lemma}

\begin{proof}
  We distinguish between the two cases (1) $j \in L_{\sigma_{ji}+1}^i$ and (2) $j \in L_{\geq \sigma_{ji} + 2}^i$. Note that $j \in L_{\le \sigma_{ji}}^i$ is not possible by definition of $\sigma_{ji}$.

  \paragraph*{Case (1):} Assume $j \in L_{\sigma_{ji}+1}^i$. By definition of $\sigma_{ji}$, the job $j$ cannot reach $i$ without visiting the bottleneck $b_{ji}$ of layer $L^i_{\sigma_{ji}}$. This means that $j$ cannot reach any other entrypoint of layer $L^i_{\sigma_{ji}}$ via a $\sigma_{ji}$-entry path.  
  However, by the definition of layers, job $j$ can reach at least one entrypoint of $L^i_{\sigma_{ji}}$ via a path composed of an $E_N$-path and a single edge in $\delta^i_{\sigma{ji}}$. Since $b_{ji}$ is the only entrypoint that $j$ can reach via such paths, we must have that $j$ can reach a job $d_q \in \Delta$ with $(d_q,b_{ji}) \in \delta^i_{\sigma_{ji}}$ via an $E_N$-path.
  
  The $E_N$-path from $j$ to $d_q$ implies $\by_j \ge \by_{d_q}$ using \Cref{lem:good-path-from-active}. Then, \Cref{lem:aux:ec:edges:2} implies 
  \[
  \frac{1-\alpha}{\alpha} \by_j \ge \frac{1-\alpha}{\alpha} \by_{d_q} \ge \frac{1-\alpha}{\alpha} y_{d_q}(\tau_q) \ge p_{b_{ji}}(\tau_q) \ .
  \] 

  For the second part of the lemma, observe that $j \in  A(t) \cap L^i_{\sigma_{ji}+1}$ can reach $d_q$ via an $E_N$-path and trivially satisfies $\by_j \ge \by_j$.

  \paragraph*{Case (2):}  Assume $j \in L_{\geq \sigma_{ji} + 2}^i$.
  By \Cref{lem:aux:inequality:1}, there exists an entrypoint $k$ of layer $L^i_{\sigma_{ji}+1}$ with $\by_j \ge \by_k$. Since layer $L^i_{\sigma_{ji}}$ has a bottleneck, \Cref{obs:entrypoint:bottleneck} implies that $b_{ji}$ is the \emph{only} entrypoint of layer $L^i_{\sigma_{ji}}$ that $k$ can reach via a $\sigma_{ji}$-entry path. 
  Since the definition of layers implies that $k$ can reach at least one entrypoint via a $\sigma_{ji}$-entry path that apart from the last edge uses only edges in $E_N$, job $k$ must be able to reach some $d_q \in \Delta$ with $(d_q,b_{ji}) \in \delta^i_{\sigma_{ji}}$ via an $E_N$-path.

  The $E_N$-path from $k$ to $d_q$ implies $\by_j \ge \by_k \ge \by_{d_q}$ using \Cref{lem:good-path-from-entrypoints}. Then, \Cref{lem:aux:ec:edges:2} implies
  \[
  \frac{1-\alpha}{\alpha} \by_j \ge \frac{1-\alpha}{\alpha} \by_{k} \ge \frac{1-\alpha}{\alpha} \by_{d_q} \ge \frac{1-\alpha}{\alpha} y_{d_q}(\tau_q) \ge p_{b_{ji}}(\tau_q) \ .
  \]

  For the second part of the lemma, observe that $k \in  H^i_{\sigma_{ji}+1}$ can reach $d_q$ via an $E_N$-path and satisfies $\by_j \ge \by_k$.
\end{proof}

\Cref{lem:bottleneck:inequality:1} essentially shows that the inequality of~\Cref{lem:stealing:via:bottlenecks:2} is satisfied for \emph{one} $d_q \in \Delta$. In particular, the second part of the lemma gives us a job $j' \in (A(t) \cap L^i_{\sigma_{ji}+1}) \cup H^i_{\sigma_{ji}+1}$ with $\by_{j} \ge \by_{j'}$ that can reach $d_q$ via an $E_N$-path. If this job $j'$ can also reach a different $d_{q'} \in \Delta$ via an $E_N$-path, then~\Cref{lem:good-path-from-active,lem:good-path-from-entrypoints} imply that the inequality~\Cref{lem:stealing:via:bottlenecks:2} also holds for $d_{q'}$. The next auxiliary lemma shows that all jobs  $d_{q'} \in \Delta$ with a certain property can be reached by $j'$ via an $E_N$-path and, thus, satisfy the inequality of~\Cref{lem:stealing:via:bottlenecks:2}.

\begin{lemma}
    \label{lem:stealing:via:bottleneck:1}
    Let $j' \in (A(t) \cap L^i_{\sigma_{ji}+1}) \cup H^i_{\sigma_{ji}+1}$ be a job that can reach a job $d_q \in \Delta = \{d_1,\ldots,d_r\} = \{d \in J \mid (d,b_{ji}) \in \delta^i_{\sigma_{ji}}\}$ via an $E_N$-path. Let $d_{q'} \in \Delta$ with $d' \not= d$ be a job that is executed at some point in time $t' \ge C_{b_{ji}}$ while being non-clairvoyant. Then $j'$ can reach $d_{q'}$ via an $E_N$-path.
\end{lemma}

\begin{proof}

We distinguish between the two cases (1) $j' \in A(t) \cap L_{\sigma_{ji}+1}^i$ and (2) $j' \in H_{\sigma_{ji}+1}^i$.

\textbf{Case (1):} Assume $j' \in A(t) \cap L_{\sigma_{ji}+1}^i$. Let $P$ be the $E_N$-path from $j'$ to $d_q$ that exists by assumption of the lemma. 
Since $j' \in A(t)$, we have $[\min\{r_{j'},r_{d_q}\},t] \subseteq I(P)$ by \Cref{obs:lifetime:2}. If $d_{q'}$ is executed during $[\min\{r_{j'},r_d\},t]$ while being non-clairvoyant, then this implies the existence of an $E_N$-path from $j'$ to $d_{q'}$, and we are done. Since $(d_q,b_{ji}) \in \delta_{\sigma_{ji}}^i$, we get $r_{d_q} \le C_{b_{ji}}$ by \Cref{lem:ec:cuts:prop:1}. By assumption, $d_{q'}$ is executed after $C_{b_{ji}}$ while being non-clairvoyant. Thus, $d_{q'}$ is executed during $[\min\{r_{j'},r_{d_q}\},t]$ while being non-clairvoyant.

\textbf{Case (2):} Assume $j' \in H_{\sigma_{ji}+1}^i$ and let $P$ be an $E_N$-path from $j'$ to $d_q$. Note that such a path exists by assumption of the lemma.
Furthermore, let $P'$ denote a $(\sigma_{ji}+1)$-entry path from some $h \in A(t) \cap L^i_{\geq \sigma_{ji}+2}$ to $j'$, which must exist because $j' \in H_{\sigma_{ji}+1}^i$. Using \Cref{obs:lifetime:2}, we get $I(P) \cup I(P') \supseteq [\min\{r_{j'},r_d\},t]$.

If $d_{q'}$ is executed during $I(P)$ while being non-clairvoyant, then this implies the existence of an $E_N$-path from $j'$ to $d_{q'}$, and we are done.
As argued in Case (1), $d_{q'}$ is executed during $[\min\{r_{j'},r_{d_q}\},t]$ while being non-clairvoyant. If this execution is not during $I(P)$, then it must be during  $I(P')\setminus I(P)$. However, this implies the existence of an edge $(h',d_{d'}) \in E_N$ with $h' \in L_{\geq x+2}^i$ and $d' \in L_{x+1}^i$; a contradiction.
\end{proof}

Having established all these auxiliary lemmas, we are ready to prove~\Cref{lem:stealing:via:bottlenecks:2}.

\begin{proof}[Proof of~\Cref{lem:stealing:via:bottlenecks:2}]
    Assume that the elements of $\Delta= \{d_1,\ldots,d_r\} = \{d \in J \mid (d,b_{ji}) \in \delta^i_{\sigma_{ji}}\}$ are indexed by non-decreasing $\tau_h$, i.e., $h < h'$ implies $\tau_h \le \tau_{h'}$. In particular, this also means that $h < h'$ implies $p_{b_{ji}}(\tau_h) \ge p_{b_{ji}}(\tau_{h'})$. So if $\frac{1-\alpha}{\alpha} \by_{j} \ge p_{b_{ji}}(\tau_h)$ holds, then also $\frac{1-\alpha}{\alpha} \by_{j} \ge p_{b_{ji}}(\tau_{h'})$ for all $h' > h$.

    Let $h^*$ denote the maximum index that satisfies  $\frac{1-\alpha}{\alpha} \by_j < p_{b_{ji}}(\tau_{h^*})$. If no such index exists, then we are done as the first property of the lemma then holds for every $d_h \in \Delta$. Thus, assume in the following that $h^*$ exists.
    
    By choice of $h^*$, all $d_h$ with $h > h^*$ satisfy $\frac{1-\alpha}{\alpha} \by_{j} \ge p_{b_{ji}}(\tau_h)$. Thus, it only remains to show that~$j$ cannot reach any $d_h$ with $h \le h^*$ via an $L^i_{\sigma_{ji}+1}$-path.

    To this end, we show that the following three properties hold, and argue that they imply that~$j$ cannot reach any $d_h$ with $h \le h^*$ via an $L^i_{\ge \sigma_{ji}+1}$-path.

    \begin{enumerate}
      \item[\textbf{Property 1}.] $r_j > \tau_{h^*}$.
      \item[\textbf{Property 2}.] No $d_h \in \Delta$ with $h \le h^*$ is executed after $\tau_{h^*}$. In particular, $r_{d_h} < \tau_{h^*}$.
      \item[\textbf{Property 3}.] No $j' \in L^i_{\ge \sigma_{ji}+1} \cap A(\tau_{h^*})$ is executed after $\tau_{h^*}$. 
    \end{enumerate}

    Before we prove that the properties indeed hold, we argue that they imply that $j$ cannot reach any $d_h$ with $h \le h^*$ via an $L^i_{\ge \sigma_{ji}+1}$-path. 
    To see this, assume otherwise, i.e., that there exists an $L^i_{\ge \sigma_{ji}+1}$-path $P$ from $j$ to a $d_h \in \Delta$ with $h \le h^*$. Since $r_j > \tau_{h^*}$ (Property 1) but $d_h$ with $r_{d_h} < \tau_{h^*}$ is not executed after $\tau_{h^*}$ (Property 2), the path $P$ must contain a job $j' \in A(\tau_{h^*})$ that is executed after $\tau_{h^*}$. However, as $P$ is an $L^i_{\ge \sigma_{ji}+1}$-path, the job $j'$ would be in $L^i_{\ge \sigma_{ji}+1} \cap A(\tau_{h^*})$ and executed after $\tau^*$. 
    The existence of such a job $j'$ is a contradiction to Property 3. Hence, we can conclude that $j$ cannot reach any $d_h$ with $h \le h^*$ via an $L^i_{\ge \sigma_{ji}+1}$-path. 

    To complete the proof, it remains to show that the three properties indeed hold.

    \paragraph{Proof of Property 1.}
    By choice of $h^*$, we have $\frac{1-\alpha}{\alpha} \by_{j} < p_{b_{ji}}(\tau_{h^*})$. If we had $r_j < \tau_{h^*}$, then we would have $\tau_{h^*} \in I_j$ as $j$ is in $A(t)$. This would imply $(j,b_{ji}) \in \delta_{\sigma_{ji}}$ and, thus, $\frac{1-\alpha}{\alpha} \by_j \ge p_{b_{ji}}(\tau_j)$ for the earliest point in time $\tau_j \in I_j$ at which $b_{ji}$ is executed (cf.~\Cref{lem:aux:ec:edges:2}). Since $\tau_{h^*} \in I_j$, we have $\tau_j \le \tau_{h^*}$ and, thus, $\frac{1-\alpha}{\alpha} \by_j \ge p_{b_{ji}}(\tau_j) \ge p_{b_{ji}}(\tau_{h^*})$. However, this contradicts $\frac{1-\alpha}{\alpha} \by_{j} < p_{b_{ji}}(\tau_{h^*})$.

    \paragraph{Proof of Property 2.}
    By \Cref{lem:stealing:via:bottleneck:1,lem:bottleneck:inequality:1}, we have  $\frac{1-\alpha}{\alpha} \by_j \ge p_{b_{ji}}(\tau_{h})$ for each $d_{h} \in \Delta$ that is executed after $C_{b_{ji}}$ while being non-clairvoyant. 
    Thus, by the choice of $h^*$, each $d_{h} \in \Delta$ with $h \le h^*$ cannot be executed after $C_{b_{ji}}$ while being non-clairvoyant; otherwise we would have $\frac{1-\alpha}{\alpha} \by_j \ge p_{b_{ji}}(\tau_{h})$ for an $h \le h^*$ which contradicts the definition of $h^*$. Furthermore, by~\Cref{lem:ec:cuts:prop:1}, each $d_{h} \in \Delta$ is non-clairvoyant at $C_{b_{ji}}$, i.e., $d_{h} \in N(C_{b_{ji}})$. This implies that each $d_{h} \in \Delta$ with $h \le h^*$ cannot be executed after $C_{b_{ji}}$ at all, because such an execution after $C_{b_{ji}}$ would be non-clairvoyant.
    
    \paragraph{Proof of Property 3.} Let $j'$ be a job in $L_{\ge \sigma_{ji}+1} \cap A(\tau_{h^*})$. Since $j'$ is alive at point in time $\tau_{h^*}$, we have $(j', b_{ji}) \in E_C$ because $b_{ji}$ is executed at $\tau_{h^*}$ while being clairvoyant. Using $j' \in L_{\ge \sigma_{ji}+1}$, this implies $(j', b_{ji}) \in \delta_{\sigma_{ji}}$. This allows us to apply~\Cref{lem:ec:cuts:prop:1} to obtain $j' \in N(C_{b_{ji}})$. Furthermore, \Cref{obs:clairvoyant-jobs-block-earlier-jobs} implies that $j'$ is not executed during $[\tau_{h^*},C_{b_{ji}}]$. Thus, if $j'$ is executed after $\tau_{h^*}$, then this execution must take place at a point in time $t'$ with $t' \ge C_{b_{ji}}$ and $j' \in N(t')$. However, then \Cref{lem:stealing:via:bottleneck:1,lem:bottleneck:inequality:1} imply $\frac{1-\alpha}{\alpha} \by_j \ge p_{b_{ji}}(\tau_{j'})$ for the earliest point in time $\tau_{j'} \in I_{j'}$ at which $b_{ji}$ is executed. Since $\tau_{h^*} \in I_{j'}$, this implies $\tau_{j'} \le \tau_{h^*}$ and, thus  $\frac{1-\alpha}{\alpha} \by_j \ge p_{b_{ji}}(\tau_{j'}) \ge p_{b_{ji}}(\tau_{h^*})$. This is a contradiction to the choice of $h^*$ according to which we have $\frac{1-\alpha}{\alpha} \by_j < p_{b_{ji}}(\tau_{h^*})$. Thus, we can conclude that $j'$ cannot be executed after $C_{b_{ji}}$.
\end{proof}

\section{Clairvoyant Jobs and the Proof of \Cref{thm:local:clairvoyant}}\label{sec:clairvoyant}

For every job $j$, let $t_j$ be the latest point in time before $t$ at which job $j$ is being processed.

\begin{definition}\label{def:clairvoyant:sets}
    For every job $i \in J$ with $r_i \leq t$, we define 
    \begin{itemize}
        \item 
        $H_i \coloneq \{j \in C(t) \setminus \{i\} \mid t_j \leq r_i \text{ and } i \in R_j \}$ and
        \item 
        $L_i \coloneq \{j \in C(t) \setminus \{i\} \mid t_j > r_i \text{ and } i \in R_j \}$.
    \end{itemize}
\end{definition}

This definition immediately implies the following proposition. 

\begin{proposition}\label{lemma:possible-borrows}
    For every $j \in C(t)$, it holds that $R_j \setminus \{j\} = \{ i \in J \mid j \in H_i \cup L_i \}$.
\end{proposition}

\begin{lemma}\label{lemma:partial-rem-lb}
    The following properties hold:
    \begin{enumerate}[(a)]
        \item For every $i \in N(t)$ and for every $j \in H_i$, it holds that $(1-\alpha) y_i(t) \leq p_j(t)$.
        \item For every $i \in J \setminus N(t)$ and for every $j \in H_i$, it holds that $(1-\alpha)p_i \leq p_j(t)$.
    \end{enumerate}
\end{lemma}

\begin{proof}
    For (a), let $i \in N(t)$ and $j \in H_i$. 
    Since $j \in H_i$, we have that $t_j \leq r_i \leq t_i$ and $j \in C(t_i)$, where the latter holds as $j$ is clairvoyant at point in time $t$, and thus, must have become clairvoyant the latest at time $t_j \le t_i$.
    The algorithm guarantees that at time $t_i$ it holds that $\frac{1-\alpha}{\alpha} y_i(t_i) < p_j(t_i)$ as $j \in C(t_i)$ and $i \in N(t_i)$. Therefore, $(1-\alpha) y_i(t) = (1-\alpha) y_i(t_i^+) \leq \frac{1-\alpha}{\alpha} y_i(t_i^+) \leq p_j(t_i) = p_j(t)$. 

    For (b), let $i \in J \setminus N(t)$ and $j \in H_i$. Let $s_i$ be the time when $i$ emits.
    Since $j \in H_i$, we have that $t_j \leq r_i \leq s_i$. Furthermore, we also have $j \in C(s_i)$ because $i \in J \setminus N(t)$ and $t_j \le r_i$ imply that $t_j \le s_i \le t \le C_j$ so $j$ is still clairvoyant when job $i$ emits the signal.
    Since at time $s_i$, job $i$ is being processed, it must hold that $\frac{1-\alpha}{\alpha} y_i(s_i) < p_j(s_i)$ by the definition of the algorithm.
    Thus, $(1-\alpha)p_i = \frac{1-\alpha}{\alpha} y_i(s_i^+) \leq p_j(s_i^+) = p_j(t)$.
\end{proof}

The following statement, which we will use multiple times in the following, is an immediate corollary of \Cref{obs:clairvoyant-jobs-block-earlier-jobs}.

\begin{corollary}\label{lemma:partial-clairvoyant-separation}
  For every $j \in C(t)$, there is no job $i$ with $r_i < t_j$ being processed during $[t_j, t]$.
\end{corollary}

\begin{lemma}\label{lemma:set-L-small}
    For every job $i$, it holds that $|L_i| \leq 1$.
\end{lemma}

\begin{proof}
    For sake of contradiction, suppose that $|L_i| \geq 2$. Let $j_1,j_2 \in L_i$ with $j_1 \neq j_2$ and $t_{j_1} < t_{j_2}$. Note that $t_{j_1} = t_{j_2}$ is not possible as $j_1$ and $j_2$ are clairvoyant at time $t$. 
    We distinguish two cases.

    \textbf{Case 1:} Job $i$ is not being processed in $[t_{j_1},t]$. In particular, this implies $t_i < t_{j_1}$.
    We again distinguish two cases. %

    \textbf{Case 1a:} $t_{j_1} \leq r_{j_2}$. Since $i \in R_{j_2}$ and $t_i < t_{j_1} \leq r_{j_2}$, we conclude $t_{j_1}$ is in the interior of $I(R_{j_2})$. By~\Cref{obs:flow:cut-closure}, this implies $j_1 \in R_{j_2}$. Thus, there must be a path $P$ in the borrow graph from $j_2$ to $j_1$. Since $t_{j_1} \leq r_{j_2}$, we have $t_{j_1} < r_{j_2}$ as $j_1 \in C(t_{j_1})$ holds by definition of $L_i$. The fact that $t_{j_1} < r_{j_2}$ implies that the path $P$ must contain a job $j'$ with $r_{j'} < t_{j_1}$ and $t_{j'} > t_{j_1}$; otherwise $P$ would not be a $j_2$-$j_1$-path.
    However, this is a contradiction to \Cref{lemma:partial-clairvoyant-separation}.
    
    \textbf{Case 1b:} $r_{j_2} < t_{j_1}$.
    In this case, we again have a contradiction to \Cref{lemma:partial-clairvoyant-separation} because $j_2$ is being processed after $t_{j_1}$.

    \textbf{Case 2:} Job $i$ is being processed in $[t_{j_1},t]$. This case is also an immediate contradiction via \Cref{lemma:partial-clairvoyant-separation}, because it requires $r_i < t_{j_1}$.
\end{proof}

\begin{proof}[Proof of \Cref{thm:local:clairvoyant}]
  Using \Cref{thm:main:borrowing}, we have
\begin{align*}
  |C(t) \setminus O(t)| = \sum_{j \in C(t) \setminus O(t)} \frac{p_j(t)}{p_j(t)} 
  &= \sum_{j \in C(t) \setminus O(t)}  \sum_{i \in O(t) \cap R_j} \frac{\beta(j,i)}{p_j(t)} \\
  &= \sum_{j \in C(t) \setminus O(t)}  \sum_{\substack{i \in O(t) \cap R_j \\ \text{s.t.~}j \in H_i}} \frac{\beta(j,i)}{p_j(t)} + \sum_{j \in C(t) \setminus O(t)}  \sum_{\substack{i \in O(t) \cap R_j \\ \text{s.t.~} j \in L_i}} \frac{\beta(j,i)}{p_j(t)} \ .
\end{align*}

The final equality holds because for every $j \in C(t) \setminus O(t)$, we have $j \notin O(t) \cap R_j$, and thus, \Cref{lemma:possible-borrows} implies that either $j \in H_i$ or $j \in L_i$.
Now, note that for every $j \in C(t) \setminus O(t)$ and $i \in O(t)$ we have $\beta(j,i) \leq p_j(t)$. Moreover, \Cref{lemma:set-L-small} gives that  
and $|L_i| \leq 1$, hence we conclude that 
\[
  \sum_{j \in C(t) \setminus O(t)}  \sum_{\substack{i \in O(t) \cap R_j \\ \text{s.t. }j \in L_i}} \frac{\beta(j,i)}{p_j(t)} 
  \leq \sum_{j \in C(t) \setminus O(t)}  \sum_{\substack{i \in O(t) \cap R_j \\ \text{s.t. } j \in L_i}} 1 
  \leq \sum_{i \in O(t)} \sum_{\substack{j \in C(t) \setminus O(t) \\ \text{s.t. } j \in L_i}} 1 
  \leq |O(t)| \ .
\]

For the other part, \Cref{lemma:partial-rem-lb} yields
\begin{align*}
  \sum_{j \in C(t) \setminus O(t)}  \sum_{\substack{i \in O(t) \cap R_j \\ \text{s.t. } j \in H_i}} \frac{\beta(j,i)}{p_j(t)} 
  &\leq \sum_{j \in C(t) \setminus O(t)} \bigg( \bigg(  \sum_{\substack{i \in O(t) \cap N(t) \\ \text{s.t. } j \in H_i}} \frac{\beta(j,i)}{p_j(t)} \bigg) +  \bigg(  \sum_{\substack{i \in O(t) \setminus N(t) \\ \text{s.t. } j \in H_i}} \frac{\beta(j,i)}{p_j(t)} \bigg) \bigg) \\
  &\leq \frac{1}{1-\alpha} \sum_{j \in C(t) \setminus O(t)} \bigg( \bigg(  \sum_{\substack{i \in O(t) \cap N(t) \\ \text{s.t. } j \in H_i}} \frac{\beta(j,i)}{y_i(t)} \bigg) +  \bigg(  \sum_{\substack{i \in O(t) \setminus N(t) \\ \text{s.t. } j \in H_i}} \frac{\beta(j,i)}{p_i} \bigg) \bigg) \\
  &\leq \frac{1}{1-\alpha}  \bigg(  \sum_{i \in O(t) \cap N(t)} \sum_{j \in C(t) \setminus O(t)} \frac{\beta(j,i)}{y_i(t)}  +  \sum_{i \in O(t) \setminus N(t)} \sum_{j \in C(t) \setminus O(t)} \frac{\beta(j,i)}{p_i} \bigg) \ .
\end{align*}
Since $y_i(t) = p_i$ for every $i \in O(t) \setminus N(t)$, \Cref{thm:main:borrowing} implies that the above is at most
\begin{align*}
\frac{1}{1-\alpha}  \sum_{i \in O(t) \cap N(t)} 1  +  \frac{1}{1-\alpha} \sum_{i \in O(t) \setminus N(t)} 1  \leq \frac{1}{1-\alpha} |O(t)| \ .
\end{align*}

This concludes the proof of the theorem.
\end{proof}

\section{Lower Bounds}
\label{app:lowerbounds}

For a fixed algorithm and a fixed instance, we denote by $\delta(t)$ and $\delta(t,1)$ the number of unfinished jobs and the number of unfinished jobs with a remaining  processing time of at least $1$ at time $t$ in the algorithm's schedule, respectively, and by $\delta^*(t)$ the number of unfinished jobs at time $t$ in an optimal solution. The following technique is well-known in literature~\cite{MotwaniPT94,AzarLT21,AzarLT22}.

\begin{proposition}[Denial of Service]\label{prop:dos}
    If there exists an instance such that $\delta(t,1) \geq \rho \cdot \delta^*(t)$ for some constant $\rho$ at some time $t$, then the algorithm has a competitive ratio of at least $\Omega(\rho)$. Moreover, if for every sufficiently large integer $k$ there exists an instance such that $\delta^*(t) = k$ and $\delta(t,1) \geq \rho \cdot \delta^*(t)$ for some constant $\rho$ at some time $t$, then the algorithm has a competitive ratio of at least $\rho$.
\end{proposition}

\begin{proof}
  After time $t$, the adversary releases at each of the next $M$ integer times a job of length $1$. Thus, the best strategy for the algorithm is to run such jobs to completion before any other job with remaining processing time at least $1$ at time $t$. 
  Therefore, the algorithm has between time $t$ and $t+M$ at least $\delta(t) + 1$ unfinished jobs. Similarly, we can assume that an optimal solution has at most $\delta^*(t) + 1$ pending jobs during this time. Therefore, as $M$ tends to $\infty$, the competitive ratio of the algorithm is at least $\frac{\delta(t) + 1}{\delta^*(t) + 1} \geq \frac{\rho}{2}$, proving the first part of the statement.

  For the second part, we can analogously conclude that the algorithm's competitive ratio for every instance of the family is at least $\frac{\delta(t) + 1}{\delta^*(t) + 1} \geq \frac{\rho k + 1}{k + 1}$, which tends to $\rho$ as $k \to \infty$. 
\end{proof}

\subsection{Deterministic Lower Bound}

\begin{theorem}
For every $\alpha \in (0,1)$, every $\alpha$-clairvoyant deterministic algorithm has a competitive ratio of at least $\Omega(\frac{1}{1-\alpha})$ for minimizing the total flow time on a single machine.    
\end{theorem}

\begin{proof}
    Fix any deterministic algorithm. Initially, there are $k$ jobs $J_1=\{1,\ldots,k\}$ released at time $0$ for some large constant $k$. We ensure that no job emits its signal in the algorithm's schedule until time $k$. Thus, the algorithm has $k$ unfinished jobs at time $k$.
    For every $i \in J_1$, let $y_i$ be the total work to job $i$ in the algorithm's schedule until time $k$, and assume w.l.o.g.\ that $y_1 \leq \ldots \leq y_k$. 
    We set for every job $i \in J_1$ its processing time to $p_i =  \frac{1}{\alpha}y_i + \frac{1}{k}$ and observe that $y_i < \alpha p_i$.
    Note that every unfinished job at time $k$ has a constant remaining processing time, and $\delta(k,1) \geq k$ by scaling.

    We now present an alternative strategy until time $k$.
    Let $\ell = \floor{\alpha k} - 1$. 
    We distinguish two cases.
    If $y_{\ell+1} \leq 1$, we can easily conclude that $\sum_{i=1}^{\ell} y_i \leq \ell$. Thus, we can complete jobs $1,\ldots,\ell$ until time $k$, because 
    \[
        \sum_{i=1}^\ell p_i =  \sum_{i=1}^\ell \left( \frac{1}{\alpha} y_i + \frac{1}{k} \right) \leq \frac{1}{\alpha}\ell + 1 \leq k - \frac{1}{\alpha} + 1 \leq k \ .
    \]
    Otherwise, that is, $y_{\ell+1} \geq 1$, our assumption on the jobs indices and the choice of $\ell$ gives that $\sum_{i=1}^{\ell+1} y_i \leq \alpha k$.
    We can again complete jobs $1,\ldots,\ell$ until time $k$, because
    \[
        \sum_{i=1}^\ell p_i =  \sum_{i=1}^\ell \left( \frac{1}{\alpha} y_i + \frac{1}{k} \right) \leq k - \frac{1}{\alpha}y_{\ell+1} + 1 \leq k \ .   
    \]
    Therefore, in any case, there are at most $\delta^*(k) \leq k - \ell = \cO((1-\alpha)k)$ jobs unfinished at time $k$ in this strategy.
    Then, the lemma follows via \Cref{prop:dos}.
\end{proof}

\begin{theorem}
	\label{thm:lb:2}
    For every $\alpha \in (0,1)$, every $\alpha$-clairvoyant deterministic algorithm has a competitive ratio of at least $2$.
\end{theorem}

\begin{proof}
    Let $k$ be an integer and $\lambda = \frac{4 + \alpha}{\alpha}$. Fix any deterministic algorithm. We construct an instance of $k$ consecutive phases indexed by $k,\ldots,1$, where phase $i$ has length $\lambda^i$. Let $t$ denote the end of the last phase.

    At the start of every phase $i$ at time $s_i$, the adversary releases two jobs $j^i_1$ and $j^i_2$. Then, it waits until time $t_i = s_i + \alpha \lambda^i$ while ensuring that none of both jobs emits its signal. Assume w.l.o.g.\ that $y_{j_{1}^i}(t_i) \geq y_{j_{2}^i}(t_i)$. Thus, $y_{j_{1}^i}(t_i) \leq \alpha \lambda^i$ and $y_{j_{2}^i}(t_i) \leq \frac{\alpha}{2}\lambda^i$. At time $t_i$ the adversary sets $p_{j_{1}^i} = 2\lambda^i$ and $p_{j_{2}^i} = \lambda^i$. Since $\alpha p_{j_{1}^i} = 2 \alpha \lambda^i > y_{j_{1}^i}(t_i)$ and  $\alpha p_{j_{2}^i} = \alpha \lambda^i > y_{j_{2}^i}(t_i)$, this choice is consistent with our assumption that none of both jobs emitted its signal at time $t_i$.
    During phase~$i$, job $j_{1}^i$ can be processed by at most $\lambda^i$ units and job $j_{2}^i$ can be processed by at most $\lambda^i - \frac{\alpha}{2}\lambda^i$ units.
    Thus, at the end of phase $i$ (at time $s_{i-1}$) the remaining processing time of each of $j_{1}^i$ and $j_{2}^i$ is at least $\frac{\alpha}{2}\lambda^i$.
    Now observe that after phase $i$, no algorithm can finish any of $j_{1}^i$ and $j_{2}^i$ until time $t$, and moreover each of $j_{1}^i$ and $j_{2}^i$ has a remaining processing time of at least $\frac{\alpha}{4}\lambda^i = \frac{\alpha}{2}\lambda^i - \frac{\alpha}{4}\lambda^i$ at time $t$.
    To see this, note that
    \[
        t - s_{i-1} \leq \sum_{i'=0}^{\infty} \lambda^{i-1-i'} = \lambda^{i-1} \sum_{i'=0}^{\infty} \lambda^{-i'} = \lambda^{i-1} \cdot \frac{1}{1 - \lambda^{-1}} = \lambda^{i-1} \cdot \frac{4+\alpha}{4} = \frac{\alpha}{4}\lambda^i \ . 
    \]

    Therefore, the algorithm has $2k$ unfinished jobs at time $t$, each with a remaining processing time of at least $\frac{\alpha}{4}\lambda^1 > 1$. Thus, $\delta(t,1) = 2k$. An optimal solution has only $\delta^*(t) = k$ unfinished jobs at time $t$, as it can complete in every phase $i$ job $j_2^i$.
    Then, the lemma follows via \Cref{prop:dos}.
\end{proof}

\subsection{Randomized Lower Bound}

\begin{theorem}\label{theorem:randomized-lb}
    For every $\alpha \in (\frac{1}{2},1)$, every $\alpha$-clairvoyant randomized algorithm has a competitive ratio of at least $\Omega(\frac{1}{1-\alpha})$ for minimizing the total flow time on a single machine.    
\end{theorem}

We present a randomized instance on which every deterministic algorithm has an expected competitive ratio of at least $\Omega(\frac{1}{1-\alpha})$. The theorem then follows using Yao's principle. 
    
Fix any deterministic algorithm. We consider a distribution of instances where at time $0$ we release $k =  \floor{2^{1/(2-2\alpha)}}$ jobs with integer processing times $p_j = y_j + 1$ with $y_j$'s drawn independently from the geometric distribution with mean $2$. Note that $\EX[p_j] = 3$. Let $t = \floor{3(k-k^{3/4})}$ and let $L$ denote the event that for every job $j$ it holds that $p_j \leq \frac{1}{1-\alpha}$.

We first show that conditioned on $L$, the algorithm has many unfinished jobs at time $t$ in expectation.

\begin{lemma}\label{lemma:randomized-lb-alg}
    It holds that $\EX[\delta(t,1) \mid L] \geq \Omega(k^{3/4})$. 
\end{lemma}

\begin{proof}
To prove this statement, we make the following two assumptions w.l.o.g.:
\begin{itemize}
    \item Every job emits its signal after $y_j$ units have been processed. This does not weaken the algorithm, because $y_j = p_j - 1 \leq \alpha p_j$ whenever event $L$ occurs. Moreover, whenever a job emits, we assume that in the next time unit the algorithm also completes it by processing the final unit of processing requirement. Note that this assumption cannot increase $\delta(t)$.
    \item All preemptions until time $t$ occur at integer times. Since $t$ is an integer, all processing times are integers and, by the previous assumption, all signals emit after integer processing amounts, we can convert any schedule with fractional time allocations to an integral schedule without increasing $\delta(t)$.
\end{itemize}

Now consider an integer time unit $t'$ between $1$ and $t$. If at time $t'$ a non-clairvoyant job $j$ is being processed, observe that $j$ emits its signal with probability $\frac{1}{2}$ if the job $j'$ that ran at time $t'-1$
did not emit its signal, because otherwise, the algorithm would finish $j'$ at time $t'$.
This is in particular independent of the choice of $j$, because the algorithm cannot distinguish between non-clairvoyant jobs. 
Therefore, we have
\[
\mu_{t'} := \pr \left[\text{ the job that runs at time } t' \text{ emits } \right] =  \frac{1}{2}\left( 1 - \mu_{t'-1} \right) \ ,  
\]
and $\mu_0 := \frac{1}{2}$.

Thus, the expected number of jobs that have emitted their signal until time $t$ is at most
\begin{align*}    
    \sum_{t'=0}^{t-1} \mu_{t'} &= \sum_{t'=0}^{t-1} \sum_{\ell = 0}^{t'} (-1)^\ell \left(\frac{1}{2}\right)^{\ell+1}
    = \sum_{t'=0}^{t-1} \frac{1}{2} \cdot \frac{1 - (-1/2)^{t'+1}}{1-(-1/2)} \\
    &= \sum_{t'=0}^{t-1} \frac{1}{3} + \frac{1}{6} \left(-\frac{1}{2}\right)^{t'} 
    \leq \frac{1}{3}t +  \frac{1}{9} - \frac{1}{9} \left(-\frac{1}{2} \right)^t \leq \frac{1}{3}t + 1 = k - k^{3/4} + 1 \ . 
\end{align*}

Since every job which has not emitted until time $t$ has a remaining processing time of at least~$1$, we conclude that $\EX[\delta(t,1) \mid L] \geq k - (k - k^{3/4} + 1) = \Omega(k^{3/4})$.
\end{proof}

\begin{lemma}\label{lemma:randomized-lb-opt}
    It holds that $\EX[\delta^*(t)] \leq \cO(\frac{k^{3/4}}{\log k})$.
\end{lemma}

\begin{proof}
First, note that the total processing time $P$ of all jobs is a sum of $k$ independent random variables with mean $3k$ and variance $2k$. Thus, Chebyshev's inequality gives 
\[
    \pr \left[P \geq 3k + k^{3/4} \right] = \pr \left[ P \geq 3k + k^{1/2} \cdot k^{1/4} \right] \leq \cO \left( \frac{1}{k^{1/2}} \right) \ .
\]

Moreover, let $b = \frac{\log k}{4}$ and let $B$ be the number of jobs with processing time more than $b$. 
Note that $B$ is binomially distributed, because it can be written as the sum of $k$ binary random variables $Y_j = \ind[p_j > b]$ with equal success probabilities $\pr[p_j > b] = \pr[y_j > b - 1] = (\frac{1}{2})^{b-1} = 2k^{-1/4}$.
Therefore, $\EX[B] = k^{-1/4} \cdot k = k^{3/4}$ and $\VAR[B] = k^{-1/4} (1-k^{-1/4}) k = O(k^{3/4})$.
Thus, Chebyshev's inequality gives 
\[
    \pr \left[ B \leq k^{3/4} \right] = \pr \left[B \leq 2 k^{3/8} \cdot k^{3/8} - k^{3/4} \right] \leq \cO \left( \frac{1}{k^{3/4}} \right) \ .
\]
    
Therefore, with probability of at least $1-\cO(\frac{1}{k^{1/2}})$ it holds that $P < 3k + k^{3/4}$ and $B > k^{3/4}$.
In this case, we compute the maximum number of jobs of size at least $b$ that an optimum solution cannot complete until time $t$ as follows
\begin{align*}    
    \frac{1}{b} \left(P - t \right) \leq \frac{1}{b} \left(3k + k^{3/4} - 3(k-k^{3/4}) + 1 \right) 
    \leq \cO\left( \frac{k^{3/4}}{\log k} \right) \ . 
\end{align*}
Note that $\cO(\frac{k^{3/4}}{\log k}) \leq k^{3/4} \leq B$, i.e., there are such many long jobs.

In total, the expected number of alive jobs in an optimum solution at time $t$ is at most
\[
    \EX[\delta^*(t)] \leq \left(1-\cO\left(\frac{1}{k^{1/2}} \right) \right) \cdot \cO\left( \frac{k^{3/4}}{\log k} \right) + \cO \left( \frac{1}{k^{1/2}} \right) \cdot k = \cO\left( \frac{k^{3/4}}{\log k} \right) \ ,
\]
which concludes the proof of the lemma.
\end{proof}

Let $\bL$ denote the complementary event to $L$.

\begin{proposition}\label{prop:compl-prob}
    It holds that $\pr[\bL] \leq \cO(\frac{1}{k})$.
\end{proposition}

\begin{proof}
    Event $\bL$ occurs if for at least one job $j$ it holds that $p_j > \frac{1}{1-\alpha}$. Since this happens for job $j$ with probability $\pr[y_j > \frac{1}{1-\alpha} - 1] = 2 \cdot 2^{- 1 / (1-\alpha)}$, we can apply the union bound and conclude that
    \[
        \pr(\bL) \leq k \cdot 2 \cdot 2^{-\frac{1}{1-\alpha}} \leq 2 \cdot 2^{\frac{1}{2(1-\alpha)}} \cdot 2^{-\frac{1}{1-\alpha}} = \cO\left(\frac{1}{k} \right) \ ,
    \]
    which implies the stated bound.
\end{proof}

Finally, we can prove \Cref{theorem:randomized-lb}.

\begin{proof}[Proof of \Cref{theorem:randomized-lb}]
    \Cref{lemma:randomized-lb-alg,lemma:randomized-lb-opt} and the law of total probability imply
    \begin{align*}
        \frac{\EX[\delta(t,1)]}{\EX[\delta^*(t)]} = \frac{\pr[L] \cdot \EX[\delta(t,1) \mid L]}{\pr[L] \cdot \EX[\delta^*(t) \mid L]} &\geq \frac{\EX[\delta(t,1) \mid L] - \pr[\bL] \cdot \EX[\delta(t,1) \mid \bL]}{\EX[\delta^*(t)]} \\
        &\geq \Omega(\log k) - \frac{\pr[\bL] \cdot \EX[\delta(t,1) \mid \bL]}{\EX[\delta^*(t)]} \ .
    \end{align*}
    
    Now, observe that $\EX[\delta(t,1) \mid \bL] \leq k$ and $\EX[\delta^*(t)] \geq \frac{1}{2}$. The latter holds because with probability at least $\frac{1}{2}$, the total volume $P$ of the $k$ jobs is more than $t$, in which case no solution can finish all jobs. This can be seen via Chebyshev's inequality because $\pr[P \leq t] \leq \frac{1}{2k^{1/2}} \leq \frac{1}{2}$. Thus, using \Cref{prop:compl-prob}, we conclude that the above is at least
    \[
        \Omega(\log k) - \frac{\cO(1/k) \cdot k}{1/2} = \Omega(\log k) = \Omega \left(\frac{1}{1-\alpha} \right) \ .
    \]
    The theorem now follows by applying \Cref{prop:dos} to every realization of our distribution and finally using Yao's principle.
\end{proof}

\begin{theorem}
	For every $\alpha \in (0,1)$, every $\alpha$-clairvoyant randomized algorithm has a competitive ratio of at least $\frac{3}{2}$.
\end{theorem}

\begin{proof}
	Exploiting Yao's Theorem, we show the statement by giving a randomized instance for which each deterministic algorithm has an expected competitive ratio of at least $\frac{3}{2}$.
	
	The randomized instance is based on the deterministic lower bound of Theorem~\ref{thm:lb:2}. We randomize the instance by, at the beginning of each phase $i$, drawing one of the jobs $\{j_1^i,j_2^i\}$ uniformly at random. The drawn job gets assigned the processing time of $\lambda_i$ while the other job gets a processing time of $2\lambda_i$. Thus, each job in $\{j_1^i,j_2^i\}$ has processing time $\lambda_i$ with probability $\frac{1}{2}$ and processing time $2\lambda_i$ with probability $\frac{1}{2}$. The rest of the instance remains deterministic as defined in the proof of Theorem~\ref{thm:lb:2}.
	
	Consider a fixed deterministic algorithm. At any point in time during the first $\alpha\lambda_i$ time units of a phase $i$, the probability that the algorithm executes the shorter job of $\{j_1^i,j_2^i\}$ is at most $\frac{1}{2}$. This is because none of the jobs can emit during the first $\alpha\lambda_i$ time units of the phase, so the algorithm cannot distinguish the jobs.
	Thus, the probability that the shorter job is processed for at most $\frac{\alpha}{2}\lambda_i$ time units within the first $\alpha\lambda_i$ time units of phase $i$ is a least $\frac{1}{2}$. Thus, with probability at least $\frac{1}{2}$, both jobs in $\{j_1^i,j_2^i\}$ have a remaining processing time of at least $\frac{\alpha}{2} \lambda_i$ at the end of phase $i$. As argued in the proof of Theorem~\ref{thm:lb:2}, each job that has a remaining processing time of at least $\frac{\alpha}{2} \lambda_i$ at the end of its phase has a remaining processing time of at least $\frac{\alpha}{4} \lambda_i$ at point in time $t$. 
	Thus, for each phase $i$, the probability that both jobs $j \in \{j_1^i,j_2^i\}$ satisfy $p_j(t) \ge \frac{\alpha}{4} \lambda_1 > 1$ is at least $\frac{1}{2}$. With the remaining probability, at least the longer job of $\{j_1^i,j_2^i\}$ has a remaining processing time of at least $1$ at time $t$.
	Thus, the expected number of jobs released in phase $i$ that have a remaining processing time of at least $1$ at time $t$ is at least $\frac{3}{2}$.
	This implies that the expected number of jobs released during the $k$ phases that have a remaining processing time of at least $1$ at time $t$ is at least~$\frac{3}{2}k$.
	
	At the same time, the same clairvoyant optimal solution as in the proof of Theorem~\ref{thm:lb:2} has only $k$ jobs alive at time $t$.
	By~\Cref{prop:dos}, this implies an expected competitive ratio of at least $\frac{3}{2}$, and thus, the statement.
\end{proof}

\end{document}